%% file: ms.tex
\documentclass[preprint]{elsarticle}

\usepackage{graphicx}
\usepackage{listings} 
\usepackage{amsmath}
\usepackage{listings} 
\usepackage{amsthm}
\usepackage{xspace}
\usepackage{latexsym,amssymb,stmaryrd}
\usepackage{color}
\usepackage{tikz}
\usepackage{prooftree}
\usepackage{proof}
\usetikzlibrary{calc}
\usetikzlibrary{shapes,snakes}
\usepackage{float}
\restylefloat{figure}
\usepackage{wrapfig,subfig}
\input{Def}
\begin{document}\newif\ifsubmit
%\submittrue
\submitfalse
\ifsubmit
\newcommand{\EZ}[1]{#1}
\newcommand{\EZComm}[1]{}
\newcommand{\PG}[1]{#1}
\newcommand{\PGComm}[1]{}
\newcommand{\MS}[1]{#1}
\newcommand{\MSComm}[1]{}
\newcommand{\JC}[1]{#1}

\newcommand{\JCComm}[1]{}
\else
\newcommand{\EZ}[1]{\textcolor{blue}{#1}}

\newcommand{\EZComm}[1]{{\scriptsize \textcolor{blue}{[Elena{:} #1]}}}

\newcommand{\PG}[1]{\textcolor{blue}{#1}}

\newcommand{\PGComm}[1]{{\scriptsize \textcolor{magenta}{[Paola{:} #1]}}}
\newcommand{\MS}[1]{\textcolor{green}{#1}}
\newcommand{\MSComm}[1]{{\scriptsize \textcolor{green}
{[MS{:} #1]}}}
\newcommand{\JC}[1]{\textcolor{red}{#1}}
\newcommand{\JCComm}[1]{{\scriptsize \textcolor{red}
{[JC{:} #1]}}}

\fi

\begin{frontmatter}

\title{{Flexible recovery of uniqueness and immutability\\
(Extended Version)}}

\author[pg]{Paola Giannini\corref{cor1}}
\ead{giannini@di.unipmn.it}

\author[ms]{Marco Servetto}
\ead{marco.servetto@ecs.vuw.ac.nz}

\author[ez]{Elena Zucca}
\ead{elena.zucca@unige.it}

\author[ms]{James Cone}
\ead{james.cone@ecs.vuw.ac.nz}

\cortext[cor1]{Principal corresponding author}

\address[ms]{School of Engineering and Computer Science, Victoria University of Wellington, Gate 6, Kelburn Parade
Wellington, New Zealand}
\address[pg]{Computer Science Institute, DiSIT, Universit\`a del Piemonte Orientale, viale Teresa Michel 11,  Alessandria, Italy\footnote{This original research has the financial support of the Universit\`a  del Piemonte Orientale.}
}
\address[ez]{DIBRIS, Universit\`a di Genova, Via Dodecaneso 35, Genova, Italy}
 
\begin{abstract}
We present an imperative object calculus where types are annotated with qualifiers for aliasing {and mutation} control.  There are two key novelties with respect to similar proposals. {First, the type system is very expressive. Notably, it adopts the \emph{recovery} approach, that is, using the
type context to justify strengthening types, greatly improving its power by permitting to recover uniqueness and immutability properties even in presence of other references. This is achieved by rules which restrict the use of such other references in the portion of code which is recovered.  Second, execution is modeled by a non standard operational model, where} 
properties of qualifiers can be directly expressed on source terms, rather than as invariants on an auxiliary structure which mimics physical memory. Formally, this is achieved by the block construct, introducing local variable declarations, which, when evaluated, play the role of store.
\end{abstract}

\end{frontmatter}

\input{Intro}
\input{Informal}
\input{TypeSystem}

\input{Examples}
\input{Calculus}

\input{Results}

\input{RelatedWork}

\input{Conclu}

\section*{Bibliography}
\bibliographystyle{elsart-num-sort}
\bibliography{main}

\input{Appendix}

\end{document}

%% file: Def.tex
\theoremstyle{plain}

\newtheorem*{theorem*}{Theorem}
\newtheorem{theorem}{Theorem}
\newtheorem{lemma}[theorem]{{\bf Lemma}}
\newtheorem{definition}[theorem]{{\bf Definition}}
\newtheorem{proposition}[theorem]{{\bf Proposition}}

\renewenvironment{proof}[1][Proof]{\begin{trivlist}
\item[\hskip \labelsep {\bfseries #1}]}{\end{trivlist}}

\newcommand*{\ttfamilywithbold}{\fontfamily{pcr}\selectfont}
\definecolor{darkRed}{RGB}{100,0,10}
\definecolor{darkBlue}{RGB}{10,0,100}
\newcommand{\Q}{\lstinline}

\newcommand{\refToFigure}[1]{Figure~\ref{fig:#1}}
\newcommand{\refToSection}[1]{Section~\ref{sect:#1}}
\newcommand{\refToLemma}[1]{Lemma~\ref{lemma:#1}}
\newcommand{\refToTheorem}[1]{Theorem~\ref{theo:#1}}

\newcommand{\Space}{\hskip 0.8em}
\newcommand{\HSep}{\hbox to \textwidth{\bf\hrulefill}}

\newcommand{\Tuple}[1]    {\langle{#1}\rangle}

\newcommand{\FourTuple}[4]     {\Tuple{{#1},{#2},{#3},{#4}}}
\newcommand{\SubstFun}[2]{#1[#2]}

\newcommand{\Rule}[3]{\displaystyle                  
\frac{#1}{#2}        %  #1 = premesse (modo math) %  #2 = conseguenza (modo math)
{#3}     %  #3 = side conditions (modo math)
}

\newcommand{\NamedRule}[4]{\scriptstyle{\textsc{(#1)}}
\displaystyle                  %  #1 = nome regola
\frac{%
\begin{array}{l}#2\end{array}%
}{#3}        %  #2 = premesse (modo math) %  #3 = conseguenza (modo math)
\begin{array}{l}#4\end{array}     %  #4 = side conditions (modo math)
}

\newcommand{\NamedRuleOL}[3]{\scriptstyle{\textsc{(#1)}}
\displaystyle                  %  #1 = nome regola
{%
\begin{array}{l}#2\end{array}}\           
\begin{array}{l}#3\end{array}     %  #3 = side conditions (modo math)
}

\newcommand{\rn}[1]{{\scriptsize (\textsc{#1})}}					% Rule name

\newenvironment{grammatica}{$\begin{array}{lcll}}{\end{array}$}
\newcommand{\produzione}[3]{#1&::=&#2&\mbox{#3}}
\newcommand{\produzioneinline}[2]{#1::=#2}
\newcommand{\seguitoproduzione}[2]{&&#1&\mbox{#2}}
\newcommand{\terminale}[1]{\texttt{#1}}
\newcommand{\metavariable}[1]{\mathit{#1}}

\newcommand{\cd}{\metavariable{cd}}

\newcommand{\fd}{\metavariable{fd}}
\newcommand{\fds}{\metavariable{fds}}
\newcommand{\md}{\metavariable{md}}
\newcommand{\mds}{\metavariable{mds}}
\newcommand{\x}{\metavariable{x}}
\newcommand{\y}{\metavariable{y}}
\newcommand{\z}{\metavariable{z}}
\newcommand{\xs}{\metavariable{xs}}
\newcommand{\ys}{\metavariable{ys}}
\newcommand{\zs}{\metavariable{zs}}
\newcommand{\m}{\metavariable{m}}
\newcommand{\e}{\metavariable{e}}
\newcommand{\es}{\metavariable{es}}
\newcommand{\C}{\metavariable{C}}
\newcommand{\D}{\metavariable{D}}
\newcommand{\f}{\metavariable{f}}
\newcommand{\dec}{\metavariable{d}}
\newcommand{\decs}{\metavariable{ds}}
\newcommand{\X}{\metavariable{X}} %insieme di variabili

\newcommand{\this}{\terminale{this}}
\newcommand{\mutable}{\terminale{mut}}
\newcommand{\imm}{\terminale{imm}}
\newcommand{\lent}{\terminale{lent}}
\newcommand{\readable}{\terminale{read}}
\newcommand{\capsule}{\terminale{capsule}}

\newcommand{\Field}[2]{#1\ #2\terminale{;}}
\newcommand{\MethDec}[5]{{#1\ #2\ \terminale{(}#3,#4\terminale{)}\ \terminale{\{}#5\terminale{\}}}}
\newcommand{\Param}[2]{#1\ #2}
\newcommand{\FieldAccess}[2]{#1\terminale{.}#2}
\newcommand{\MethCall}[3]{#1\terminale{.}#2\terminale{(}#3\terminale{)}}
\newcommand{\FieldAssign}[3]{#1\terminale{.}#2\terminale{=}#3}
\newcommand{\ConstrCall}[2]{\terminale{new}\,#1\!\terminale{(}#2\terminale{)}}
\newcommand{\Block}[2]{\terminale{\{}#1\ #2\terminale{\}}}
\newcommand{\Dec}[3]{#1\,#2\terminale{=}#3\terminale{;}}
\newcommand{\Sequence}[2]{#1\terminale{;} #2}

\newcommand{\WellFormedTypeCtx}[1]{\vdash#1}

\newcommand{\MutGroup}{\metavariable{xs}}
\newcommand{\intType}{\terminale{int}}
\newcommand{\T}{\metavariable{T}}
\newcommand{\TPrime}{{\T}'}
\newcommand{\Type}[2]{#1\,#2}
\newcommand{\IsWellTyped}[1]{\vdash #1}
\newcommand{\TypeCheckGround}[2]{\vdash #1:#2}
\newcommand{\TypeCheck}[5]{#1;#2;#3\vdash #4:#5}
\newcommand{\AuxTypeCheck}[6]{#1;#2\,[#3];#4\vdash #5:#6}
\newcommand{\TypeCheckShort}[3]{#1\vdash #2:#3}
\newcommand{\TypeDec}[2]{#2{:}#1}
\newcommand{\LentLocked}{\metavariable{xss}}
\newcommand{\StronglyLocked}{\metavariable{ys}}
\newcommand{\Reduct}[2]{#1_{|#2}}
\newcommand{\Extends}[3]{#1\sqsubseteq_{#2}#3}
\newcommand{\LentVars}[1]{\aux{lentVars}(#1)}

\newcommand{\LessEq}[2]{#1\sqsubseteq#2}

\newcommand{\dv}{\metavariable{dv}}
\newcommand{\dvs}{\metavariable{dvs}}
\newcommand{\Ctx}[1]{\ctx[#1]}

\newcommand{\CtxP}[1]{\ctxP[#1]}

\newcommand{\ctx}{{\cal{E}}}
\newcommand{\genCtx}{{\cal{G}}}
\newcommand{\GenCtx}[1]{\genCtx[#1]}
\newcommand{\genCtxP}{{\cal{G'}}}
\newcommand{\GenCtxP}[1]{\genCtxP[#1]}
\newcommand{\Xctx}{{\cal{C}}}
\newcommand{\ctxP}{{\cal{E}'}}

\newcommand{\emptyctx}{[\ ]}
\newcommand{\val}{\metavariable{v}}
\newcommand{\stVal}{\metavariable{rv}}
\newcommand{\valPrime}{\metavariable{u}}
\newcommand{\vals}{\metavariable{vs}}
\newcommand{\vs}{\metavariable{vs}}

\newcommand{\reduce}[2]{#1\longrightarrow#2}
\newcommand{\Subst}[3]{#1[#2/#3]}
\newcommand{\Update}[4]{{#1[#2.#3{=}#4]}}

\newcommand{\congruence}[2]{{#1}\cong{#2}}

\newcommand{\extractMod}[1]{\mu^{#1}}
\newcommand{\aux}[1]{\textsf{#1}}
\newcommand{\notRef}[1]{{\aux{notVar}}(#1)}
\newcommand{\extractDec}[2]{\aux{dec}(#1,#2)}
\newcommand{\extractType}[2]{\aux{type}(#1,#2)}
\newcommand{\typeOf}[1]{\aux{type}(#1)}
\newcommand{\extractField}[3]{\aux{get}(#1,#2,#3)}
\newcommand{\fieldOf}[2]{\aux{get}(#1,#2)}
\newcommand{\fields}[1]{\aux{fields}(#1)}
\newcommand{\method}[2]{{\aux{method}(#1,#2)}}
\newcommand{\FV}[1]{\aux{FV}(#1)}
\newcommand{\HB}[1]{\aux{HB}(#1)}
\newcommand{\dom}[1]{\aux{dom}(#1)}
\newcommand{\domMut}[1]{\aux{dom}^\mutable(#1)}

\newcommand{\domGeqMut}{{\aux{dom}^{\geq\terminale{mut}}}}
\newcommand{\ImmClosed}[2]{#1\models\aux{imm-closed}(#2)}
\newcommand{\allImm}[1]{\forall^\imm(#1)}
\newcommand{\allMut}[1]{\forall^{\geq\mutable}(#1)}
\newcommand{\noCapture}[2]{#1{\parallel}#2}

\newsavebox{\saveScaled}
\newcommand*{\varScalableEnvironmentWidth}{1}
\newcommand*{\varScalableEnvironmentHeight}{1}
\newenvironment{Scaled}[2]{%
  \renewcommand*{\varScalableEnvironmentWidth}{#1}
  \renewcommand*{\varScalableEnvironmentHeight}{#2}
  \begin{lrbox}{\saveScaled}
  }{
  \end{lrbox}
  \noindent\scalebox{\varScalableEnvironmentWidth}[\varScalableEnvironmentHeight]{\usebox{\saveScaled}}
  }
\newcommand{\storableVal}{right-value}

\newcommand{\storableVals}{right-values}
\newcommand{\preRedex}{{\rho}}
\newcommand{\deriv}{{\cal D}}
\newcommand{\WFrv}[1]{\models#1}

\newcommand{\WFdv}[1]{\models_{st}#1}

\newcommand{\TypeEnv}[1]{\Gamma_{\!#1}}

\newcommand{\DecsOK}[4]{#1;#2;#3\vdash #4\,{\tt OK}}
\newcommand{\AuxDecsOK}[5]{#1;#2[#3];#4\vdash #5\,{\tt OK}}

\newcommand{\Capsulectx}[1]{\Decctx{\cx}{#1}}
\newcommand{\Immctx}[1]{\Decctx{\ix}{#1}}
\newcommand{\CapsulectxP}[1]{{\cal C}'_{\cx}[#1]}
\newcommand{\ImmctxP}[1]{{\cal C}'_{\ix}[#1]}
\newcommand{\ImmctxS}[1]{{\cal C}''_{\ix}[#1]}
\newcommand{\decctx}[1]{{\cal C}\!_{#1}}
\newcommand{\Decctx}[2]{{\cal C}\!_{#1}[#2]}
\newcommand{\DecEvCtx}[2]{#1_{#2}}

\newcommand{\mux}{\mu\hspace{.015cm}\x}
\newcommand{\cx}{{\tt c}\hspace{.015cm}\x}
\newcommand{\ix}{{\tt i}\hspace{.015cm}\x}

\newcommand{\Range}[3]{\forall#1\in#2..#3}

\newcommand{\cOrx}{w}

\newcommand{\connected}[3]{#2\stackrel{#1}{\longrightarrow}#3}

%% file: Intro.tex
\section{Introduction}
In languages with state and  mutations,  keeping control of aliasing relations is a key issue for correctness. This is  {exacerbated} by  concurrency mechanisms, since side-effects in one thread can affect the behaviour of another thread, hence unpredicted aliasing can induce unplanned/unsafe communication.

For these reasons, the last few decades have seen considerable interest in type systems for controlling sharing and interference, to make programs easier to maintain and understand, notably using \emph{type qualifiers} to restrict the usage of references \cite{ZibinEtAl10,GordonEtAl12,NadenEtAl12,ClebschEtAl15}. 

{In particular, it is very useful for a programmer to be able to rely on the \emph{capsule} and \emph{immutability} properties of a reference $\x$. To informally explain their meaning, let us consider the store as a graph, where nodes contain records of fields, which may be
references to other nodes, and let $\x$ be a reference\footnote{{Equivalently, a variable naming a reference: in our formalism the two notions are identified.}} to a node in the graph.
Depending on the qualifier of $\x$,  restrictions are imposed {and} assumptions can be made on the whole subgraph reachable from $\x$, as detailed below. }

{If $\x$ is $\capsule$, then the subgraph reachable from 
$\x$ is an isolated portion of store, that is, all its (non immutable) nodes
can be reached only through this reference. This allows programmers (and static analysis)
to identify state that can be safely handled by a thread.  In this paper we will use the name \emph{capsule} for this property, to avoid confusion with many variants in literature \cite{ClarkeWrigstad03,Almeida97,ServettoEtAl13a,Hogg91,DietlEtAl07,GordonEtAl12}.}

{If $\x$ is $\imm$ (immutable), then the subgraph reachable from 
$\x$ is an immutable portion of store. That is, we impose the restriction that fields of the node cannot be modified\footnote{This corresponds to the \emph{read-only} notion in literature.} ($\x.\f{=}\e$ is not legal), and we can assume that the subgraph reachable from $\x$ will 
not be modified through any other reference. As a consequence, an immutable reference can be safely shared
in a multithreaded environment.}

{Note that the above properties are \emph{deep/full}, that is, related to the whole object graph.} {For this reason,} unlike early
 discussions of read-only
  \cite{Boyland06} and the Rust language\footnote{\texttt{rust-lang.org}}, we do not offer any kind of back door to support internal mutability, since this would destroy the immutability assumption on the whole reachable object graph.
  
{In addition to $\capsule$ and $\imm$, we will consider three other type qualifiers.
\begin{itemize}
\item If $\x$ is $\mutable$ (mutable), then no restrictions are imposed and no assumptions can be made on $\x$.  
\item If $\x$ is $\lent$ \cite{ServettoZucca15,GianniniEtAl16}, also called \emph{borrowed} in literature \cite{Boyland01,NadenEtAl12}, then the subgraph reachable from $\x$ can be manipulated, but not shared, by a client. More precisely, the evaluation of an expression which uses a lent reference $\x$ can neither connect $\x$ to any other external reference, nor to the result of the expression, unless this expression is, in turn, lent.
\item Finally, if $\x$ is $\readable$ (readable), then  neither modification nor sharing are permitted. That is, we impose both the \emph{read-only} and the $\lent$ {(\emph{borrowed})} restriction; however, note that there is no immutability guarantee.
\end{itemize}}
  The
last two qualifiers
ensure intermediate properties used to derive the  capsule and immutable properties. 

%{First, the type system permits to recover uniqueness and immutability properties from mutable or readable references even in the presence of other mutable or readable references. This is achieved by rules which restrict the use of such other references in the portion of code for which the property should be recovered.}

Whereas (variants of) such qualifiers have appeared in previous literature (see \refToSection{related} 
for a detailed discussion on related work), there are two key novelties in our approach. 

{Firstly and more importantly, the type system is very expressive, as will be illustrated by many examples in \refToSection{examples}. Indeed, it adopts the \emph{recovery} technique \cite{GordonEtAl12,ClebschEtAl15}, that is, references (more in general, expressions) originally typed as mutable or readable can be recovered to be capsule or immutable, respectively, provided that no other mutable/readable references are used. 
However, expressivity of recovery is greatly enhanced, that is, the type system permits recovery of $\capsule$ and $\imm$ properties in much more situations. Indeed,  other mutable/readable references are not hidden once and for all. Instead, they can be used in a controlled way, so that it can be ensured that they will be not referenced from the result of the expression for which we want to recover the capsule or immutability property. Notably, they can be possibly used in some subexpression, if certain conditions are satisfied, thanks to two rules for \emph{swapping} and \emph{unrestricting} which are the key of the type system's expressive power.}

Secondly, we adopt an innovative execution model \cite{ServettoLindsay13,CapriccioliEtAl15,ServettoZucca15} for imperative languages which, 
differently from traditional  ones, {is a \emph{reduction relation on language terms}. That is, we do not add an auxiliary structure which mimics 
physical memory, but such structure is encoded in the language itself.} Whereas this makes no difference from a programmer's point of view, it is important on the foundational side, since, as will be {informally illustrated in \refToSection{informal} and formalized in \refToSection{results}, it makes possible to express %and prove 
uniqueness and immutability properties in a simple and direct way. }

This paper is an improved and largely extended version of \cite{GianniniEtAl16}. The novel contributions include reduction rules, more examples and proofs of results.

The rest of the paper is organized as follows: we provide syntax and an informal introduction in \refToSection{informal}, type system in \refToSection{typesystem},  {programming} examples in \refToSection{examples}, reduction rules in \refToSection{calculus},  results in \refToSection{results}, 
related work in \refToSection{related}, summary of paper contribution and outline of further work in \refToSection{conclu}. 
The appendix contains some of the proofs omitted from \refToSection{results}.

%% file: Informal.tex
\section{Informal presentation of the {calculus}}\label{sect:informal}
Syntax and types are given in \refToFigure{syntax}.  We assume sets of \emph{variables} $\x$, $\y$, $\z$, \emph{class names} $\C$, \emph{field names} $\f$, and \emph{method names} $\m$. 
We adopt the convention that a metavariable which ends by $\metavariable{s}$ is implicitly defined as a (possibly empty) sequence, {for example} {$\decs$ is defined by $\produzioneinline{\decs}{\epsilon\mid \dec\ \decs}$}, where $\epsilon$ denotes the empty sequence.

\begin{figure}
\framebox{{\small \begin{grammatica}
\produzione{\cd}{\terminale{class}\ \C\ \terminale{\{}\fds\ \mds \terminale{\}}}{class declaration}\\*
\produzione{\fd}{
\Field{\Type{\imm}{\C}}{\f}
\mid
\Field{\Type{\mutable}{\C}}{\f}\mid\Field{\intType}{\f}}{field declaration}\\*
\produzione{\md}{\MethDec{\T}{\m}{{\mu}}{\T_1\,\x_1,\ldots,\T_n\,\x_n}{\e}}{method declaration}\\*
\produzione{\e}{\x\mid\FieldAccess{\e}{\f}\mid\MethCall{\e}{\m}{\es}\mid\FieldAssign{\e}{\f}{\e}\mid\ConstrCall{\C}{\es}\mid\Block{\decs}{\e}}{expression}\\*
\produzione{\dec}{\Dec{\T}{\x}{\e}}{variable declaration}\\*
\\
\produzione{\T}{\Type{\mu}{\C}\mid\intType}{type}\\*
\produzione{\mu}{\mutable\mid\imm\mid\capsule\mid\lent\mid\readable}{{(type)} qualifier}\\*
\\
\produzione{\dv}{\Dec{\T}{\x}{\stVal}}{evaluated declaration}\\
%\produzione{{\valPrime,}\val}{{\x\mid\ConstrCall{\C}{\vals}\mid\Block{\dvs}{\val}}}{value}\\*
\produzione{\stVal}{\ConstrCall{\C}{\xs}\mid\Block{\dvs}{\x}\mid\Block{\dvs}{\ConstrCall{\C}{\xs}}}{\storableVal}
\end{grammatica}
}}
\caption{Syntax and types}\label{fig:syntax}
\end{figure}

The syntax mostly follows Java and Featherweight Java (FJ) \cite{IgarashiEtAl01}. A class declaration consists of a class name, a sequence of field declarations and a sequence of method declarations. A field declaration consists of a field type and a field name.
 A method declaration consists, as in FJ, of a return type, a method name, a list of parameter names with their types, and a body which is an expression. {However, there is an additional component: the type qualifier for $\this$, which is {written as additional (first) element of the parameter list}.}
  As in FJ, we assume for each class a canonical constructor whose parameter list exactly corresponds to the class fields, and we assume no multiple declarations of classes in a class table, fields and methods in a class declaration. 

For expressions, in addition to the standard constructs of imperative object-oriented languages, we have 
\emph{blocks}, which are sequences of variable declarations, followed by a \emph{body} which is an expression.
Variable declarations consist of a type, a variable and an initialization expression. 
Types are class names decorated by a \emph{{(type)} qualifier}. We also include $\intType$ as an example of primitive type, but we do not formally model related operators used in the examples, such as integer constants and sum. We assume no multiple declarations for variables in a block, that is, $\decs$ can be seen as a map from variables into declarations, and we use the notations $\dom{\decs}$ and $\decs(\x)$.

Blocks are a fundamental construct of our {calculus}, since sequences of local variable declarations, when evaluated, are used to directly represent store  in the language itself. A declaration is evaluated if its initialization expression is a \emph{right-value}, that is, either an \emph{object state} (a constructor invocation where arguments are references) or a block in which declarations are evaluated and the body is a reference or an object state.

For instance\footnote{In the examples, we omit for readability the brackets of the outermost block.}, assuming that class \lstinline{B} has a $\mutable$ field of type \lstinline{B}{, in the block}:
\begin{lstlisting}
mut B x= new B(y); mut B y= new B(x); x 
\end{lstlisting}
the two declarations can be seen as a store where \lstinline{x} denotes an object of class \lstinline{B} whose field is \lstinline{y}, and conversely, as shown in \refToFigure{es1}(a).
\begin{figure}[ht]
\begin{center}
\includegraphics
[width=1\textwidth]
{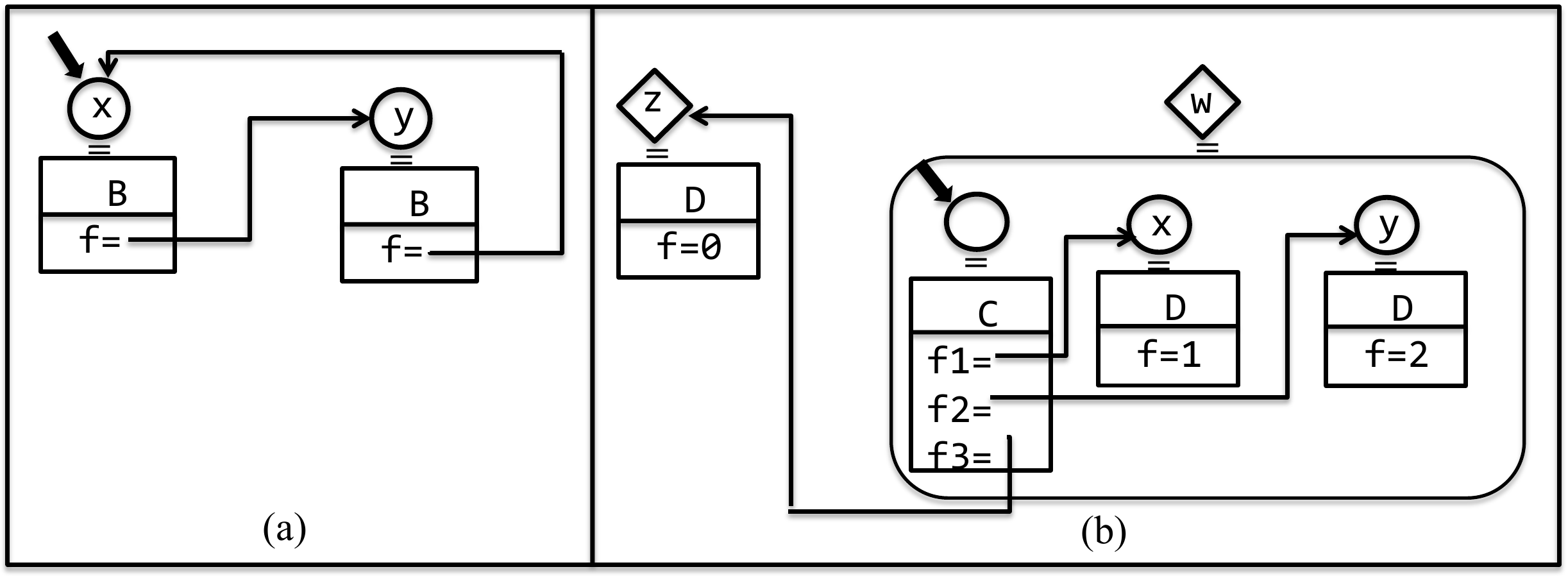}
\end{center}
\caption{Graphical representation of the store} \label{fig:es1}
\end{figure}
The whole block denotes a store with an entry point (graphically represented by a thick arrow). {In our graphical representation circles denote mutable references {(\lstinline{x}{} and \lstinline{y}{} in \refToFigure{es1})}, diamonds denote immutable references (\lstinline{z}{} and \lstinline{w}{} in \refToFigure{es1}(b)) and squares denote {capsule references}  (\lstinline{z}{} in \refToFigure{esRed2})}.
 
Stores are \emph{hierarchical}, rather than flat as it usually happens in models of imperative languages. 
For instance, assuming that class \lstinline{C}{} has two $\mutable$ \lstinline{D}{} and one $\imm$ \lstinline{D}{} fields,  and class \lstinline{D}{} has an integer field, the following is a store:
\begin{small}
\begin{lstlisting}
imm D z= new D(0);
imm C w= {mut D x= new D(1); mut D y= new D(2); new C(x,y,z)}  
\end{lstlisting}
\end{small}

Here, the {right-value} associated to \lstinline{w}{} is a block introducing local declarations, that is, in turn a store {with an entry point}, as shown in \refToFigure{es1}(b). 
The advantage of this hierarchical shape is that it models in a simple and natural way constraints about aliasing among objects, notably:
\begin{itemize}
\item the fact that an object is not referenced from outside some enclosing object is directly modeled by the block construct: for instance, the objects denoted by \lstinline{x}{} and \lstinline{y}{} can only be reached through \lstinline{w}{}
\item conversely, the fact that an object does not refer to the outside is modeled by the fact that the corresponding block is closed, that is, has no free variables\footnote{In other words, our calculus smoothly integrates memory representation with shadowing and $\alpha$-conversion.}: for instance, the object denoted by \lstinline{w}{} is not closed, since it refers to the external object \lstinline{z}{}.
\end{itemize}
%In the graphical representation, {variables in a grey circle are immutable references, others are mutable.}
In the graphical representation the reference corresponding to \lstinline{new C(x,y,z)}{} is anonymous.
Note also that, in this example, mutable variables in the local store of \lstinline{w}{} are not visible from the outside. This models in a natural way the fact that the portion of store denoted by \lstinline{w}{} is indeed immutable, as will be detailed in the sequel.

We illustrate now the meaning of the qualifiers $\mutable$, $\imm$, and $\capsule$.
A mutable variable refers to a portion of store that can be modified during execution. For instance, the block 
\begin{lstlisting}
mut B x= new B(y); mut B y= new B(x); x.f=x
\end{lstlisting}
reduces to
\begin{lstlisting}
mut B x= new B(x); mut B y= new B(x); x
\end{lstlisting}
We give a graphical representation of this reduction in \refToFigure{esRed1}
\begin{figure}[ht]
\begin{center}
\includegraphics[width=1\textwidth]{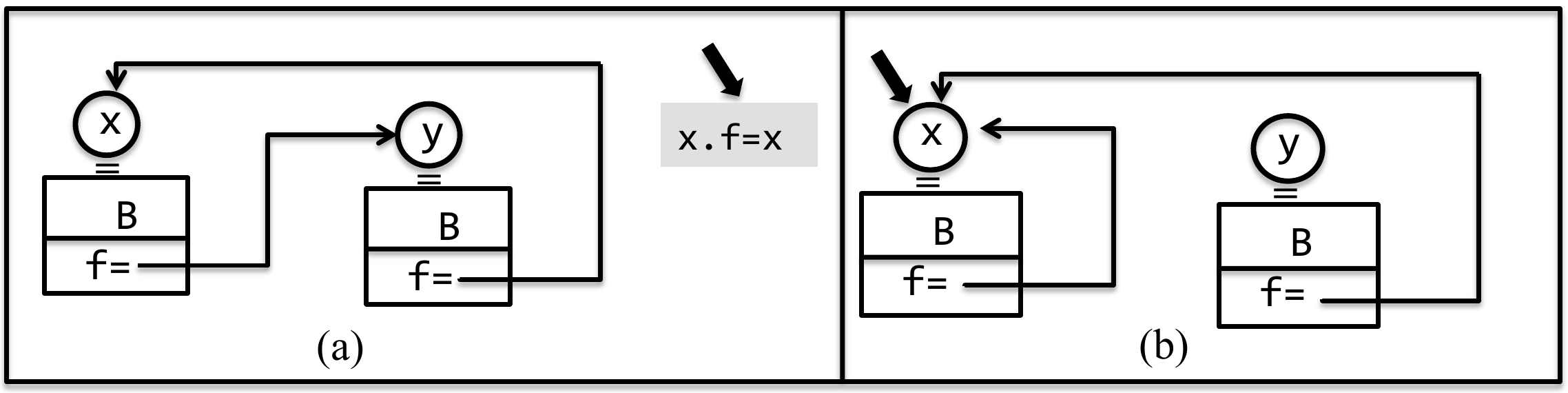}
\end{center}
\caption{Example of reduction (1)} \label{fig:esRed1}
\end{figure}
 where we
highlight in grey expressions which {are not reduced yet}. So in  (a) the body of the block is the expression \lstinline{x.f=x}{}, whose evaluation
modifies the field \lstinline{f}{} of \lstinline{x}{}, and returns \lstinline{x}{}. The result of the reduction is shown in (b). 

Variables declared immutable, instead, {refer to a portion of store which cannot be modified.}
 Immutability is \emph{deep/full}, that is, all the nodes in the reachable object graph of an immutable reference are immutable themselves.
Therefore, in the enclosing scope of the declaration
\begin{lstlisting}
imm C w= {mut D x= new D(1); mut D y= new D(2); new C(x,y,z)}
\end{lstlisting}
the variable \Q@z@ must be declared \Q@imm@, and we cannot have an assignment to a field of \Q@w@. 

A variable declared $\capsule$ refers to an \emph{isolated} portion of store, where local objects can freely reference
each other but for which the variable is the only external reference. For instance in:
\begin{lstlisting}
capsule B z = { mut B x= new B(y); mut B y= new B(x); x } 
\end{lstlisting}
the internal objects denoted by \Q@x@ and \Q@y@ can be only be accessed  through \Q@z@. A capsule variable can be used 
once and for all to ``move'' an isolated portion of store 
to another node in the store. To get more flexibility, external immutable 
references are freely allowed.  For instance, in the example above of the declaration of \lstinline{w}{}, the inizialization expression has a $\capsule$ type.
In our type system, capsule types are subtypes of both mutable and immutable types. Hence,
capsule expressions can initialize both mutable and immutable references. To preserve the capsule property, we need a 
{\em linearity constraint}: that is, a \emph{syntactic} constraint that in well-formed expressions capsule references 
can occur at most once in their scope (see more comments at page \pageref{linearity}). 

Consider the term\label{capsule-example-1}
\begin{lstlisting}
mut D y=new D(0); 
capsule C z={mut D x=new D(y.f=y.f+1); new C(x,x)}
\end{lstlisting}
In \refToFigure{esRed2}(a) we have a graphical representation of this term. 
%, where the variable circled in blue is a capsule reference.  
\begin{figure}[ht]
\begin{center}
\includegraphics[width=1\textwidth]{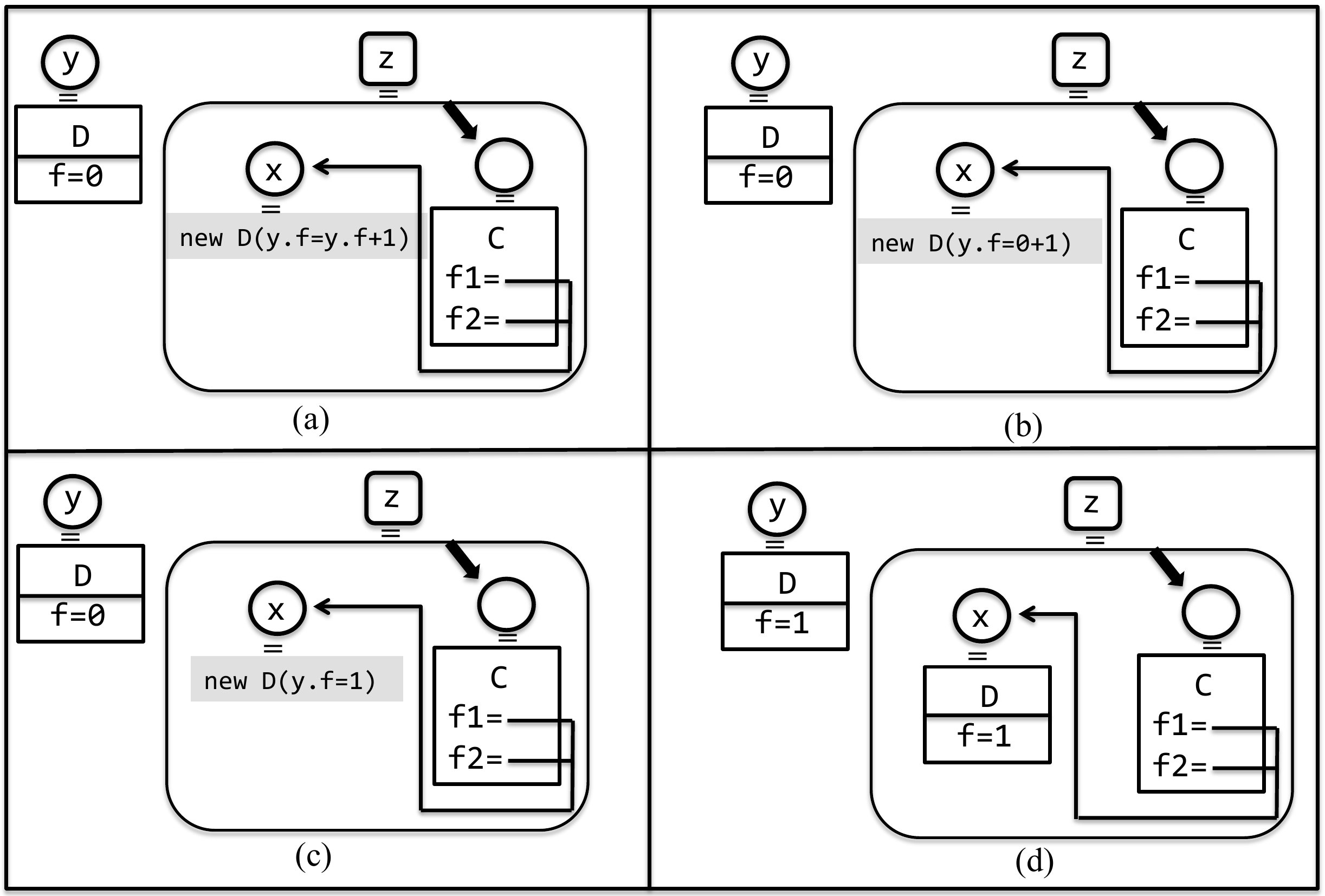}
\end{center}
\caption{Example of reduction (2)} \label{fig:esRed2}
\end{figure}
The evaluation of the expression on the right-hand side of \Q@x@ starts by evaluating 
\Q@y.f+1@, which triggers the evaluation of \Q@y.f@. The result is shown in (b),
then the sum \Q@0+1@ is evaluated, returning \Q@1@, as shown in (c). The evaluation of 
the field assignment 
\Q@y.f=1@ updates the field \Q@f@ of \Q@y@ to \Q@1@, and \Q@1@ is returned.
Since \Q@new D(1)@ is a value, the whole term is fully evaluated, and 
it is shown in (d). 

To be able to typecheck more expressions as $\capsule$ or $\imm$, 
we introduce the $\lent$ and $\readable$ qualifiers.
References with such qualifiers can be used in a restricted way. That is, no aliasing can be introduced between a $\lent$ reference and another reference, and a $\lent$ reference cannot be part of the result of the expression where it is used, unless this expression is, in turn, $\lent$. A $\readable$ reference is $\lent$ and, moreover, cannot be modified.

The relation $\mu\leq\mu'$ intuitively means that qualifier $\mu'$ imposes more restrictions than $\mu$, hence a $\mu$ reference can be safely used where a $\mu'$ is required.
Notably,  $\lent$ and $\imm$ references can be used where a $\readable$ is required, a mutable reference where a lent is required, and a $\capsule$ reference can be used everywhere. 

{Altogether, the subtyping relation is the reflexive and transitive relation on types induced by
\begin{quote}
{$\intType\leq\intType$}\\
$\Type{\mu}{\C}\leq\Type{\mu'}{\C}$ if $\mu\leq\mu'$\\
$\capsule\leq\mutable\leq\lent\leq\readable$\\
$\capsule\leq\imm\leq\readable$
\end{quote}
Note that capsules can be used as mutable or immutable, with the constraint that capsule references can be used only once.
So, a capsule reference can be seen as a reference whose destiny has not been decided yet.

In some cases it is possible to move the type of an expression against the subtype hierarchy. Notably, {we can recover the $\capsule$ property for a $\mutable$ expression, and the $\imm$ property for a $\readable$ expression}, provided that some of the free variables in the expression are used in a controlled way. Recovery will be described in the following section.}

{The situation is graphically depicted in \refToFigure{hierarchy}.}
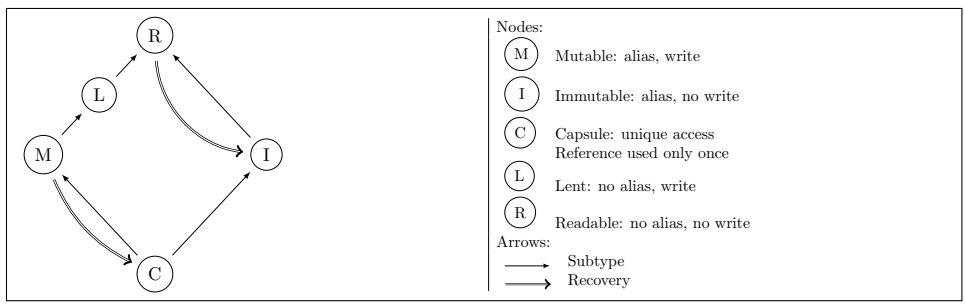
\begin{figure}
\framebox
{
\begin{minipage}[H]{0.5\textwidth}
\begin{Scaled}{0.7}{0.7}
\begin{tikzpicture}[shorten <=0.4em,shorten >=0.4em]
\node [circle,draw,shift={(1ex,15ex)}] (M){M};
\node [circle,draw,shift={(15ex,0ex)}] (C){C};
\node [circle,draw,shift={(29ex,15ex)}] (I){I};
\node [circle,draw,shift={(8ex,22.5ex)}] (L){L};
\node [circle,draw,shift={(15ex,30ex)}] (R){R};
\draw [arrows={-latex}] (I) -- (R);
\draw [arrows={-latex}] (M) -- (L);
\draw [arrows={-latex}] (L) -- (R);
\draw [arrows={-latex}] (C) -- (M);
\draw [arrows={-latex}] (C) -- (I);
\draw[->,double] (M) to[bend right=20] (C);
\draw[->,double](R) to[bend right=40] (I); 
\end{tikzpicture}
\end{Scaled}
\end{minipage}
\begin{minipage}[H]{0.5\textwidth}
\begin{Scaled}{0.6}{0.6}
$
\begin{array}{|l}
\mbox{Nodes:}
\\
\begin{array}{l l}
\begin{tikzpicture}[shorten <=0.4em,shorten >=0.4em]
\node [circle,draw,shift={(0,0)}] (M){M};
\end{tikzpicture}
&\raisebox{1.5ex}{Mutable: alias, write}
\\
\begin{tikzpicture}[shorten <=0.4em,shorten >=0.4em]
\node [circle,draw,shift={(0,0)}] (I){${\ }$I${\ }$ };
\end{tikzpicture}
&\raisebox{1.5ex}{Immutable: alias, no write}
\\
\begin{tikzpicture}[shorten <=0.4em,shorten >=0.4em]
\node [circle,draw,shift={(0,0)}] (C){C};
\end{tikzpicture}
&\parbox{50ex}{Capsule: unique access\\*
Reference used only once}
\\

\begin{tikzpicture}[shorten <=0.4em,shorten >=0.4em]
\node [circle,draw,shift={(0,0)}] (L){L};
\end{tikzpicture}
&
\parbox{50ex}{Lent: no alias, write}
\\
\begin{tikzpicture}[shorten <=0.4em,shorten >=0.4em]
\node [circle,draw,shift={(0,0)}] (C){R};
\end{tikzpicture}
&
\parbox{50ex}{Readable: no alias, no write}
\end{array}
\\
\mbox{Arrows:}
\\
\begin{array}{l l}
\begin{tikzpicture}
\draw [arrows={-latex}] (0,0) -- (1,0);
\end{tikzpicture}
&\mbox{Subtype}
\\
\begin{tikzpicture}
\draw[->,double] (0,0) to (1,0);
\end{tikzpicture}
&\mbox{{Recovery}}
\end{array}
\end{array}
$
\end{Scaled}
\end{minipage}
}
\caption{Type qualifiers and their relationships}
\label{fig:hierarchy}
\end{figure}

%% file: TypeSystem.tex
\section{Type system}\label{sect:typesystem}
Type contexts are defined by:
\begin{center}
\begin{grammatica}
\produzione{\Delta}{\Gamma;\LentLocked;\StronglyLocked}{type context}\\*
\produzione{\Gamma}{\TypeDec{\T_1}{\x_1},\ldots,\TypeDec{\T_n}{\x_n}}{type assignment}
\end{grammatica}
\end{center} 
{In a type context $\Delta=\Gamma;\LentLocked;\StronglyLocked$, $\Gamma$ is a usual {\em assignment of types} to variables. We write $\dom{\Gamma}$ for the set of variables in $\Gamma$, and  $\SubstFun{\Gamma}{\Gamma'}$ for the concatenation of $\Gamma$ and $\Gamma'$ where, for the variables occurring in both domains, $\Gamma'$ takes precedence.  We will also use $\domMut{\Gamma}$ for the set of mutable variables in $\Gamma$, and other analogous notations.}

{According to our convention, $\LentLocked$ is a sequence $\xs_1\ldots\xs_n$ of sequences of variables, and $\StronglyLocked$ is a sequence of variables. All such sequences are assumed to be sets (that is, order and repetitions are immaterial). The sets  $\xs_1\ldots\xs_n$ are assumed to be pairwise disjoint, and are called \emph{{lent group}s}, whereas the $\mutable$ variables in $\Gamma$ which do not belong to any lent group form the \emph{{mutable group}}. Finally, variables in $\ys$ are called \emph{restricted} variables.}

{Lent groups and restricted variables are the key novelty of our type system.
%, hence we roughly explain their role, which will be formally detailed in the following. 
The property they ensure is that, if $\TypeCheck{\Gamma}{\LentLocked}{\StronglyLocked}{\e}{\T}$, then the final result of $\e$ will be not connected to variables which are in $\LentLocked$ or $\StronglyLocked$. In this way, to ensure that the final result of $\e$ is an isolated portion of store, it is enough to typecheck $\e$ in a type context where all its free mutable variables are in a {lent group} or restricted.  This is a crucial improvement with respect to previous mechanisms of recovery \cite{GordonEtAl12,ClebschEtAl15}, where this could be only ensured if  $\e$ had \emph{only} immutable or isolated ($\capsule$ in our teminology) free variables. 
We have distinct {lent group}s $\xs_1\ldots\xs_n$ since there can be nested recoveries, hence, to ensure that all are safe, we need also the auxiliary property that no aliasing is introduced between different groups, as will be illustrated in detail by the last example of \refToSection{examples}. That is, if the reachable graphs of $\xs_i$ and $\xs_j$ were disjoint before the evaluation of the expression, then after the evaluation they should be still disjoint.\footnote{Note that this is different from \emph{regions} as in, e.g., \cite{BocchinoEtAl09,HallerOdersky10}, which are sets of references \emph{assumed} to be disjoint.}} 

{The type system ensures the above properties by imposing that variables in $\LentLocked$ and $\StronglyLocked$ are used in controlled way:
\begin{itemize}
\item A $\mutable$ variable in $\Gamma$ which belongs to a {lent group} $\xs_i$ can only be used as $\lent$. Notably, write access is not directly allowed since it could introduce aliasing.
\item A $\readable$, $\lent$ or $\mutable$ variable in $\Gamma$ which is restricted cannot be used at all, that is, is hidden.\footnote{{However, we use the terminology ``restricted'' since this hiding is non permanent.}}
\end{itemize}
The important point is that such constraints are not imposed once and for all when typechecking an expression. That is, a subexpression can be typechecked with different {lent group}s and restricted variables, more precisely:
\begin{itemize}
\item In a subexpression we can \emph{swap} one of the {lent group}s $\xs_i$ with the {mutable group}, weakening to $\lent$ the type of the subexpression if it was mutable. 
In this way, write access to variables in $\xs_i$ is allowed, but this is safe since the result of the subexpression is in turn $\lent$, so no aliasing can be introduced with the result of the main expression. 
\item In a subexpression we can \emph{unrestrict} variables in $\ys$, that is, freely use such variables, provided that the type of the subexpression is $\imm$ or $\capsule$.  Again this guarantees that no aliasing can be introduced with the result of the main expression.
\end{itemize} }

Lent groups and restricted variables are introduced as an effect of applying \emph{recovery} rules \rn{t-capsule} and \rn{t-imm}, whereas swapping and unrestricting are obtained by applying rules \rn{{t-swap}} and \rn{t-unrst} rules, as will be explained in detail later.

Contexts $\Delta=\Gamma;\LentLocked;\StronglyLocked$ are \emph{well-formed}, written $\WellFormedTypeCtx{\Delta}$, 
if, for $\LentLocked=\xs_1\ldots\xs_n$, the following conditions hold:
\begin{itemize}
\item no variable is $\lent$  in $\Gamma$
\item if  $\x\in\xs_i$ for some $i\in 1..n$,  then $\x$ is $\mutable$ in $\Gamma$
\item if $\y\in\StronglyLocked$, then  $\y$ is $\mutable$ or $\readable$  in $\Gamma$
\item if $\y\in\StronglyLocked$ and is $\mutable$ in $\Gamma$, then $\y\in\xs_i$ for some $i\in 1..n$.
%and, for all $\x\in\xs_i$, $\x\in\StronglyLocked$
\end{itemize}
We do not consider $\lent$ variables in $\Gamma$ since assigning the type $\Type{\lent}{\C}$ to a variable $\x$  is encoded by assigning to $\x$ the type $\Type{\mutable}{\C}$ and having $\x$ in a lent group in $\LentLocked$ (see the explanation of rule \rn{t-block} in the following). The  last condition requires \emph{coherency} between $\LentLocked$ and $\StronglyLocked$, that is, restricted mutable variable are in a lent group. It is easy to check that well-formedness of contexts is preserved in type derivations. That is, if $\Delta\TypeCheckGround{\e}{\T}$ and ${\Delta}$ is well-formed, then in all sub-derivations the contexts are well-formed. In the following when we write 
$\Delta\TypeCheckGround{\e}{\T}$ we assume that ${\Delta}$ is well-formed.

In the rules we use information extracted from the class table, which is modelled, as usual, by the following functions:{
\begin{itemize}
\item $\fields{\C}$ gives, for each declared class $\C$, the sequence of its fields declarations
\item {$\method{\C}{\m}$ gives, for each method $\m$ declared in class $\C$, the tuple\\
 $\FourTuple{\T}{\mu}{\Param{\T_1}{\x_1}\ldots\Param{\T_n}{\x_n}}{\e}$ consisting of its return type, type qualifier for $\this$, parameters, and body. }
\end{itemize}
We assume method bodies to be well-typed w.r.t.\ {the type annotations in the method declaration}. More precisely, if
%\begin{quote}
%$\method{\C}{\m}=\FourTuple{\T}{\mu}{\Param{\T_1}{\x_1}\ldots\Param{\T_n}{\x_n}}{\e}$
%\end{quote}
%then $\e$ should be well-typed in a type context where parameters declared $\lent$, including the implicit parameter $\this$, have been encoded by singleton lent groups, expressing the requirement that they should not be aliased by the method. Formally, it should be
%\begin{quote}
%$\TypeCheck{\Gamma}{\{\x_i\mid T_i=\Type{\lent}{\C_i}\}\cup\{\this\mid\mu=\lent\}}{\emptyset}{\e}{\T}$
%\end{quote}
%with $\Gamma=\TypeDec{\Type{\mu'}{\C}}{\this}, \TypeDec{\T'_1}{\x_1},\ldots,\TypeDec{\T'_n}{\x_n}$, where $\T_i=\Type{\mutable}{\C_i}$ if $\T_i=\Type{\lent}{\C_i}$, $\T'_i=\T_i$ otherwise, and $\mu'=\mutable$ if $\mu=\lent$, $\mu'=\mu$ otherwise. 
{\begin{quote}
$\method{\C}{\m}=\FourTuple{\T}{\mu}{\Param{\T_1}{\x_1}\ldots\Param{\T_n}{\x_n}}{\e}$
\end{quote}
then 
\begin{quote}
$\TypeCheck{\Gamma}{\LentLocked}{\emptyset}{\e}{\T}$
\end{quote}
where $\LentLocked$ {consists of} the singletons of the parameters declared $\lent$, including the implicit parameter $\this$, expressing the requirement that they should not be aliased by the method, and $\Gamma=\TypeDec{\Type{\mu'}{\C}}{\this}, \TypeDec{\T'_1}{\x_1},\ldots,\TypeDec{\T'_n}{\x_n}$, where $\T_i=\Type{\mutable}{\C_i}$ if $\T_i=\Type{\lent}{\C_i}$, $\T'_i=\T_i$ otherwise, and $\mu'=\mutable$ if $\mu=\lent$, $\mu'=\mu$ otherwise. 
}

Typing rules are given in \refToFigure{typing}. {In the typing rules, when we need to make explicit the mutable group, we use the  auxiliary judgment $\AuxTypeCheck{\Gamma}{\LentLocked}{\MutGroup}{\StronglyLocked}{\e}{\T}$, which stands for the judgment 
$\TypeCheck{\Gamma}{\LentLocked}{\StronglyLocked}{\e}{\T}$ with the side condition $\MutGroup=\domMut{\Gamma}{\setminus}\LentLocked$, meaning that $\MutGroup$ are the $\mutable$ variables in $\Gamma$ which do not belong to {any lent} group in $\LentLocked$.}

\begin{figure}[ht!]
\framebox{
\begin{footnotesize}
\begin{math}
\begin{array}{l}
\NamedRule{t-capsule}{\AuxTypeCheck{\Gamma}{\LentLocked\ {\xs}}{{\emptyset}}{\StronglyLocked}{\e}{\Type{\mutable}{\C}}}{{\AuxTypeCheck{\Gamma}{\LentLocked}{\MutGroup}{\StronglyLocked}{\e}{\Type{\capsule}{\C}}}}\\
\Space
%{\xs=\domMut{\Gamma}{\setminus}\LentLocked}
\NamedRule{t-imm}{
\AuxTypeCheck{\Gamma}{\LentLocked\ \xs}{{\emptyset}}{\domGeqMut(\Gamma) }{\e}{\Type{\readable}{\C}}
}{
{\AuxTypeCheck{\Gamma}{\LentLocked}{\xs}{\StronglyLocked}{\e}{\Type{\imm}{\C}}}
}{
%{\xs=\domMut{\Gamma}{\setminus}\LentLocked}
}
\\[5ex]
\NamedRule{{t-swap}}{\AuxTypeCheck{\Gamma}{\LentLocked\ \xs'}{\xs}{\StronglyLocked}{\e}{{\T}}}{{\AuxTypeCheck{\Gamma}{\LentLocked\ \xs}{\xs'}{\StronglyLocked}{\e}{{\TPrime}}}}{
%{\xs'=\domMut{\Gamma}{\setminus}(\LentLocked\ \xs)}\\
{\xs\cap\ys}=\emptyset\\
{\TPrime=
\begin{cases}
\Type{\lent}{\C}&\mbox{if}\ \T=\Type{\mutable}{\C}\\
\T&\mbox{otherwise}
\end{cases}}
}
\\[5ex]
\NamedRule{t-{unrst}}{
\TypeCheck{\Gamma}{\LentLocked}{\emptyset}{\e}{{\T}}
}{
\TypeCheck{\Gamma}{\LentLocked}{\StronglyLocked}{\e}{{\T}}
}{{\T=\Type{\mu}{\C}\Longrightarrow\mu\leq\imm}}
\Space
\NamedRule{t-sub}{\TypeCheckShort{\Delta}{\e}{\T}
}{
\TypeCheckShort{\Delta}{\e}{\TPrime}}{
\T\leq\TPrime
}
\\[5ex]
{\NamedRule{t-var}{}{\TypeCheck{\Gamma}{\LentLocked}{\StronglyLocked}{\x}{\TPrime}}{
\Gamma(\x)=\T\ \wedge \ 
\x\notin\StronglyLocked\\
\TPrime=
\begin{cases}
\Type{\lent}{\C}&\mbox{if}\ \T{=}\Type{\mutable}{\C}\,\wedge\,\x\in\LentLocked\\
\T&\mbox{otherwise}
\end{cases}
}}
\\[6ex]
{\NamedRule{t-field-access}{\TypeCheckShort{\Delta}{\e}{\Type{\mu}{\C}}}{\TypeCheckShort{\Delta}{\FieldAccess{\e}{\f_i}}{\TPrime_i}}{
\fields{\C}=\Field{\T_1}{\f_1}\ldots\Field{\T_n}{\f_n}\\
\TPrime_i=\begin{cases}
\Type{\mu}{\C_i}\ \mbox{if}\ \T_i=\Type{\mutable}{\C_i}\\
\T_i\ \mbox{otherwise}
\end{cases}
}}
\\[5ex]
\NamedRule{t-meth-call}{\TypeCheckShort{\Delta}{\e_i}{\T_i}\Space\forall i\in 0..n}{\TypeCheckShort{\Delta}{\MethCall{\e_0}{\m}{\e_1,\ldots,\e_n}}{\T}}{
\begin{array}{l}
\T_0=\Type\mu\C\\
\method{\C}{\m}=\FourTuple{\T}{\mu}{\Param{\T_1}{\x_1}\ldots\Param{\T_n}{\x_n}}{\e}
\end{array}
}
\\[5ex]
\NamedRule{t-field-assign}{
  \TypeCheckShort{\Delta}{\e}{\Type{\mutable}{\C}}
  \Space
  \TypeCheckShort{\Delta}{\e'}{\T_i}
  }{
  \TypeCheckShort{\Delta}{\FieldAssign{\e}{\f_i}{\e'}
  }{\T_i}
  }{
  {\fields{\C}=\Field{\T_1}{\f_1}\ldots\Field{\T_n}{\f_n}}
  }
\\[5ex]
\NamedRule{t-new}{\TypeCheckShort{\Delta}{\e_i}{\T_i}\Space\forall i\in 1..n}{\TypeCheckShort{\Delta}{\ConstrCall{\C}{\e_1,\ldots,\e_n}}{\Type{\mutable}{\C}}}{
\fields{\C}=\Field{\T_1}{\f_1}\ldots\Field{\T_n}{\f_n}\\
}
\\[5ex]
{\NamedRule{t-block}{
{\AuxTypeCheck{\SubstFun{\Gamma}{\TypeEnv{\decs}}}{\LentLocked_i}{\MutGroup_i}{\StronglyLocked'}{\e_i}{\TypeEnv{\decs}(\x_i)}\ \ {\forall i\in 1..n}}\\
\AuxTypeCheck{\SubstFun{\Gamma}{\TypeEnv{\decs}}}{\LentLocked'}{{\MutGroup'}}{\StronglyLocked'}{\e}{\T}}
{\AuxTypeCheck{\Gamma}{\LentLocked}{\MutGroup}{\StronglyLocked}{\Block{\decs}{\e}}{\T}}{
\begin{array}{l}
\decs=\Dec{\T_1}{\x_1}{\e_1}\ldots\Dec{\T_n}{\x_n}{\e_n}\\
\LessEq{(\LentLocked{\setminus}\dom{\decs})}{\LentLocked'}\\
(\MutGroup{\setminus}\dom{\decs})\subseteq\MutGroup'\\
\forall i\in 1..n\ \LentLocked_i\ \MutGroup_i = \LentLocked'\ \MutGroup'\\
\Space \mbox{and}\ \x_i\in\domMut{\TypeEnv{\decs}}\Rightarrow\x_i\in \MutGroup_i\\
\StronglyLocked'={\StronglyLocked}{\setminus}{\dom{\decs}}
\end{array}}}
\end{array}
\end{math}
\end{footnotesize}
}
\caption{Typing rules}\label{fig:typing}
\end{figure}
 
Rules \rn{t-{capsule}} and \rn{t-{imm}} 
 model \emph{recovery},
 that is, can be used to recover a more specific type for an expression,
 under the conditions that the use of some free variables in the expression is} {controlled}.
There are two kinds of recovery:
\begin{itemize}
\item $\mutable\Rightarrow\capsule$
\\* As shown in rule \rn{t-{capsule}}, an expression can be typed $\capsule$ in $\Gamma;\LentLocked;\StronglyLocked$ if it can be typed $\mutable$ by turning in lent the mutable group ($\xs$), which becomes empty.
Formally, this group is added to $\LentLocked$.  
\item $\readable\Rightarrow\imm$ 
\\*As shown in rule \rn{t-{imm}}, an expression can be typed $\imm$ in $\Gamma;\LentLocked;\StronglyLocked$  if it can be typed $\readable$ by {turning in lent the mutable group ($\xs$)}, as in the {recovery} above, and, moreover, restricting currently available mutable, lent, and readable variables ($\domGeqMut(\Gamma)$).
\end{itemize}

Along with {recovery} rules which introduce {lent group}s and restricted variables, we have two corresponding elimination rules.
In the detail:
\begin{itemize}
\item\rn{{t-swap}}
\\* An expression can be typed in $\Gamma;\LentLocked\ \xs;\StronglyLocked$ if it can be typed by turning into mutable 
some {lent group} {($\xs$)}, by swapping this group with the {current mutable group}  ($\xs'$). {The side condition $\xs\cap\ys=\emptyset$ prevents to swap a lent group including restricted variables.}
The type obtained in this way is weakened to $\lent$, if it was $\mutable$.
\item
\rn{t-unrst} 
\\* An expression can be typed in $\Gamma;\LentLocked;\StronglyLocked$ if it can be typed by unrestricting all restricted variables {$\StronglyLocked$}, provided that the type obtained in this way is $\capsule$ or $\imm$ {or a primitive type}.
\end{itemize}

{Note that, without these last two rules, recovery rules would be essentially equivalent to those in prior work \cite{GordonEtAl12,ClebschEtAl15}. Rules \rn{t-swap} and \rn{t-unrst}  add power to the type system, and they are also the reason {it} requires the $\LentLocked$ and $\StronglyLocked$ information during typechecking.
Prior work \cite{GordonEtAl12,ClebschEtAl15} can afford to simply locally using weakening and subtyping to make references inaccessible
or convert $\mutable$ references to $\lent$.}

{Note also that recovery rules, \rn{t-swap}, and \rn{t-unrst}  are not syntax-directed, analogously to the subsumption rule. 
In other words, their functionality does not  restrict how code is written: they are \emph{transparent} to the programmer, in the sense that they are applied when needed.  The programmer can simply rely on the fact that expressions are $\capsule$ or $\imm$, respectively, in situations where this is intuitively expected, as we illustrate by the following examples (other will be provided in next section).}

{We explain now in detail how recovery, swapping and unrestricting work, then comment the other rules.}

\paragraph{Capsule {recovery}} Let us discuss, first, when an expression can be safely typed $\capsule$. {The evaluation of the expression should produce a portion of store which is isolated, apart from external immutable references, formally a right-value where all free variables are immutable. Obviously, this is safe for an expression which has itself no free variables, or where all the free variables are {$\imm$ or $\capsule$}, and, indeed, this was the requirement needed for obtaining {recovery} in previous work \cite{GordonEtAl12}.} However, this requirement is too strong. 
Consider the following sequence of declarations:\label{capsule-example-2}

\begin{lstlisting}
mut D y=new D(0); 
capsule C z={mut D x=new D(y.f); new C(x,x)};  
\end{lstlisting}

The inner block (right-hand side of the declaration of \Q@z@) {can} be typed $\capsule$, even though there is a mutable free variable \Q@y@, since this variable is only used in a field access, hence, intuitively, no aliasing is introduced between \lstinline{y}{} and the final result of the block. 
Indeed, the block reduces\footnote{As formalized in \refToSection{calculus}.} to 
\begin{lstlisting} 
{mut D x=new D(0); new C(x,x)};  
\end{lstlisting}
which is closed.
%, that is, has no free variables. 
To allow such typing, the inner block is {typed} applying rule \rn{t-capsule}, since it can be typechecked $\mutable$ in a type context where variable \lstinline{y}{} is $\lent$. In \refToFigure{TypingOne} we show the type derivation for this example,
{where, for clarity, we always write the mutable group even when it is empty.}
%We use $\deriv$ as a metavariable to denote derivations.
\begin{figure}[t]
\framebox{
\begin{footnotesize}
\begin{math}
\begin{array}{l}
\\
\infer[\scriptstyle{\textsc{(t-capsule)}}]
{\TypeCheck{\Gamma}{\emptyset\ [\texttt{y}]}{\emptyset}{\texttt{\{mut D x=new D(y.f); new C(x,x)\}}}{\Type{\capsule}{\C}}}
{
\infer[\scriptstyle{\textsc{(t-block)}}]
{\TypeCheck{\Gamma}{\texttt{y}\ [\emptyset]}{\emptyset}{\texttt{\{mut D x=new D(y.f); new C(x,x)\}}}{\Type{\mutable}{\C}}}
{
%\infer[\scriptstyle{\textsc{(t-sub)}}]
%{ \TypeCheck{\Gamma'}{\texttt{y}}{\emptyset}{\texttt{new D(y.f)}}{\Type{\lent}{\texttt{D}}}      }
{
\infer[\scriptstyle{\textsc{(t-new)}}]
{\TypeCheck{\Gamma'}{\texttt{y}\ [\texttt{x}]}{\emptyset}{\texttt{new D(y.f)}}{\Type{\mutable}{\texttt{D}}}}
{
\infer[\scriptstyle{\textsc{(t-field-access)}}]
{\TypeCheck{\Gamma'}{\texttt{y}\ [\texttt{x}]}{\emptyset}{\texttt{y.f}}{\intType}}
{
\infer[\scriptstyle{\textsc{(T-var)}}]
{\TypeCheck{\Gamma'}{\texttt{y}\ [\texttt{x}]}{\emptyset}{\texttt{y}}{\Type{\lent}{\texttt{D}}}}{} 
}
}
}
&
\infer[\scriptstyle{\textsc{(T-New)}}]
{\TypeCheck{\Gamma'}{\texttt{y}\ [\texttt{x}]}{\emptyset}{\texttt{new C(x,x)}}{\Type{\mutable}{\texttt{C}}}}
{
\infer[\scriptstyle{\textsc{(T-var)}}]
{\TypeCheck{\Gamma'}{\texttt{y}\ [\texttt{x}]}{\emptyset}{\texttt{x}}{\Type{\mutable}{\texttt{D}}}}
{} 
}
}
}
\\
\\
\Gamma=\texttt{y}{:}\Type{\mutable}{\texttt{D}},\texttt{z}{:}\Type{\capsule}{\texttt{C}}\\
\Gamma'=\SubstFun{\Gamma}{\texttt{x}{:}\Type{\mutable}{\texttt{D}}}=\texttt{y}{:}\Type{\mutable}{\texttt{D}},\texttt{z}{:}\Type{\capsule}{\texttt{C}},\texttt{x}{:}\Type{\mutable}{\texttt{D}}
\end{array}
\end{math}
\end{footnotesize}
}
\caption{Example of type derivation (1)}\label{fig:TypingOne}
\end{figure}

{Note that in an analogous example where the field of class \lstinline{D} has a non primitive type, e.g., \lstinline{String}:}
\begin{lstlisting}
mut D y=new D("hello"); 
capsule C z={mut D x=new D(y.f); new C(x,x)};  
\end{lstlisting}
{the qualifier of the field should be $\imm$, since, otherwise, by introducing a singleton {lent group} \lstinline{y} we would get a $\lent$ type for \lstinline{y.f} as well, see rule \rn{t-field-access}, and a $\lent$ type is not accepted for a constructor argument.}

As a counterexample, consider the following sequence of declarations:
\begin{lstlisting}
mut D y=new D(0); 
capsule C z={mut D x=y; new C(x,x)};  
\end{lstlisting}
Here the inner block cannot be typed $\capsule$, since \Q@y@ is internally aliased. Indeed, the block reduces to \Q@new C(y,y)@ which contains a free mutable variable.
Formally, we cannot apply rule \rn{t-capsule} on the block, since
we should typecheck the block with \Q@y@ in a singleton {lent group}, while the initialization expression of \lstinline{x} should be mutable, see rule \rn{t-block}.

Rule \rn{t-field-assign} requires a mutable type for the receiver.
So, how is it possible to modify (the portion of store denoted by) a $\lent$ reference?
Consider the following simple example:
\begin{lstlisting}
lent D y= new D(0);
y.f=y.f+1;
\end{lstlisting}

This code should be well-typed, since the assignment does not introduce any alias. {To get such typing, we use rule \rn{t-swap} to type the expression
\Q@y.f=y.f+1@. {Indeed, we can} swap the singleton {lent group} \Q@y@ with the empty set.

Moreover, swapping can be applied to achieve {recovery}. Take
the example already considered at page \pageref{capsule-example-1}:
\begin{lstlisting}
mut D y=new D(0); 
capsule C z={mut D x=new D(y.f=y.f+1); new C(x,x)}
\end{lstlisting}
Let $\e$ be the inner block (right-hand side of the declaration of \Q@z@). As in the first example, $\e$ can be typed $\capsule$ if it can be typed \Q@mut@ in a context with type assignment $\texttt{y}{:}\Type{\mutable}{\texttt{D}},\texttt{z}{:}\Type{\capsule}{\texttt{C}}$ and
the {lent group} \Q@y@. However, the assignment \Q@y.f=y.f+1@
is not well-typed in this type context, since the variable \Q@y@ has type \Q@lent D@. However, intuitively, we can see that the assignment does not introduce any alias between \Q@y@ and the final result of $\e$, since it involves only variables which are in the same group (the
singleton \Q@y@), and produces {a result which is not mutable}. {In other words, the result of the evaluation of $\e$ is a capsule, as it has been shown in \refToFigure{esRed2}, so it should be possible to
{type} the expression $\e$ {as} $\capsule$.} 

To {get} such typing, we can apply rule \rn{t-swap} when deriving the type for the subexpression \Q@y.f=y.f+1@, by swapping the group \Q@y@ with the group \Q@x@. 
This ensures that the evaluation of the subexpression typed by this rule will not introduce any alias between the
variables in the swapped group and the mutable group. {The type derivation for the example is given in \refToFigure{TypingTwo}.}
 \begin{figure}[t]
 \framebox{
\begin{footnotesize}
\begin{math}
\begin{array}{l}
\infer[\scriptstyle{\textsc{(t-capsule)}}]
{\TypeCheck{\Gamma}{\emptyset\ [\texttt{y}]}{\emptyset}{\texttt{\{mut D x=new D(y.f=y.f+1); new C(x,x)\}}}{\Type{\capsule}{\texttt{C}}}}
{
\infer[\scriptstyle{\textsc{(t-block)}}]
{\TypeCheck{\Gamma}{\texttt{y}\ [\emptyset]}{\emptyset}{\texttt{\{mut D x=new D(y.f=y.f+1); new C(x,x)\}}}{\Type{\mutable}{\texttt{C}}}
}
{
\infer[\scriptstyle{\textsc{(T-New)}}]
{\TypeCheck{\Gamma'}{\texttt{y}\ [\texttt{x}]}{\emptyset}{\texttt{new C(x,x)}}{\Type{\mutable}{\texttt{C}}}}
{
\infer[\scriptstyle{\textsc{(T-var)}}]
{\TypeCheck{\Gamma'}{\texttt{y}\ [\texttt{x}]}{\emptyset}{\texttt{x}}{\Type{\mutable}{\texttt{D}}}}{}
}
&
\hspace{-1cm}
\infer[\scriptstyle{\textsc{(T-New)}}]
{\TypeCheck{\Gamma'}{\texttt{y}\ [\texttt{x}]}{\emptyset}{\texttt{new D(y.f=y.f+1)}}{\Type{\mutable}{\texttt{D}}}}
{
\infer[\scriptstyle{\textsc{(t-swap)}}]
{\TypeCheck{\Gamma'}{\texttt{y}\ [\texttt{x}]}{\emptyset}{\texttt{y.f=y.f+1}}{\intType}}
{
\infer[\scriptstyle{\textsc{(t-fld-ass)}}]
{\TypeCheck{\Gamma'}{\texttt{x}\ [\texttt{y}]}{\emptyset}{\texttt{y.f=y.f+1}}{\intType}}
{
\infer[\scriptstyle{\textsc{(T-var)}}]
{\TypeCheck{\Gamma'}{\texttt{x}\ [\texttt{y}]}{\emptyset}{\texttt{y}}{\Type{\mutable}{\texttt{D}}}}{}
&
\deduce{\TypeCheck{\Gamma'}{\texttt{x}\ [\texttt{y}]}{\emptyset}{\texttt{y.f+1}}{\intType}}{\vdots}
}
}
}
}
}\\
\\
\Gamma=\texttt{y}{:}\Type{\mutable}{\texttt{D}},\texttt{z}{:}\Type{\capsule}{\texttt{C}}\\
\Gamma'=\SubstFun{\Gamma}{\texttt{x}{:}\Type{\mutable}{\texttt{D}}}=\texttt{y}{:}\Type{\mutable}{\texttt{D}},\texttt{z}{:}\Type{\capsule}{\texttt{C}},\texttt{x}{:}\Type{\mutable}{\texttt{D}}
\end{array}
\end{math}
\end{footnotesize}
}
\caption{Example of type derivation (2)}\label{fig:TypingTwo}
\end{figure}

Note that, when using rule \rn{t-swap} to typecheck a subexpression of an expression {for which we want the capsule or immutability property}, no alias should be introduced between the variables in the group $\xs$ and the result of the expression.  Indeed, in this case the result of the subexpression could contain references to the variables in group $\xs$, which was lent in the original context.To ensure this, the type obtained in this way is weakened to $\lent$, if it was $\mutable$. 
 This is shown by the following example:
\begin{lstlisting}
mut D y=new D(x1,x2);  
mut x1=new A(0); 
mut x2=new A(1);
capsule C z={mut A x=(y.f1=y.f2); new C(x,x)};
\end{lstlisting}
If we apply rule \rn{{t-swap}} when deriving the
type for \Q@y.f1=y.f2@, therefore swapping the group \Q@y@ with \Q@x@, then we derive type
\Q@mut A@, and rule \rn{{t-swap}} would assign type \Q@lent A@ to the expression.
Therefore, the declaration \Q@mut A x=(y.f1=y.f2)@ and the whole expression would be ill-typed.
Indeed, the expression reduces to
\begin{lstlisting}
mut D y=new D(x2,x2);  
mut x1=new A(0); 
mut x2=new A(1);
capsule C z=new C(x2,x2);
\end{lstlisting}
in which the value of \Q@z@ is not a capsule.

\paragraph{{Immutability recovery}} 
{Note that $\capsule$ recovery can only happen for $\mutable$ expressions. In other words, $\mutable$ expressions which reduce to a portion of store with no external mutable references can be safely used where (either a mutable or) an immutable is required. Indeed, every expression which can be typed $\capsule$ can be typed $\imm$ by subtyping. }

{Consider now an expression for which $\capsule$ recovery cannot happen, that is, which can be typed $\lent$ or $\readable$, but cannot be typed $\mutable$, hence should not be used where a $\mutable$ is required. We can \emph{directly}\footnote{That is, not by subtyping.} recover immutability for such an expression, if we can guarantee that the result of the expression will be not connected to external mutable references.   This can be ensured as for the case of $\capsule$ recovery, with one difference. For $\capsule$ recovery, $\lent$ and $\readable$ references can be freely used, since in any case they will be not connected to the final result of the expression. However, if the expression is in turn $\lent$ or $\readable$, its result \emph{could} be connected to $\lent$ or $\readable$ references, hence this should be explicitly prevented by the type system. This is achieved by \emph{restricting} such references, that is, allowing their use only to typecheck subexpressions of $\imm$ $\capsule$ or primitive type.}

Consider the following variant of the first example of capsule {recovery}:
\begin{lstlisting}
mut D y=new D(0); 
imm C z={lent D x=new D(y.f); new C(x,x)};  
\end{lstlisting}
As in the original version, the inner block (right-hand side of the declaration of \lstinline{z}) uses the mutable free variable \lstinline{y} only in a field access, and indeed reduces to the block \Q@{lent D x= new D(0);  new C(x,x)}@ which is closed. However, this block {cannot be typed $\capsule$}, since it cannot be safely assigned to a mutable reference. On the other hand, the block can be safely typed  $\imm$, since, intuitively, it reduces to a portion of store which cannot be modified through any other reference. Formally, the inner block can be {typed $\imm$} by rule \rn{t-imm}, since it can be typechecked $\lent$ in a type context where variable \lstinline{y}{} is in  singleton {lent group} and, moreover, restricted, that is, can be only used to typecheck subexpressions which are $\imm$, $\capsule$, or of a primitive type. 
The type derivation for the example is given in \refToFigure{TypingThree}. 
{Note that, instead of putting \texttt{x} is a sigleton group we could have put \texttt{x} and \texttt{y} in the same 
group in the typing of \lstinline{new C(x,x)}{} and \lstinline{new D(y.f)}{}. That is, replacing}
\verb!{x} {y}! with \verb!{x,y}!  the derivation would still be correct. (Clearly the rule \rn{T-Swap} would swap \verb!{x,y}!
with the empty mutable group.)

\begin{figure}[th]
\framebox{
\begin{footnotesize}
\begin{math}
\begin{array}{l}
\\
\infer[\scriptstyle{\textsc{(t-imm)}}]
{\TypeCheck{\Gamma}{\emptyset\ [\texttt{y}]}{\emptyset}{\texttt{\{lent D x=new D(y.f); new C(x,x)\}}}{\Type{\imm}{\texttt{C}}}}
{
\infer[\scriptstyle{\textsc{(t-block)}}]
{\TypeCheck{\Gamma}{\texttt{y}\ [\emptyset]}{\texttt{y}}{\texttt{\{lent D x=new D(y.f); new C(x,x)\}}}{\Type{\lent}{\texttt{C}}}
}
{
\infer[\scriptstyle{\textsc{(t-new)}}]
{
\TypeCheck{\Gamma'}{\{\texttt{y}\}\ \{\texttt{x}\}\ [\emptyset]}{\texttt{y}}{\texttt{new D(y.f)}}{\Type{\mutable}{\texttt{D}}}
}
{
\infer[\scriptstyle{\textsc{(t-unrst)}}]
{\TypeCheck{\Gamma'}{\{\texttt{y}\}\ \{\texttt{x}\}\ [\emptyset]}{\texttt{y}}{\texttt{y.f}}{\intType}}
{
\infer[\scriptstyle{\textsc{(t-field-access)}}]
{\TypeCheck{\Gamma'}{\{\texttt{y}\}\ \{\texttt{x}\}\ [\emptyset]}{\emptyset}{\texttt{y.f}}{\intType}}
{
\infer[\scriptstyle{\textsc{(T-var)}}]
{\TypeCheck{\Gamma'}{\{\texttt{y}\}\ \{\texttt{x}\}\ [\emptyset]}{\emptyset}{\texttt{y}}{\Type{\lent}{\texttt{D}}}}{} 
}
}
}
&
\infer[\scriptstyle{\textsc{(t-swap)}}]
{\TypeCheck{\Gamma'}{\{\texttt{y}\}\ \{\texttt{x}\}\ [\emptyset]}{\texttt{y}}{\texttt{new C(x,x)}}{\Type{\lent}{\texttt{C}}}}
{
\infer[\scriptstyle{\textsc{(t-new)}}]
{\TypeCheck{\Gamma'}{\texttt{y}\ [\texttt{x}]}{\texttt{y}}{\texttt{new C(x,x)}}{\Type{\mutable}{\texttt{C}}}}
{
\infer[\scriptstyle{\textsc{(T-var)}}]
{\TypeCheck{\Gamma'}{\texttt{y}\ [\texttt{x}]}{\texttt{y}}{\texttt{x}}{\Type{\mutable}{\texttt{D}}}}{}
}
}
}
}
\\  \\
\Gamma=\texttt{y}{:}\Type{\mutable}{\texttt{D}},\texttt{z}{:}\Type{\imm}{\texttt{C}}\\
\Gamma'=\SubstFun{\Gamma}{\texttt{x}{:}\Type{\mutable}{\texttt{D}}}=\texttt{y}{:}\Type{\mutable}{\texttt{D}},\texttt{z}{:}\Type{\imm}{\texttt{C}},\texttt{x}{:}\Type{\mutable}{\texttt{D}}
\end{array}
\end{math}
\end{footnotesize}
}
\caption{Example of type derivation (3)}\label{fig:TypingThree}
\end{figure}

{Restricting $\y$} prevents typechecking examples like the following:
\begin{lstlisting}
mut D y=new D(0); 
imm C z={lent D x=y; new C(x,x)};  
\end{lstlisting}

The significance of the {immutability recovery} is more clearly shown by considering method calls, as will be illustrated in \refToSection{examples}.

\paragraph{Blocks}
{A block $\Block{\decs}{\e}$, where $\decs=\Dec{\T_1}{\x_1}{\e_1}\ldots\Dec{\T_n}{\x_n}{\e_n}$, is well-typed if the right-hand sides of declarations and the body are well-typed, as detailed below.
\begin{itemize}
\item All the expressions are typechecked w.r.t.\ the type assignment $\SubstFun{\Gamma}{\TypeEnv{\decs}}$ where $\TypeEnv{\decs}$  is the same of $\TypeDec{\T_1}{\x_1},\ldots,\TypeDec{\T_n}{\x_n}$, apart that local variables declared $\lent$ have type $\mutable$ (indeed the fact that they are $\lent$ is encoded by  including them in lent groups, see next point).
\item The body is typechecked w.r.t.\ lent groups $\LentLocked'$ and mutable group $\MutGroup'$ which extend those of the enclosing scope, modulo hiding (second and third side conditions). More precisely: variables which are $\mutable$ in $\TypeEnv{\decs}$ can be possibly added to a lent group of the enclosing scope, or can form a new lent group, or can be added to the mutable group $\xs$ of the enclosing scope.
\item The right-hand side of each declaration $\e_i$ is typechecked w.r.t.\ to lent groups $\LentLocked_i$ and mutable group $\xs_i$ obtained from those of the body by swapping $\xs'$ with the group which contains $\x_i$, if $\x_i$ is $\mutable$ in $\TypeEnv{\decs}$ (fourth side condition). Recall that sequences of groups are considered as sets, hence the notation $\LentLocked_i\ \MutGroup_i = \LentLocked'\ \MutGroup'$ means that the two sides are the same set of groups.  This swapping models the fact that the initialization expression of a variable $\x_i$ in a lent group is typechecked considering as mutable group that containing $\x_i$.
For variables in the mutable group, or declared with other modifiers, no swapping is needed.
\item All the expressions are typechecked w.r.t.\ restricted variables which are exactly those of the enclosing scope (modulo hiding).
\end{itemize}
We use the following auxiliary notations.
\begin{itemize}
\item {The {\em type assignment extracted} from a sequence of declarations $\decs$, denoted $\TypeEnv{\decs}$, is defined by: $\TypeEnv{\Dec{\T_1}{\x_1}{\e_1}\ldots\Dec{\T_n}{\x_n}{\e_n}}=\TypeDec{\T'_1}{\x_1},\ldots,\TypeDec{\T'_n}{\x_n}$ where $\T'_i=\Type{\mutable}{\C}$ if  $\T_i=\Type{\lent}{\C}$, $\T'_i=\T_i$ otherwise.}
\item $(\xs_1\ldots\xs_n){\setminus}\xs=(\xs_1{\setminus}\xs)\ldots(\xs_n{\setminus}\xs)$  
\item $\dom{\xs_1\ldots\xs_n}=\{\x\mid \x\in\xs_i\mbox{ for some}\ i\}$
\item $\LessEq{\xs_1\ldots\xs_n}{\ys_1\ldots\ys_m}$ if $\dom{\xs_1\ldots\xs_n}\subseteq\dom{\ys_1\ldots\ys_m}$, and {for all $\x,\y\in\dom{\xs_1\ldots\xs_n}$, $\x,\y\in\xs_i$ if and only if 
$\x,\y\in\ys_j$.}
%\item $\LessEq{\LentLocked}{\LentLocked'}$ if $\dom{\LentLocked}\subseteq\dom{\LentLocked'}$, and {for all $\x,\y\in\dom{\LentLocked}$, $\x\in\xs\in\LentLocked$  and $\y\in\xs\in\LentLocked$ if and only if 
%$\x\in\xs'\in\LentLocked'$ and $\y\in\xs'\in\LentLocked'$.}
%\item $\Extends{\LentLocked}{\zs}{\LentLocked'}$ if $\LessEq{\LentLocked}{\LentLocked'}$ holds, and, moreover, $\dom{\LentLocked'}=\dom{\LentLocked}\cup\zs$.
\end{itemize} }

{Note that local variables declared $\lent$ can be arbitrarily assigned to {lent group}s, to improve expressivity. For instance, it can be necessary to assign a $\lent$ local variable to the same {lent group} {as} some variables of the enclosing scope. This is shown by the following example, where the local variable \lstinline{z1} is used to modify the external reference \lstinline{x}, rather than to construct the block result.}
\begin{lstlisting}
mut D z=new D(0);  
mut C x=new C(z,z);
capsule C y= {
  lent D z1=new D(1); 
  lent D z2=(x.f1=z1);  
  mut D z3 = new D(2);
  new C(z3,z3)
};
\end{lstlisting}
Since 
we need to {recover the capsule property for} the
block on the right-hand side of the declaration of \Q@y@, applying rule \rn{T-capsule} to the block, the context of the typing of such block 
must be \\
{\small $\ \ \ \ {\tt z:\mutable\ D,x:\mutable\ C,y:\capsule\ C; {\{z, x\}}; \emptyset}$}\\
that is, the variables
\Q@z@ and \Q@x@ must be in the same {lent group}. However, assuming that field \lstinline{f1}{} is mutable, to apply rule \rn{t-field-assign} to the expression \Q@x.f1=z1@,  both \Q@x@ and \Q@z1@ have to be mutable. Therefore, we have to apply rule
\rn{t-swap}, and it must be the case that  \Q@x@ and \Q@z1@ are in the same {lent group}. This is possible, with rule \rn{t-block}, by adding \Q@z1@ to the group \verb!{z,x}! in typing the right-hand sides of the declarations and the body.

Note that the following variant of the example, where the result of the block is connected to \lstinline{z1} instead,}
\begin{lstlisting}
mut D z=new D(0);  
mut C x=new C(z,z);
capsule D y= {
  lent D z1=new D(1); 
  lent D z2=(x.f1=z1);  
  new C(z1,z1)
};
\end{lstlisting}
is ill-typed.  Indeed, as before, \Q@z1@ must be in the same group {as} \Q@x@ in order to {recover the $\capsule$ property of} the block, but in this case \Q@z1@ would be $\lent$, hence the whole block.

\paragraph{Other typing rules} 
Other rules are mostly standard, except that they model the expected behaviour of type qualifiers.

{In} rule \rn{t-var}, a variable is weakened to $\lent$ if it belongs to some group in $\LentLocked$, and cannot be used at all if it belongs to $\StronglyLocked$. 

In rule \rn{t-field-access}, in case the field is $\mutable$, the type qualifier of the receiver is propagated to the field. For instance, {mutable} fields referred through a $\lent$ reference are $\lent$ as well. If the field is immutable {(or of a primitive type), instead, then the expression has the field type,} regardless of the receiver type.  

In rule \rn{t-field-assign}, the receiver should be mutable, and the right-hand side must have the field type. Note that this implies the right-hand side to be either $\mutable$ or $\imm$ {(or of a primitive type)}. Hence, neither the left-hand nor the right-hand sides can be $\lent$ or $\readable$. {This prevents the introduction of aliases for such references. However, recall that for $\lent$ references the constraint can be escaped by using, before this rule, rule \rn{t-swap}, at the price of weakening to $\lent$ the type of the expression. }
  
In rule \rn{t-new}, expressions assigned to fields should be either $\mutable$ or $\imm$ (or of a primitive type).  {Again, for $\lent$ references the constraint can be escaped by swapping, getting a $\lent$ expression. }
 Note that an object is created with no restrictions, that is, as $\mutable$.

Finally, note that primitive types are used in the standard way. For instance, in the premise of rule \rn{t-new} the types of constructor arguments could be primitive types, whereas in rule \rn{t-meth-call} the type of receiver could not.

%% file: Examples.tex
\section{Programming examples}\label{sect:examples}

In this section we illustrate the {expressivity of the type system by more significant} examples. 
For sake of readability, we use additional constructs, such as {operators of} primitive types, static methods{, private fields,} and while loops. Moreover, we generally omit the
brackets of the outermost block, and abbreviate $\Block{\Dec{\T}{\x}{\e}}{\e'}$ by $\Sequence{\e}{\e'}$ when $\x\not\in\FV{\e'}$, with $\FV{\e}$ the set of the free variables
 of $\e$. 
  
\paragraph{Capsules and swapping}
{The following example illustrates} how a programmer can declare lent references to achieve {recovery}.
The class \Q@CustomerReader@ below models reading information about customers out of a text file formatted as shown in the example:

\begin{lstlisting}
Bob
1 500 2 1300
Mark
42 8 99 100
\end{lstlisting}

In even lines we have customer names, in odd lines we have a shop history: a sequence of product codes.

\begin{lstlisting}
class CustomerReader {
  static capsule Customer readCustomer(lent Scanner s){
      mut Customer c=new Customer(s.nextLine())
      while(s.hasNextNum()){
        c.addShopHistory(s.nextNum())
      }
      return c //ok, capsule recovery
  }
}
\end{lstlisting}
The method \Q@CustomerReader.readCustomer@ takes a \Q@lent Scanner@, assumed to be a class similar to the one in Java,
for reading a file and extracting different kinds of data.
A \Q@Customer@ object is {created reading its name} from the file, and then its shop history is added.
Since the scanner is declared $\lent$, and there are no other parameters, by  {recovery} the result can be declared $\capsule$. {That is, we can statically ensure that the data of the scanner are not mixed with the result.}
Previous work offers cumbersome solutions requiring the programmer to manually handle multiple initialization phases like ``raw'' and ``cooked''~\cite{ZibinEtAl10}.

{The following method \Q@update@ illustrates how we can ``open'' capsules, modify their values and then recover the original \Q@capsule@ guarantee. The} 
method  takes an old customer (as capsule) and a \Q@lent Scanner@ as before.
{%\scriptsize
\begin{lstlisting}
class CustomerReader {...//as before
  static capsule Customer update(capsule Customer old,
                                  lent Scanner s){
    mut Customer c=old//we open the capsule `old'
    while(s.hasNextNum()){
      c.addShopHistory(s.nextNum());
    }recovery
    return c; //ok, capsule recovery
  }
}
\end{lstlisting}
}
Every method {with no} mutable parameters can use the pattern illustrated above: one (or many) capsule parameters are opened  (that is, assigned to mutable local variables) and, in the end, the result is guaranteed to be again a capsule. {This mechanism is not possible in previous work \cite{Almeida97,ClarkeWrigstad03,DietlEtAl07}. The type systems of Gordon et al.~\cite{GordonEtAl12} and Pony \cite{ClebschEtAl15} lack borrowed references in the formalization, but could support the variant where the scanner is isolated ($\capsule$ in our teminology) without losing isolation of the scanner or customer (the
 scanner could remain isolated throughout, relying on their support for dispatching appropriate
 methods on isolated references without explicit unpacking)}.

{In the former example, explicit \capsule\ opening and recovery take place \emph{inside the method body}.
Consider the following alternative version:}
\begin{lstlisting}
class CustomerReader {...//as before
  static mut Customer update(mut Customer c,lent Scanner s){
    while(s.hasNextNum()){
      c.addShopHistory(s.nextNum());
    }
    return c;
  }
}
\end{lstlisting}
{This method {takes now a \mutable\ object and returns} another \mutable\ one, as it usually happens in imperative programming.
Surprisingly enough, this method turns out to be just a more flexible version of the former one. Indeed:
\begin{itemize}
\item If called on \mutable\ data, then it returns \mutable\ data, while a call of the former method was ill-typed.
\item If called on \capsule\ data, then capsule opening implicitly takes place during method call execution, and recovery can be achieved for the method call expression, returning a \capsule\ as for the former method.
\end{itemize}
That is, recovery happens \emph{on the client side}\footnote{{The type systems  of Gordon et al.~\cite{GordonEtAl12} and Pony \cite{ClebschEtAl15} could support a variant analogously to the case above.}}. However, a client does not need to worry about this, and can simply call the method.
In a sense, this version of \Q@update@ is polymorphic: it can be used as either
$\mutable\rightarrow\mutable$ or as $\capsule\rightarrow\capsule$.}

{This pattern can be used in combination with function composition. That is, a chain of $\mutable\rightarrow\mutable$ methods can be called, and, if we start from a \capsule, we will get a \capsule\ in the end.
As shown in our example, these methods can have additional (non \mutable) parameters.}

{Moreover, this method can be trasparently used as $\lent\rightarrow\lent$. That is, if called on  $\lent$ data, then it returns $\lent$ data, by applying rule \rn{t-swap} to the method call expression. Again this can be extended to chains of methods which may have additional (non \mutable) parameters. }

Finally, we show the code of the \Q@Scanner@ itself, and how swapping can be used to update {the
fields of a $\lent$ scanner in a safe way.}

\begin{lstlisting}
class Scanner{
  mut InputStream stream;
  imm String nextLine(mut/*=this qualifier*/){...}
  int nextNum(mut/*=this qualifier*/){...}
  bool hasNextNum(read/*=this qualifier*/){...}
}

lent Scanner s=...
mut InputStream stream1=...
capsule InputStream stream2 = ... 
//s.stream=stream1 //(a)wrong
s.stream=new InputStream("...")//(b)ok, swapping 
s.stream=stream2 //(c)ok, swapping 
\end{lstlisting}

In our type system,  a $\lent$ reference can be regarded as mutable if all the mutable references are regarded as $\lent$, as formally modeled by rule \rn{t-swap}.
% {This mechanism is similar to \emph{viewpoint adaptation}
%as in \cite{DietlEtAl07}. 
This can be trivially applied to line (b), where \Q@s@ is the only free variable, and to  line (c), where the other free variable is declared $\capsule$. 
In line (a), instead, swapping would assign a $\lent$ type to \lstinline{stream1}.

{The $\this$ qualifiers reflect the fact that the first two methods modify, whereas the third only reads, the scanner's state. Note that, as in the previous example, the first two methods can be safely applied to $\lent$ scanners by swapping (in this case the result type, being not mutable, remains the same). Note also that such methods (as any method with no parameters and non mutable result type) could be equivalently have $\this$ qualifier $\lent$, since, intuitively, no aliasing can be introduced for $\this$ (formally, we can apply rule \rn{t-swap} to the method body). On the other hand, the first two methods can be invoked on $\lent$ scanners by by applying rule \rn{t-swap} to the method call expression.}

\paragraph{Readable and immutable references}
We illustrate now the \Q@read@ and \Q@imm@ qualifiers by the example of a \Q@Person@ class with a list of friends.
\begin{lstlisting}
class Person{  
  private mut PersonList friends;  
  read PersonList readFriends (read/*=this qualifier*/) {
    return this.friends;
  }
}
\end{lstlisting}
To give access to the private field, we declare a method {which is like a usual getter, except} that it returns a \Q@read PersonList@ reference. 
In this way, a client can only read the list of friends obtained through an invocation \lstinline{person.readFriends()}{}, with \lstinline{person} of class \lstinline{Person}.
Note that, since the $\this$ qualifier of the method is \lstinline{read}{}, which is the top of the subtyping hierarchy, it can be invoked whichever is the qualifier of \lstinline{person}.
Moreover, in the case \lstinline{person} is \lstinline{imm}, we get an \Q@imm PersonList@ back, by {recovering immutability for the \Q@read@ result}.
Indeed, in this case, we can apply rule \rn{t-imm} to the invocation \lstinline{person.readFriends()}, since the only free variable \Q@person@ is immutable, so no variable needs to be {restricted}.
{This is another case where the method is polymorphic: it can be used as either
$\readable\rightarrow\readable$ or as $\imm\rightarrow\imm$, and a client does not need to worry about, and can simply call the method.}

Assume now that we want, instead, to give permission to a client to modify the list of friends.
In this case, we can declare a getter method with different type annotations:\label{exposer}
\begin{lstlisting}
mut PersonList getFriends(mut/*=this qualifier*/) {
  return this.friends;
}
\end{lstlisting}
This method takes a mutable \Q@Person@ and returns a mutable \Q@PersonList@ reference.
Hence, it cannot be invoked on \Q@read@ or \Q@imm@ objects.
However, this $\mutable$ getter can be invoked on a lent \lstinline{person}{}:
in \Q@person.getFriends()@  the only free variable \Q@person@ can be seen as \Q@mut@.
The result of the method is turned in \Q@lent PersonList@ by the \rn{t-swap} rule {(formally, swapping the singleton {lent} group \Q@person@ with the empty set).}

{Our approach forces the programmer to explicitly choose either $\readable$ or $\mutable$ getters.}
Other works permits the developer to use
either multiple getters or polymorphic type qualifiers, for instance
the type annotation \lstinline{PolyRead} of Javari~\cite{TschantzErnst05} allows one
to keep a single method, providing an additional design choice for programmers.
However, we argue that forcing programmers to consider the two operations as logically different  can be a 
simpler and more explicit programming pattern.
 (In \refToSection{related} we discuss also a third variant, with return type $\lent$.)
%The getter only gives access to the field, so that the data can be read by a client.
%The exposer, instead, gives permission to a client to modify the reachable object graph of the receiver by returning a mutable reference to the field.
In a language supporting many levels of visibility (like protected, package, friend, ...) a programmer may choose a restricted visibility for $\mutable$ getters and a more permissive visibility for $\readable$ getters.
Collection classes also can declare $\readable$ and $\mutable$ getters, as in the following example.
\begin{lstlisting}
class PersonList{...
 void add(mut, mut Person elem){...}
 read Person readElem (read, int index){...}
 mut Person getElem(mut, int index){...}
}
\end{lstlisting}
Finally, we show how we can create mutable objects, mutate them for a while, and then {recover their immutabiity}:
\begin{lstlisting}
class C{
  static mut Person lonelyPersonFactory(){
    return new Person(new PersonList());
  }
  static imm Person happyPersonFactory(){
    mut Person fred=C.lonelyPersonFactory();
    mut Person barney=C.lonelyPersonFactory();
    fred.getFriends().add(barney);
    barney.getFriends().add(fred);
    return fred; //mut Person recovered to be imm 
    //now fred and barney are friends forever!
  }
}
\end{lstlisting}
Here \Q@lonelyPersonFactory()@  creates lonely \Q@Person@s, that have no friends.
However, there is still hope, since they are mutable:
\Q@happyPersonFactory@ creates two lonely people, \Q@fred@ and \Q@barney@, and makes them friends.
Then the function returns \Q@fred@ (and, indirectly, also \Q@barney@ that is now in the reachable object graph of \Q@fred@).
The function body does not use any free variable, so we can {recover the capsule property of} the result, hence return it as \Q@imm@.

\paragraph{Qualifiers are deep} Note that recovery work since qualifiers have a \emph{deep/full} interpretation, in the sense that
they are propagated to the whole reachable object graph. In
a shallow interpretation, instead, it is possible, for instance,
to reach a mutable object from an immutable object. The
advantage of such interpretation is that libraries can declare strong intentions in a coherent and uniform way, independently of the concrete representation of the user input
(that, with the use of interfaces, could be unknown to the
library). On the other side, providing (only) deep/full qualifiers
means that we do not offer any language support for (as an
example) an immutable list of mutable objects.

\paragraph{Nested recovery and groups}
We conclude the section by an example showing that {recoveries} can be nested, {and, to ensure that all are safe, distinct lent groups must be kept.}
Consider the following code, where implementation of \lstinline{A}{} is omitted to emphasize that only type information provided by qualifiers is significant.
%\newpage
\begin{lstlisting}
class A{...
  mut A mix(mut, mut A a){...}
  //inserts a  in the reachable object graph 
    //of the receiver and returns a
  capsule A clone (read){...} 
  //returns a capsule clone of the receiver
  static mut A parse(){...} //reads an A from input
}

mut A a1= A.parse() //outside of recovery
capsule A outerA={//outer recovery
  mut A a2= A.parse()//inside outer recovery
  capsule A nestedA={//nested recovery
    mut A a3= A.parse()//inside nested recovery
    mut A res= ...
    res.mix(a3)
    //this is promoted and assigned to nestedA
  }
  nestedA.mix(a2)
  //this is promoted and assigned to outerA
}
//program continues, using outerA as capsule
\end{lstlisting}

The question is, what can we write instead of the dots, and why.
Clearly, (1) \Q@a3@ is allowed, {since the result of the inner block will be only connected to a local reference}, while
(2) \Q@a1@ and \Q@a2@ are not, since it will be connected to an external mutable reference. 
However, (3) \Q@a1.clone()@ and \Q@a2.clone()@ are allowed, since the method \lstinline{clone} returns a capsule.
In the same way,
(4) \Q@a2.mix(a2).clone()@ is allowed, as well as
 \Q@a1.mix(a1).clone()@.

However, when we start mixing different variables, things become trickier.
 For example, (5) \lstinline{a2.mix(a1).clone()}{} is not well-typed in our type system.
  Indeed, even though, thanks to cloning, mixing \Q@a2@
 and \Q@a1@ does not compromise the capsule well-formedness of \lstinline{nestedA}{} (that is, the nested recovery can be safely applied), the fact that \Q@a2@ and 
 \Q@a1@ are mixed could compromise the capsule well-formedness of \Q@outerA@ when \Q@outerA@ is computed (that is, the outer recovery would be unsafe).

In summary, mixing \Q@a@s {lent groups introduced} for different {recoveries} must be avoided.
Rule \rn{t-swap} swaps one group with another,
thus keeping them distinct.

This last example \label{comparison} illustrates many of the differences w.r.t. {the type system proposed in \cite{GordonEtAl12}, whose notion of \emph{recovery} is less expressive}.
Their system allows (1), and rejects (2) and (5), as ours. However, they conservatively rejects (3) and (4), since 
 the flow is not tracked at a fine enough granularity.
Depending on the concrete application, programmers may need to work around the
limitations of \cite{GordonEtAl12} by reordering local variables,
 introducing stricter type qualifiers or, in general, re-factoring their code.
 However, there may be cases where there is no possible reordering.

%% file: Calculus.tex
\section{Calculus}\label{sect:calculus}
{The calculus has a simplified syntax where we assume that, except from right-hand sides of declarations, subterms of a compound expression are only {values}. This simplification can be easily obtained by a (type-driven) translation of the syntax of \refToFigure{syntax} generating for each subterm {which is not a value} a local declaration of the appropriate type.} 

Simplified syntax and values are defined in \refToFigure{calculus}. 
\begin{figure*}
\framebox{
{\small \begin{grammatica}
\produzione{\e}{\x\mid\FieldAccess{{\val}}{\f}\mid\MethCall{{\val}}{\m}{{\vals}}\mid\FieldAssign{{\val}}{\f}{{\val}}
\mid\ConstrCall{\C}{{\vals}}\mid\Block{\decs}{\val}
}{expression}\\*
\produzione{\dec}{\Dec{\T}{\x}{\e}}{variable declaration}\\*
\\
\produzione{\T}{\Type{\mu}{\C}}{type}\\*
\produzione{\mu}{\mutable\mid\imm\mid\capsule\mid\lent\mid\readable}{type qualifier}\\*
\\
\produzione{{\valPrime,}\val}{{\x\mid\ConstrCall{\C}{\vals}\mid\Block{\dvs}{\val}}}{{value}}\\*
\\
\produzione{\dv}{\Dec{\T}{\x}{\stVal}}{evaluated declaration}\\
%\produzione{{\valPrime,}\val}{{\x\mid\ConstrCall{\C}{\vals}\mid\Block{\dvs}{\val}}}{value}\\*
\produzione{\stVal}{\ConstrCall{\C}{\xs}\mid\Block{\dvs}{\x}\mid\Block{\dvs}{\ConstrCall{\C}{\xs}}}{\storableVal}
%\produzione{\bv}{}{block value}\\
\end{grammatica} }
}
\caption{Simplified syntax and values}\label{fig:calculus}
\end{figure*}
{Syntax of evaluated declarations and right-values from \refToFigure{syntax} is reported for reader's convenience.}

A {\it value} is either a variable (reference), {or a constructor invocation}, or a {value} enclosed in a local store.

{We denote by $\FV{\e}$ the set of the free variables of expression $\e$ (the standard formal definition is in Appendix A).}

We assume the following well-formedness constraints on expressions:
\begin{enumerate}
\item In a block $\Block{\Dec{\T_1}{\x_1}{\e_1}\ldots\Dec{\T_n}{\x_n}{\e_n}}{\val}$, mutual recursion, that is, $\x_j\in\FV{\e_i}$ and $\x_i\in\FV{\e_j}$, is only allowed if both declarations are {evaluated declarations, formally defined in \refToFigure{calculus}}. Hence, as expected, cyclic stores are allowed, such as
\begin{lstlisting}
{mut C y=new C(x); mut C x=new C(y); x}
{mut C x=new C(x); x}
\end{lstlisting}
whereas cyclic expressions such as
\begin{lstlisting}
{mut C y=x.f; mut C x=new C(y); x}, 
{mut C x= x.f; x}
{mut C x=y; mut C y = x; x}
\end{lstlisting}
are ill-formed. Allowing general recursion would require a sophisticated type system, 
as in \cite{ServettoEtAl13}, but this is not the focus of this paper.
\item As already mentioned in \refToSection{typesystem}, variables of $\capsule$ types can occur at most once in their scope. \label{linearity}
This is simply a \emph{syntactic} constraint, that is, we do not 
deal with linear types and the resultant context splitting, or flow-sensitive typing judgments. For instance, the following expression, which clearly
violates the intent of the capsule and immutable qualifiers, is ill-formed:
\begin{lstlisting}
capsule C c= {new C(0)}; 
capsule D d1= {new D(c)}; 
imm D d2 = {new D(c)}; 
...
\end{lstlisting}
Note that a capsule variable is not yet determined to
become mutable or immutable when it is declared, and that determination is made at the time of its unique occurrence.
\end{enumerate}

{Evaluated declarations associate a \emph{right-value} to a variable, which plays the role of a \emph{reference}.
{Hence, they correspond to} the \emph{store} in conventional models of imperative languages.}
%{We recall that a \emph{store} is a sequence $\dvs$ of evaluated declarations, since it plays the role of the store in conventional models of imperative languages. Indeed, each $\dv$ associates a right-}value to a variable, which plays the role of a \emph{reference}. 
In \refToFigure{wellformed} we define {\em {well-formed} right-values and store}. 
\begin{figure}[ht]
\framebox{
\begin{small}
\begin{math}
\begin{array}{c}
\Rule{}{\WFrv{\ConstrCall{\C}{\xs}}}{}\Space\Space
\Rule
{\forall \Dec{\T}{\y}{\stVal}\in\dvs\ \ \WFrv{\stVal}\ }
{{\WFrv{\Block{\dvs}{{\val}}}}}
{\begin{array}{l}
\dvs\neq\epsilon\\
%\cOrx=\x\ \vee\ \cOrx=\ConstrCall{\C}{\xs}\\ 
\dvs=\Reduct{\dvs}{{\FV{{\val}}}}
\end{array}}
%\Rule{\forall \Dec{\T}{\y}{\stVal}\in\dvs\ \ \WFrv{\stVal}\ }{\WFrv{\Block{\dvs}{\x}}}{\begin{array}{l}\dvs\neq\epsilon\\ \dvs=\Reduct{\dvs}{\x}\end{array}}
%\Space\Space\Rule{\forall \Dec{\T}{\y}{\stVal}\in\dvs\ \ \WFrv{\stVal}\ }{\WFrv{\Block{\dvs}{\ConstrCall{\C}{\xs}}}}{\begin{array}{l}\dvs\neq\epsilon\\ \dvs=\Reduct{\dvs}{\xs}\end{array}}\\
\\ \\
\Rule{}{\WFdv{\Dec{\Type{\mu}{\C}}{\x}{\ConstrCall{\C}{\xs}}}}{\mu\neq\capsule}
\Space\Space\Space\Space\Space
\Rule{\WFrv{\Block{\dvs}{{\val}}}\ \ \WFdv{\dvs}}
{\WFdv{\Dec{\Type{\imm}{\C}}{\x}{\Block{\dvs}{{\val}}}}}
{\begin{array}{c}
%\dvs=\Reduct{\dvs}{\y}{\neq\emptyDvs}\\
{\allMut{\dvs}}
\end{array}
}
%\\
%\\
%\Rule{\WFdv{\dvs}}{\WFdv{\Dec{\Type{\imm}{\C}}{\x}{\Block{\dvs}{\ConstrCall{\C}{\xs}}}}}
%{\begin{array}{c}
%%\dvs=\Reduct{\dvs}{\xs}{\neq\emptyDvs}\\
%{\allMut{\dvs}}
%\end{array}
%}
\\ \\
\Rule{\WFdv{\dv}\Space\Space\WFdv{\dvs}}{\WFdv{\dv\ \dvs}}{}
\end{array}
\end{math}
\end{small}
}
\caption{Well-formed right-values and store}\label{fig:wellformed}
\end{figure}

{For a sequence of declarations $\decs=\Dec{\T_1}{\x_1}{\e_1}\ldots\Dec{\T_n}{\x_n}{\e_n}$, and a set of variables $\X$, we write $\connected{\decs}{\X}{\x}$ if $\x$ is (transitively) connected to $\X$ through $\decs$. This relation  is inductively defined as follows:
\begin{center}
$\connected{\decs}{\X}{\x}$ if $\x\in\X$\\
$\connected{\decs}{\X}{\x}$ if $\x\in\FV{\e_i}$, for some $i\in 1..n$, and $\connected{\decs}{\X}{\x_i}$.
\end{center}
The subsequence $\Reduct{\decs}{\X}$ of the declarations that are (transitively) used by $\X$ is defined by: for all $i\in 1..n$, $\Dec{\T_i}{\x_i}{\e_i}\in\Reduct{\decs}{\X}$ if  $\connected{\decs}{\X}{\x_i}$.
}
\\
We write $\allMut{\dvs}$ if, for each $\Dec{\Type{\mu}{\C}}{\x}{{\stVal}}\in\dvs$, $\mu\geq\mutable$, and $\allImm{\dvs}$ if, for each $\Dec{\Type{\mu}{\C}}{\x}{{\stVal}}\in\dvs$, $\mu=\imm$. 

Rules in the first line define well-formed right-values. 
They state that a right-value should not contain garbage and useless blocks, and that, in case it is a block, 
{the right-hand sides of declarations are, in turn, well-formed right-values.}

Rules in the second and third line define {well-formed}  evaluated declarations or stores. The first rule {states} that
a right-value which is an object state can be associated to any reference which is not $\capsule$.
The second rule {states} that a right-value which is a block can only be associated to an $\imm$ reference,
its local store should not contain $\imm$ references, and, in addition to be a well-formed
right-value, the block should contain a well-formed local store.    
The last rule states that a (non empty) sequence of evaluated declarations is a well-formed store if each one is a well-formed. 

These rules {altogether} imply that in a well-formed store:
\begin{itemize}
\item There are no $\capsule$ references. Indeed, $\capsule$ declarations are ``temporary'', that is, are expected to be consumed as soon as their right-hand side has been evaluated, as will be formalized by reduction rule \rn{capsule-elim} in \refToFigure{reduction}.
\item  There are at most \emph{two} levels, that is, the right-value of an $\imm$ reference can contain in turn a local store of non $\imm$ references. Indeed, additional levels can safely be ``flattened'', as will be formalized by reduction rules \rn{mut-move} and \rn{imm-move} in \refToFigure{reduction}.
\end{itemize}
For instance, assuming that class \lstinline{C} has two $\mutable$  fields of class \lstinline{D} and one $\imm$ of class \lstinline{C}{}, and class  \lstinline{D} has an integer field, the following is a store:
\begin{small}
\begin{lstlisting}
mut D x = new D(0);
imm D y = new D(1);
imm C z = { 
  mut D x = new D(0);
  mut D y = new D(1);
  new C(x,y,z) 
  };
\end{lstlisting}
\end{small}
Note that mutable variables in the local store of \lstinline{z}{} are not visible from the outside. 
This models in a natural way the fact that the portion of store denoted by  \lstinline{z}{} is indeed immutable

Expressions are equivalent modulo the congruence $\cong$ defined by the rules of \refToFigure{congruence}. 
\begin{center}
\begin{figure}[ht]
\framebox{
\begin{small}
\begin{math}
\begin{array}{l}
\\
\Space\Space\Space{\NamedRuleOL{alpha}{\congruence{\Block{\decs\ \Dec{\T}{\x}{\e}\ \decs'}{\val}}{{\Block{\Subst{\decs}{\y}{\x}\ \Dec{\T}{\y}{\Subst{\e}{\y}{\x}}\ \Subst{\decs'}{\y}{\x}}{\Subst{\val}{\y}{\x}}}}}{}}\Space\Space\Space
\\ \\
\Space\Space\Space\Space\Space\NamedRuleOL{reorder}{\congruence{
\Block{
\decs\ {\dv}\ \decs'
}{\val}
}{
\Block{
{\dv}\ \decs\ \decs'\ 
}{\val}
} }{}\Space\Space\Space
\\
\end{array}
\end{math}
\end{small}
}
\caption{Congruence on expressions}
\label{fig:congruence}
\end{figure}
\end{center}
Rule \rn{alpha} is the usual $\alpha$-conversion. We write $\Subst{\e}{\e'}{\x}$
for the expression obtained by replacing all (free)
occurrences of $\x$ in $\e$ by $\e'$ {(the standard formal definition is in the Appendix)}. The condition $\x,\y\not\in\dom{\decs\,\decs'}$ is implicit by well-formedness of blocks. Rule \rn{reorder} states that we can move evaluated declarations in an arbitrary order.  
On the other hand, $\decs$ and $\decs'$ cannot be swapped, since this could change 
the order of side effects. 

Values are equivalent modulo the congruence $\cong$ defined by the rules of \refToFigure{congruenceVal}. 
\begin{figure}[ht]
\framebox{
\begin{small}
\begin{math}
\begin{array}{l}
\NamedRuleOL{new}
{
\begin{array}{l}
{\ConstrCall{\C}{
  \val_1,.., \val_{i-1},\val_i,\val_{i+1},..,\val_n}}\\
  \cong{\Block{\Dec{\T_i}{\x}{{\val_i}}}{\ConstrCall{\C}{\val_1,.., \val_{i-1},\x,\val_{i+1},..,\val_n}}}
\end{array}
}
{\begin{array}{l}
\fields{\C}{=}\Field{\T_1}{\f_1}..\Field{\T_n}{\f_n}\\
\notRef{\val_i}\\
\x\not\in\FV{\val_j}\ \ (1\leq j\leq n)
\end{array}}
\\[3ex]
\NamedRuleOL{body}{\congruence{\Block{\dvs}{\Block{\dvs'\ \dvs''}{\val}}}{\Block{\dvs\ \dvs'}{\Block{\dvs''}{\val}}}}
{\begin{array}{l}
{\noCapture{\dvs}{\dom{\dvs'}}}\\
%\FV{\dvs'}\cap\dom{\dvs}=\emptyset\\
{\noCapture{\dvs'}{\dom{\dvs''}}}
%\FV{\dvs}\cap\dom{\decs}=\emptyset
\end{array}
}
\\[3ex]
{\NamedRuleOL{garbage}{\congruence{\Block{\dvs\ \dvs' }{\val}}{\Block{\dvs'}{{\val}}}}
{\noCapture{\Block{\dvs'}{\val}}{\dom{\dvs}}}}
%\FV{\Block{\decs}{\val}}\cap\dom{\dvs}=\emptyset}
\\[3ex]
\NamedRuleOL{block-elim}{\congruence{\Block{}{{\val}}}{\val}}{}
\end{array}
\end{math}
\end{small}
}
\caption{Congruence on values}
\label{fig:congruenceVal}
\end{figure}
By rule \rn{new} we can assume that arguments of a constructor invocation are references, by introducing local declarations of the appropriate type. The notation $\notRef{\val}$ means that $\val$ is not a variable. 
The following three rules deal with block values.
By rule \rn{body} we can move a sequence of evaluated declarations
from a block to the directly enclosing block, and conversely. The notation $\noCapture{\e}{\xs}$ means that free variables in $\e$ are not captured by $\xs$, formally: $\FV{\e}\cap\xs=\emptyset$, and analogously for $\noCapture{\dvs}{\xs}$. 
The two side conditions ensure that moving the declarations {$\dvs'$} does not cause either scope extrusion 
or capture of free variables. More precisely: the first condition prevents capturing with $\dvs'$ some free 
variables of the enclosing block, whereas the second prevents moving outside a declaration in $\dvs'$ which 
depends on local variables of the inner block (declared in $\dvs''$). The first side condition can be 
satisfied by $\alpha$-conversion of the inner block, whereas the second cannot. 
By rule \rn{garbage} we can remove useless local store from a block.
Finally, rule \rn{block-elim} states the obvious fact that a block with no declarations is equivalent 
to its body.

{Congruence preserves types. We have to assume that the congruent expressions be typable, since rule
\rn{garbage} eliminates declarations from blocks that without this assumption could be not typable.
\begin{proposition}\label{lemma:congruenceTypes}
Let $\TypeCheck{\Gamma}{\LentLocked}{\StronglyLocked}{\e}{{\T}}$. If $\e\cong\e'$ and for some $\Gamma'$, $\LentLocked'$, $\StronglyLocked'$ and $\T'$ we have  that $\TypeCheck{\Gamma'}{\LentLocked'}{\StronglyLocked'}{\e'}{{\T'}}$, then $\TypeCheck{\Gamma}{\LentLocked}{\StronglyLocked}{\e'}{{\T}}$
\end{proposition}
\begin{proof}
By cases on the congruence rule used. \qed
\end{proof}}

Let a {\em  well-formed value} be either a variable or a well-formed right-value.
We can show that any value is congruent to a well-formed value.
\begin{proposition}\label{lemma:congruenceValue}
Let $\val$ be a value, then either $\congruence{\val}{\x}$ for some variable $\x$, or $\congruence{\val}{\stVal}$ for some 
\storableVal\  $\stVal$ {such that $\WFrv{\stVal}$}. 
\end{proposition}
\begin{proof}
By induction on values using the congruence rule of \refToFigure{congruenceVal}.\qed
\end{proof}

The reduction relation is defined by the rules in \refToFigure{reduction}, where $\ctx$ is an {\em evaluation context} 
defined by:
\begin{center}
$\ctx ::=\emptyctx\mid\Block{\dvs\ \Dec{\T}{\x}{\ctx}\ \decs}{\val}
$
\end{center}
{In this definition, and in the following we assume that the metavariables $\val$, $\stVal$, $\dv$ and $\dvs$ denote values,  right-values, evaluated declarations, and stores which are well-formed. The assumption on $\val$ and $\stVal$ can be achieved by Proposition~\ref{lemma:congruenceValue}.}

To lighten notations, here and in what follows we sometimes use the wildcard $\_$ when a metavariable 
is not significant.

\begin{figure}[ht]
\framebox{
\begin{small}
\begin{math}
\begin{array}{l}
%\NamedRule{ctx}{\reduce{\e}{\e'}}{\reduce{\Ctx{\e}}{\Ctx{\e'}}}{}
%\Space
\NamedRuleOL{field-access}{\reduce{\Ctx{\FieldAccess{{\val}}{\f}}}
{\Ctx{\extractField{\ctx}{\val}{i}}}}{
\begin{array}{l}
{\extractType{\ctx}{\val}=\Type{\mu}{\C}}\\
\fields{\C}=\Field{\T_1}{\f_1}\ldots\Field{\T_n}{\f_n}\ \mbox{and}\  \f=\f_i
%\\ \mu=\imm\implies\isCapsule{\ctx}{\val}
\end{array}
}
\\[4ex]
\NamedRuleOL{invk}
{
\begin{array}{l}
\Ctx{\MethCall{\val}{\m}{\val_1,\ldots,\val_n}}\\
\longrightarrow
\Ctx{
\Block{
\Dec{\Type{\mu}{\C}}{\,\this}{\val}\
\Dec{\T_1}{\x_1}{\val_1}
\ldots
\Dec{\T_n}{\x_n}{\val_n}\
\Dec{\T}{\z}{\e}
}{\z}
}
\end{array}
}
{
\begin{array}{l}
{\extractType{\ctx}{\val}=\Type{\_}{\C}}\\
\method{\C}{\m}=\\
\Space \FourTuple{\T}{\mu}{\Param{\T_1}{\x_1}\ldots\Param{\T_n}{\x_n}}{\e}
\end{array}
}
\\[4ex]
{\NamedRuleOL{field-assign-prop}{\reduce{\Ctx{\FieldAssign{{\val}}{\f}{\valPrime}}}{\Ctx{\Block{\Dec{\Type{\mu}{\C}}{\x}{{\val}}\
\Dec{\T_i}{\z}{(\FieldAssign{\x}{\f}{\valPrime})}}{\z}}}}
{
\begin{array}{l}
{\notRef{\val}}, 
\typeOf{\val}=\Type{\mu}{\C}\\
\fields{\C}=\Field{\T_1}{\f_1}\ldots\Field{\T_n}{\f_n}\\
\f=\f_i
\end{array}
}}
\\[4ex]
\NamedRuleOL{field-assign}{
\begin{array}{l}
{
\Ctx{\Block{\dvs\ \Dec{\T}{\z}{\CtxP{\FieldAssign{\x}{\f}{{\valPrime}}}}\ \decs}{\val}}
}\\
\longrightarrow
{
\Ctx{\Block{{\Update{\dvs}{\x}{i}{{\valPrime}}}\ \Dec{\T}{\z}{\CtxP{{\valPrime}}}\ \decs}{\val}}
}
\end{array}
}
{
\begin{array}{l}
\dvs(\x)=\Dec{\mu}{\x}{\ConstrCall{\C}{\_}},{\mu\geq\mutable}\\
{\noCapture{\x}{\HB{\ctxP}},\noCapture{\valPrime}{\HB{\ctxP}}}\\
%(\{\x\}\cup\FV{\valPrime})\cap\HB{\ctx}=\emptyset}\\
\fields{\C}=\Field{\T_1}{\f_1}\ldots\Field{\T_n}{\f_n}\ \mbox{and}\  \f=\f_i
\end{array}
}
\\[4ex]
\NamedRuleOL{field-assign-move}
{
\begin{array}{l}
\Ctx{\Block
	{\dvs'\
		\Dec{\T'}{\z'}
		{\Block
			{\dvs\
			\Dec{\T}{\z}
			{\CtxP{\x.\f{=}{\valPrime}}
			}\
			\decs
			}
			{\val}
		}\
		\decs'
	}
	{\val'}}\\
{\longrightarrow}
	\Ctx{\Block{\dvs'\
		\dvs\
		\Dec{\T'}{\z'}
		{
		\Block{
			\Dec{\T}{\z}
			{\CtxP{\x.\f{=}{\valPrime}}
			}\
			\decs
		}
		{\val}
		}\ 
		\decs'
	}
	{\val'}	}
\end{array}
}{
\begin{array}{l}
\hspace{-.2cm}\FV{\valPrime}{\cap}\dom{\dvs}{=}\xs{\neq}\emptyset\\
\hspace{-.2cm}\noCapture{\x}{\HB{\ctxP}\cup\dom{\dvs}}\\
\hspace{-.2cm}\noCapture{\valPrime}{\HB{\ctxP}}\\
\hspace{-.2cm}{\noCapture{\Block{\dvs'\ \decs'}{\val'}}{\dom{\dvs}}}\\
\hspace{-.2cm}{\Reduct{(\dvs\ \decs)}{{\xs}}=\dvs}
\end{array}
}
\\[4ex]
\NamedRuleOL{alias-elim}{\reduce{\Ctx{\Block{\dvs\ \Dec{\T}{\x}{\y}\ \decs }{\val}}}{\Ctx{\Subst{\Block{\dvs\ \decs}{\val}}{\y}{\x}}}}
{}
\\[4ex]
\NamedRuleOL{capsule-elim}{
\begin{array}{l}
{\Ctx{\Block{\dvs\ \Dec{\Type{\capsule}
{\C}}{\x}{{\val}}\ \decs }{\val'}}}
\\ \longrightarrow
{\Ctx{\Subst{\Block{\dvs\ \decs}{\val'}}{{\val}}{\x}}}
\end{array}
}
{%\isCapsule{\Ctx{\Block{\dvs\ \Dec{\Type{\capsule}{\C}}{\x}{{\emptyctx}}\ \decs }{\val'}}}{{\val}}
}
\\[4ex]
%\NamedRuleOL{imm-right}{
%\begin{array}{l}
%\Ctx{\Block{\dvs'\ \Dec{\Type{\imm}{\C}}{\x}{\Block{\dvs}{\y}}\ \decs}{\val}}
%\\ \longrightarrow
%
%\Ctx{\Block{\dvs'\ \Dec{\Type{\imm}{\C}}{\x}{\Block{\dvs}{{\ConstrCall{\D}{\xs}}}}\ \decs}{\val}}
%\end{array}
%}
%{\dvs(\y)=\Dec{\T}{\y}{\ConstrCall{\D}{\xs}}}
%\\[4ex]
\NamedRuleOL{mut-move}{
\begin{array}{l}
{
\Ctx{\Block{\dvs'\ \Dec{\Type{\mu}{\C}}{\x}{\Block{\dvs\ \dvs''}{\val}}\ \decs'}{\val'}}
}
\\ \longrightarrow
{
\Ctx{\Block{\dvs'\ \dvs\ \Dec{\Type{\mu}{\C}}{\x}{\Block{\dvs''}{\val}}\ \decs'}{\val'}}
}
\end{array}
}{\begin{array}{l}
\mu \geq \mutable\\
{\noCapture{\Block{\dvs'\ \decs'}{\val'}}{\dom{\dvs}}}\\
%\FV{\Block{\dvs'\ \decs'}{\val'}}\cap\dom{\dvs}=\emptyset\\
{\noCapture{\dvs}{\dom{\dvs''}}}
%\FV{\dvs}\cap\dom{\decs}=\emptyset
\end{array}}
\\[4ex]
\NamedRuleOL{imm-move}{
\begin{array}{l}
{
\Ctx{\Block{\dvs'\ \Dec{\Type{\mu}{\C}}{\x}{\Block{\dvs\ \dvs''}{\val}}\ \decs'}{\val'}}
}
\\ \longrightarrow
{
\Ctx{\Block{\dvs'\ \dvs\ \Dec{\Type{\mu}{\C}}{\x}{\Block{\dvs''}{\val}}\ \decs'}{\val'}}
}
\end{array}
}{\begin{array}{l}
\mu \leq \imm, {\allImm{\dvs}}\\
%\extractMod{\dv}=\imm\ \forall \dv\in\dvs\\
{\noCapture{\Block{\dvs'\ \decs'}{\val'}}{\dom{\dvs}}}\\
%\FV{\Block{\dvs'\ \decs'}{\val'}}\cap\dom{\dvs}=\emptyset\\
{\noCapture{\dvs}{\dom{\dvs''}}}
%\FV{\dvs}\cap\dom{\decs}=\emptyset
\end{array}}
\end{array}
\end{math}
\end{small}
}
\caption{Reduction rules}
\label{fig:reduction}
\end{figure}
In rule \rn{field-access}, the type $\Type{\mu}{\C}$ of the receiver $\val$ in the context is found, fields of the class $\C$ are retrieved from the class table, it is checked that $\f$ is actually the name of a field of $\C$, say, the $i$-th field, and the field access is reduced to the $i$-th field of the receiver.
%If the qualifier $\mu$ is $\imm$, then it is checked that the receiver is a $\capsule$. 
%\MSComm{it seams disaligned with the rule, where there is no capsule checking}
%The judgment $\isCapsule{\ctx}{\val}$ holds if, for all $\x\in\FV{\val}$, {$\extractDec{\ctx}{\x}=\Dec{\Type{\imm}{\_}}{\x}{\_}$}

The auxiliary functions $\aux{type}$ and $\aux{get}$ extract the type, and the $i$-th field, respectively, of a value in a context (only needed when the value is a reference). In the definitions, by Proposition~\ref{lemma:congruenceValue}, we assume that values are either references or \storableVals.
{\begin{small}
\begin{quote}
$\extractType{\ctx}{\x}=\T$ if $\extractDec{\ctx}{\x}=\Dec{\T}{\x}{{\_}}$ \\
$\extractType{\ctx}{\stVal}=\typeOf{\stVal}$\\
$\typeOf{\ConstrCall{\C}{\xs}}=\typeOf{\Block{\dvs}{\ConstrCall{\C}{\xs}}}=\Type{\mutable}{\C}$\\
$\typeOf{\Block{\dvs}{\x}}=\T$ if $\dvs(\x)=\Dec{\T}{\x}{{\_}}$\\
\\
$\extractField{\ctx}{\x}{i}=\fieldOf{{\stVal}}{i}$ if $\extractDec{\ctx}{\x}=\Dec{\_}{\x}{{\stVal}}$\\
$\extractField{\ctx}{\stVal}{i}=\fieldOf{\stVal}{i}$\\
{$\fieldOf{\ConstrCall{\C}{\x_1,\ldots,\x_n}}{i}=\x_i$}\\
{$\fieldOf{\Block{\dvs}{\ConstrCall{\C}{\x_1,\ldots,\x_n}}}{i}=\Block{\dvs}{\x_i}$}\\
{$\fieldOf{\Block{\dvs}{\x}}{i}=\Block{\dvs}{\val}$ if $\dvs(\x)=\Dec{\_}{\x}{\stVal}$ and $\fieldOf{\stVal}{i}=\val$}
\\
$\extractDec{\Block{\dvs\ \Dec{\T}{\_}{\ctx}\ \decs}{\val}}{\x}=\begin{cases}
\extractDec{\ctx}{\x}&\mbox{if}\ \extractDec{\ctx}{\x}\ \mbox{defined, otherwise}\\
\dvs(\x)&\mbox{if}\ \dvs(\x)\ \mbox{defined}
\end{cases}$
\end{quote}
\end{small}}
Note that a field access $\FieldAccess{\x}{\f}$, if $\x$ has qualifier $\geq\mutable$,  always returns a reference, since the right-value of $\x$ is necessarily an object state.
If $\x$ is immutable, instead, the field access could return a (copy of) the value stored in the field. This duplication preserves the expected semantics in the case of an immutable reference $\x$, whereas it would be wrong for a mutable reference, since a modification of the object denoted by $\x$ is expected to affect $\FieldAccess{\x}{\f}$ as well, and conversely.

For instance, given the value:
\begin{quote}
$\val=$
\Q@{mut C x=new C(x,y,z); mut D y=new D(0); new C(x,y,z) }@
\end{quote}
we have:
\begin{enumerate}
\item $\fieldOf{\val}{1}=$
\Q@{mut C x=new C(x,y,z); mut D y=new D(0); x }@
\item $\fieldOf{\val}{2}=$
\Q@{mut C x=new C(x,y,z); mut D y=new D(0); y}@
\item $\fieldOf{\val}{3}=$
\Q@{mut C x=new C(x,y,z); mut D y=new D(0); z}@
\end{enumerate}
where 1 is a well-formed value, 2 is congruent by rule \rn{garbage} to the  well-formed value
\Q@{mut D y=new D(0)   y}@, and 3 is congruent by rules \rn{garbage} and \rn{block-elim}, to the well-formed value \Q@z@.

If the value $\val$ above is the right-value of a reference, then such reference is necessarily $\imm$. In this case, $\val$ 
was expected to be a capsule, the variable $\z$ should be declared $\imm$ in the enclosing context.

In rule \rn{invk}, the class $\C$ of the receiver $\val$ is found, and method $\m$ of $\C$ is retrieved from the class table. The call is reduced to a block where declarations of the appropriate type for $\this$, the parameters, and the result are initialized with the receiver, the arguments, and the method body, respectively. The last declaration, variable $\z$, is
 needed to preserve the simplified syntax.

Rule \rn{field-assign-prop} handles the case of a field assignment where the receiver is not a reference (denoted $\notRef{\val}$). In this case, a local reference initialized with the receiver is introduced, and the field access is propagated to such reference. 

In rule \rn{field-assign}, given a field assignment where the receiver is a reference $\x$, the first enclosing declaration for $\x$ is found (the side condition {$\noCapture{\x}{\HB{\ctxP}}$} ensures that it is actually the first), it is checked that the qualifier of the type of $\x$ is $\geq\mutable$, and fields of the class $\C$ of $\x$ are retrieved from the class table. If $\f$ is the name of a field of $\C$, say, the $i$-th, then the $i$-th field of the right-value of $\x$ is updated to $\valPrime$, which is also the result of the field assignment. 
We write $\HB{\ctx}$ for the \emph{hole binders} of $\ctx$, that is, the variables declared in blocks enclosing the
context hole (the standard formal definition is in the Appendix) and $\Update{\dvs}{\x}{i}{{\valPrime}}$ for the sequence of evaluated declarations obtained from $\dvs$ by replacing the $i$-th field of the right-value of $\x$ with ${\valPrime}$ (the obvious formal definition is omitted).

The side condition {$\noCapture{\valPrime}{\HB{\ctxP}}$}, requiring that there are no inner declarations for some free variable in $\valPrime$, prevents scope extrusion. For instance, without this side condition, the term

\begin{small}
\begin{lstlisting}
mut C x= new C(...);
imm C z= { 
  mut D y1= new D(0);
  mut D y2= ( x.f = y1);
  mut D y3= new D(1);
  new C(y3) };
x
\end{lstlisting}
\end{small}

would erroneously reduce to 

\begin{small}
\begin{lstlisting}
mut C x= new C(y1); 
imm C z= { 
  mut D y = new D(0);
  mut D y2= y1;
  mut D y3= new D(1);
  new C(y3) };
x
\end{lstlisting}
\end{small}
Thanks to the side condition, instead, rule \rn{field-assign} is not applicable. However, 
rule \rn{field-assign-move} can be applied.

Rule \rn{field-assign-move} moves store out of a block when this is needed to safely perform field-assignment (that is, to avoid scope extrusion).
In this rule,  given a field assignment of shape $\FieldAssign{\x}{\f}{\valPrime}$, the first enclosing block containing (evaluated) declarations for some free variables $\xs$  of $\valPrime$ is found (the side condition {$\noCapture{\valPrime}{\HB{\ctxP}}$} ensures that it is actually the first). If a declaration for $\x$ can only be found in an outer scope (side condition {$\noCapture{\x}{\HB{\ctxP}\cup\dom{\dvs}}$}), then the store formed by the $\xs$ references, together with all the other references they (recursively) depend on (last side condition) is moved to the directly enclosing block.

For the example above, by applying rule \rn{field-assign-move} we get:

\begin{small}
\begin{lstlisting}
mut C x = new C(...)  
mut D y1 = new D(0)
imm C z = { 
  mut D y2 = ( x.f = y1)
  mut D y3 = new D(1)
  new C(y3) }  
x
\end{lstlisting}
\end{small}
Now, rule \rn{field-assign} can be safely applied to the term.
In general, we may need to apply rule \rn{field-assign-move} many times before reaching the declaration of $\x$.

The remaining rules handle blocks.

The first two rules eliminate a single declaration of shape $\Dec{\T}{\x}{\val}$ which is \emph{not} well-formed store. 

 In rule \rn{alias-elim}, a reference $\x$ which is initialized as an alias\footnote{{An analogous rule would handle variables of primitive types in an extension of the calculus including such types.}} of another reference $\y$ is eliminated by replacing all its occurrences.
 
In rule \rn{capsule-elim}, a $\capsule$ reference $\x$ is eliminated by replacing the occurrence of $\x$ (unique by the well-formedness constraint) by its value. %, which is checked to be a capsule.
Note that, in rule \rn{alias-elim}, $\y$ cannot be a $\capsule$ reference (since we would have applied rule \rn{capsule-elim} before), hence  duplication of $\capsule$ references cannot be introduced{, and well-formedness is preserved.} 

Rule \rn{mut-move} moves store out of a block associated to a {reference with qualifier $\geq\mutable$}. 
This is always safe, provided that no variables of the outer scope are captured (second side condition, which can be obtained by $\alpha$-renaming), and that the moved declarations do not refer to other declarations of the inner block (last side condition). 

Rule \rn{imm-move} moves store out of a block associated to a {reference with qualifier $\leq\imm$}. In this case, this is only safe for a store of $\imm$ references. The same side conditions of the previous rule are needed.

%% file: Results.tex
% !TEX root =main.tex
\section{ {Type safety and properties of the type system}}\label{sect:results}
 {In this section we present  {the results}. We first give a characterisation of the values in terms of the
properties of their free variables, which  {correspond to} their reachable graph. Then we show the soundness
of the type system for the operational semantics, and finally we formalize the expected behaviour of capsule and
immutable references.}

 {In the following, we denote by $\deriv:\TypeCheckShort{\Delta}{\e}{\T}$ a derivation tree for the judgement 
 $\TypeCheckShort{\Delta}{\e}{\T}$. 
Moreover, we call the rules where expressions in the premises are the direct subterms of the expression in the consequent,  \emph{structural}; the others, that is, \rn{t-capsule}, \rn{t-imm}, \rn{t-swap}, \rn{t-unrst}, and \rn{t-sub},  \emph{non-structural}.

\subsection{Canonical Forms}
In a type derivation, given a construct, % (in the case of right-values we have constructors or blocks), 
in addition to the corresponding structural rule we can have applications of non-structural rules.  
In this section we first present some results exploring the effect of non-structural rules on the lent groups and the mutable groups of variables, and on the modifier derived for
the expression, then we give an inversion lemma for blocks (the construct which is relevant
for the analysis of right-values), and finally we present the Canonical Form theorem with its proof. 

Given a type judgement $\AuxTypeCheck{\Gamma}{\LentLocked}{\MutGroup}{\StronglyLocked}{\e}{\Type{\mu}{\C}}$,  {application of non-structural rules can only modify the lent groups and the mutable group by swapping, hence leading to a permutation $\LentLocked'\, \MutGroup$ of $\LentLocked\, \MutGroup$. In other terms,} the equivalence relation on $\domMut{\Gamma}$
 {they induce} is preserved,
as the following lemma shows.
\begin{lemma} [Non-structural rules]\label{lemma:nonStructural}
If $\deriv:\AuxTypeCheck{\Gamma}{\LentLocked}{\MutGroup}{\StronglyLocked}{\e}{\Type{\mu}{\C}}$, then 
there is a sub-derivation $\deriv'$ of $\deriv$ such that $\deriv':\AuxTypeCheck{\Gamma}{\LentLocked'}{\MutGroup'}{\StronglyLocked'}{\e}{\Type{\mu'}{\C}}$ ends with the application of a structural rule,  $\LentLocked\ \MutGroup=\LentLocked'\ \MutGroup'$, and $\mu\not=\imm$, $\mu\not=\capsule$ implies $\mu'\leq\mu$.
\end{lemma}
\begin{proof}
The proof is in~\ref{sect:proof-cf}.\qed
\end{proof}
The application of non-structural rules is finalized to the recovery of capsule and immutable properties
as expressed by the following lemma.
\begin{lemma}\label{lemma:typeStruct} 
Let $\deriv:\TypeCheck{\Gamma}{\LentLocked}{\StronglyLocked}{{\e}}{\Type{\mu}{\C}}$. 
\begin{enumerate}
  \item If $\mu=\mutable$, then the last rule applied in $\deriv$ cannot be \rn{t-swap} or \rn{t-unrst} or \rn{t-capsule} or  \rn{t-imm}.
  \item If $\mu=\imm$ or $\mu=\capsule$ and the last rule applied in $\deriv$ is  \rn{t-swap} or \rn{t-unrst}, then
  $\TypeCheck{\Gamma}{\LentLocked'}{\StronglyLocked'}{{\e}}{\Type{\mu}{\C}}$, for some $\LentLocked'$ and 
  $\StronglyLocked'$.
\end{enumerate}
\end{lemma}
\begin{proof}\
\begin{enumerate}
  \item Rule \rn{t-unrst} is only applicable when the type derived for the expression in the premise of the rule, $\Type{\mu'}{\C}$, is such that $\mu'\leq \imm$. For the rule \rn{t-swap}, when the type derived for the expression in the premise of the rule is $\Type{\mutable}{\C}$, then the the type of the expression in the consequent is  $\Type{\lent}{\C}$. Moreover, the two  {recovery} rules do not derive types with the $\mutable$ modifier.
  \item Immediate.
\end{enumerate}\qed
\end{proof}
From the previous lemma we derive that, if $\TypeCheck{\Gamma}{\LentLocked}{\StronglyLocked}{{\e}}{\Type{\mu}{\C}}$ with $\mu\not\geq\lent$, then without loss of generality we can assume that the last rule applied in the derivation is either a structural rule followed by a \rn{t-sub}, or a  {recovery rule} in case $\mu=\capsule$ or $\mu=\imm$. 
Moreover, rules \rn{t-swap} and \rn{t-unrst} cannot be used to derive the premise of rule \rn{t-capsule}.

 {If we can derive a type for a variable, then such type depends on the type assignment and the lent groups (if the variable is restricted no type can be derived).
If the variable is declared with modifier $\imm$ or $\capsule$, then its type depends only on the type assignment.}
\begin{lemma}\label{lemma:typeVars}
If $\TypeCheck{\Gamma}{\LentLocked}{\StronglyLocked}{{\x}}{\T}$ then $\Gamma(x)\leq\T$. Moreover, if 
$\T=\Type{\mu}{\C}$ with $\mu\leq\imm$, then $\TypeCheck{\Gamma}{\LentLocked'}{\StronglyLocked'}{{\x}}{\T}$ for all 
$\LentLocked'$ and $\StronglyLocked'$ such that  $\WellFormedTypeCtx{\Gamma};{\LentLocked'};\StronglyLocked'$.
\end{lemma}
\begin{proof}
The proof is in~\ref{sect:proof-cf}.\qed
\end{proof}
The block construct is central to our recovery technique. New variable are defined, and the lent and mutable
groups of the free variables of the block may be modified by the introduction of the newly defined
variables. However, as for the non-structural rules, the equivalence relation between the mutable 
variables induced by the partition determined by the lent and   mutable groups is preserved.

We first define a judgment asserting when a declaration of a variable is well typed and then give a lemma
relating the type judgement for a  block with the type judgements for the expressions
associated to its local variables by the declarations and for its body. The expression associated
with a variable declared with modifier $\lent$ has to be well typed taking as mutable the variable in its
group, whereas the others have the same mutable group of the body of the expression.  This expresses
the fact that variables declared with modifier $\lent$ should not be connected to the result of the block
(its body).
\begin{definition}\label{def:wellTypedDefsNew}
Define $\AuxDecsOK{\Gamma}{\LentLocked}{\MutGroup}{\StronglyLocked}{\Dec{\Type{\mu}{\C}}{\x}{\e}}$ by:
\begin{itemize}
  \item $\AuxTypeCheck{\Gamma}{\LentLocked}{\MutGroup}{\StronglyLocked}{\e}{\Gamma(\x)}$ if $\mu\not=\lent$
  \item $\AuxTypeCheck{\Gamma}{\LentLocked'}{\MutGroup'}{\StronglyLocked}{\e}{\Gamma(\x)}$ if $\mu=\lent$ and $\LentLocked\ \xs=\LentLocked'\ \xs'$ where $\x\in\xs'$.
\end{itemize}
$\AuxDecsOK{\Gamma}{\LentLocked}{\MutGroup}{\StronglyLocked}{\decs}$ if for all $\dec\in\decs$, $\AuxDecsOK{\Gamma}{\LentLocked}{\MutGroup}{\StronglyLocked}{\dec}$.
\end{definition}

\begin{lemma} [Inversion for blocks]\label{lemma:inversionBlock}
If $\AuxTypeCheck{{\Gamma}}{\LentLocked}{\MutGroup}{\StronglyLocked}{\Block{\decs}{\val}}{\Type{\mu}{\C}}$, then for some $\LentLocked'$ and $\xs'$
\begin{enumerate}
\item $\AuxDecsOK{\SubstFun{\Gamma}{\TypeEnv{\decs}}}{\LentLocked'}{\MutGroup'}{\StronglyLocked'}{\decs}$
\item  $\AuxTypeCheck{{\SubstFun{\Gamma}{\TypeEnv{\decs}}}}{{\LentLocked'}}{\MutGroup'}{\StronglyLocked'}{\val}{\Type{\mu'}{\C}}$ 
\item $\LessEq{(\LentLocked\ \xs){\setminus}\dom{\TypeEnv{\decs}}}{\LentLocked'\ \xs'}$ and
\item $\mu\not=\imm$ and $\mu\not=\capsule$ implies $\mu'\leq\mu$.
\end{enumerate}
\end{lemma}
\begin{proof}
The proof is in~\ref{sect:proof-cf}.\qed
\end{proof}
The two following lemmas characterize the shape of $\mutable$ right-values. 
For constructor values we also have that if the mutable variables are restricted then all fields
must have immutable types.
\begin{lemma}[Constructor value]\label{lemma:constrMut}
Let $\fields{\C}{=}\Field{\Type{\mu_1}{\C_1}}{\f_1}\ldots\Field{\Type{\mu_n}{\C_n}}{\f_n}$.
\begin{enumerate}
  \item If $\AuxTypeCheck{\Gamma}{\LentLocked}{\MutGroup}{\StronglyLocked}{\ConstrCall{\C}{\z_1,\ldots,\z_n}}{\Type{\mutable}{\C}}$, then: for all $i\in1..n$, if $\mu_i=\mutable$, then $\z_i\in\xs$, otherwise $\AuxTypeCheck{\Gamma}{\LentLocked}{\MutGroup}{\StronglyLocked}{\z_i}{\Type{\mu}{\C_i}}$ with $\mu\leq\imm$.
\item If $\AuxTypeCheck{\Gamma}{\LentLocked}{\emptyset}{\domGeqMut(\Gamma) }{\ConstrCall{\C}{\z_1,\ldots,\z_n}}{\Type{\readable}{\C}}$, then: for all $i\in1..n$ we have that $\AuxTypeCheck{\Gamma}{\LentLocked}{\MutGroup}{\StronglyLocked}{\z_i}{\Type{\mu}{\C_i}}$ with $\mu\leq\imm$.
\end{enumerate}
\end{lemma}
\begin{proof}
The proof is in~\ref{sect:proof-cf}.\qed
\end{proof}
In case a right-value is a  block, the fact that all the declared variables are connected to the
body of the block implies that the modifiers of all the variables must be either $\mutable$ or $\leq\imm$.
\begin{lemma}[Block value]\label{lemma:blockMut}
If $\AuxTypeCheck{\Gamma}{\LentLocked}{\MutGroup}{\StronglyLocked}{\Block{\dvs}{{\cOrx}}}{\Type{\mutable}{\C}}$, where
$\dvs=\Dec{\Type{\mu_1}{\C_1}}{\z_1}{\stVal_1}\cdots\Dec{\Type{\mu_n}{\C_n}}{\z_n}{\stVal_n}$, then
  \begin{enumerate}
  \item $\AuxTypeCheck{\Gamma[\TypeEnv{\dvs}]}{{\LentLocked'}}{\xs'}{\StronglyLocked}{{\cOrx}}{\Type{\mutable}{\C}}$,
  \item $\AuxTypeCheck{\Gamma[\TypeEnv{\dvs}]}{{\LentLocked'}}{\xs'}{\StronglyLocked}{\stVal_i}{\Type{\mutable}{\C_i}}$ 
  $\Range{i}{1}{n}$ 
\item for all $\z\in\FV{\Block{\dvs}{\cOrx}}$, if $\Gamma(\z)=\Type{\mutable}{\D}$, then $\z\in\xs$, otherwise 
$\AuxTypeCheck{\Gamma}{\LentLocked}{\MutGroup}{\StronglyLocked}{z}{\Type{\mu}{\C}}$ with $\mu\leq\imm$.
\end{enumerate}
\end{lemma}
\begin{proof}
The proof is in~\ref{sect:proof-cf}.\qed
\end{proof}
As usual typing depends only on the free variable of expressions, as the following weakening lemma states.
\begin{lemma}[Weakening]\label{lemma:weakening}
Let $\dom{\Gamma'}\cap\FV{\e}=\emptyset$ and $\WellFormedTypeCtx{\SubstFun{\Gamma}{\Gamma'}};{\LentLocked};\StronglyLocked$. \\
$\TypeCheck{\SubstFun{\Gamma}{\Gamma'}}{\LentLocked}{\StronglyLocked}{\e}{\T}$ if and
only if  $\TypeCheck{{\Gamma}}{{\LentLocked}{\setminus}{{\dom{\Gamma'}}}}{{\StronglyLocked}{\setminus}{{\dom{\Gamma'}}}}{\e}{\T}$.
\end{lemma}
\begin{proof}
By induction on derivations. \qed
\end{proof}
The Canonical Forms theorem describes constraints on the free variables and the extracted type of well-typed right-values. In particular, $\capsule$ and $\imm$ right-values can contain only $\capsule$ or $\imm$ references, and 
$\mutable$ right-values cannot contain $\lent$ or $\readable$ references. Moreover, the type extracted from $\capsule$ and $\mutable$ right-values is necessarily $\mutable$.  
\begin{theorem}[Canonical Forms]\label{theo:canonicalForm}
If $\AuxTypeCheck{\Gamma}{\LentLocked}{\MutGroup}{\StronglyLocked}{\stVal}{\Type{\mu}{\_}}$, and $\y\in\FV{\stVal}$, then:
 \begin{enumerate}
 \item if $\mu=\capsule$, then $\AuxTypeCheck{\Gamma}{\LentLocked}{\MutGroup}{\StronglyLocked}{\y}{\Type{\imm}{\_}}$, and $\typeOf{\stVal}=\Type{\mutable}{\_}$ 
  \item if $\mu=\mutable$, then $\AuxTypeCheck{\Gamma}{\LentLocked}{\MutGroup}{\StronglyLocked}{\y}{\Type{\mu'}{\_}}$, with $\mu'\not\geq\lent$, and $\typeOf{\stVal}=\Type{\mutable}{\_}$
  \item if $\mu=\imm$, then $\AuxTypeCheck{\Gamma}{\LentLocked}{\MutGroup}{\StronglyLocked}{\y}{\Type{\imm}{\_}}$
   \item if $\mu=\lent$, then $\AuxTypeCheck{\Gamma}{\LentLocked}{\MutGroup}{\StronglyLocked}{\y}{\Type{\readable}{\_}}$, and $\typeOf{\stVal}\leq\Type{\lent}{\_}$ 
   \item if $\mu=\readable$, then $\AuxTypeCheck{\Gamma}{\LentLocked}{\MutGroup}{\StronglyLocked}{\y}{\Type{\readable}{\_}}$
 \end{enumerate}
\end{theorem}
\begin{proof} 
By structural induction on $\stVal$ and by cases on $\mu$. 

\medskip\noindent
Let \underline{$\stVal=\ConstrCall{\C}{\x_1,\ldots,\x_n}$} and $\fields{\C}{=}\Field{\Type{\mu_1}{\C_1}}{\f_1}\ldots\Field{\Type{\mu_n}{\C_n}}{\f_n}$.  By definition
$\typeOf{\ConstrCall{\C}{\x_1,\ldots,\x_n}}=\Type{\mutable}{\C}$.\\

\medskip\noindent
If \underline{$\mu=\capsule$},  from \refToLemma{typeStruct}.2 we can assume that the last rule applied in the derivation is
\rn{t-capsule}:
\[
\prooftree
\begin{array}{c}
\AuxTypeCheck{\Gamma}{\LentLocked\,\xs}{\emptyset}{\StronglyLocked}{{\ConstrCall{\C}{\zs}}}{\Type{\mutable}{\C}}
\end{array}
\justifies
\AuxTypeCheck{\Gamma}{\LentLocked}{\MutGroup}{\StronglyLocked}{{{\ConstrCall{\C}{\zs}}}}{\Type{\capsule}{\C}}
\using
\rn{t-capsule}
\endprooftree
\]
%where $\{\zs\}=\domMut{\Gamma}{\setminus}\LentLocked$. 
From \refToLemma{constrMut}.1, since the current mutable group is empty, 
$\Range{i}{1}{n}$ we have that
$\AuxTypeCheck{\Gamma}{\LentLocked}{\MutGroup}{\StronglyLocked}{\x_i}{\Type{\imm}{\C_i}}$.

\medskip\noindent
If \underline{$\mu=\mutable$}, from rule \rn{t-new}, 
$\AuxTypeCheck{\Gamma}{\LentLocked}{\MutGroup}{\StronglyLocked}{\ConstrCall{\C}{\x_1,\ldots,\x_n}}{\Type{\mutable}{\C}}$.
Therefore, from \refToLemma{constrMut}.1 we get the result. 

\medskip\noindent
Let \underline{$\mu=\imm$.} From \refToLemma{typeStruct}.2 we can assume that the last rule applied in the derivation is
either \rn{t-sub}, or \rn{t-imm}. 
In the first case
$\AuxTypeCheck{\Gamma}{\LentLocked}{\MutGroup}{\StronglyLocked}{{{\ConstrCall{\C}{\zs}}}}{\Type{\capsule}{\C}}$, so the proof for $\mu=\capsule$ applies.
In the second, we have 
\[
\prooftree
\begin{array}{c}
\deriv:\AuxTypeCheck{\Gamma}{\LentLocked\ \xs}{\emptyset}{\domGeqMut(\Gamma) }{{\ConstrCall{\C}{\zs}}}{\Type{\readable}{\C}}
\end{array}
\justifies
\AuxTypeCheck{\Gamma}{\LentLocked}{\MutGroup}{\StronglyLocked}{{{\ConstrCall{\C}{\zs}}}}{\Type{\imm}{\C}}
\using
\rn{t-imm}
\endprooftree
\]
%where $\zs=\domMut{\Gamma}{\setminus}\StronglyLocked$. 
From \refToLemma{constrMut}.2  we get the result. 

\medskip\noindent
Let \underline{$\mu=\lent$ or $\mu=\readable$.} Then to derive a type for $\ConstrCall{\C}{\xs}$ we have
to apply \rn{t-new} followed by \rn{t-sub}, so the result is obvious.\\

\bigskip\noindent
Let \underline{$\stVal=\Block{\dvs}{\cOrx}$} where $\dvs$ is 
$\Dec{\T_1}{\y_1}{\stVal_1}$ $\ldots$ $\Dec{\T_n}{\y_n}{\stVal_n}$, and  $\cOrx=\ConstrCall{\C}{\xs}$ or
$\cOrx=x$.  

\medskip\noindent
If \underline{$\mu=\capsule$}, 
then, from \refToLemma{typeStruct}.2 we can assume that the last rule applied in the derivation is
\rn{t-capsule}:
\[
\prooftree
\begin{array}{c}
\AuxTypeCheck{\Gamma}{\LentLocked\,\xs}{\emptyset}{\StronglyLocked}{\Block{\dvs}{\_}}{\Type{\mutable}{\C}}
\end{array}
\justifies
\AuxTypeCheck{\Gamma}{\LentLocked}{\MutGroup}{\StronglyLocked}{{{\Block{\dvs}{\_}}}}{\Type{\capsule}{\C}}
\using
\rn{t-capsule}
\endprooftree
\]
%where $\zs=\domMut{\Gamma}{\setminus}\LentLocked$. 
From \refToLemma{blockMut}, since the current mutable group is empty, 
$\Range{i}{1}{n}$ we have that
$\AuxTypeCheck{\Gamma}{\LentLocked}{\MutGroup}{\StronglyLocked}{\x_i}{\Type{\imm}{\C_i}}$.

\medskip\noindent
If \underline{$\mu=\mutable$}, \refToLemma{typeStruct}.1 implies that the last rule
of the type derivation is either an application of \rn{t-sub}, or of \rn{t-block}.
In the first case  the proof for $\mu=\capsule$ applies. In the second case, we have
\[
\TypeCheck{\Gamma}{\LentLocked}{\StronglyLocked}{\Block{\dvs}{\_}}{\Type{\mutable}{\C}}
\]
From \refToLemma{blockMut} we get the result.

\medskip\noindent
If \underline{$\mu=\imm$}, from \refToLemma{typeStruct}.2 we can assume that the last rule
applied is either \rn{t-sub} or \rn{t-imm}. If the rule
applied is \rn{t-sub}, then 
$\TypeCheck{\Gamma}{\LentLocked}{\StronglyLocked}{\Block{\dvs}{\_}}{\Type{\capsule}{\C}}
$, and the proof for $\mu=\capsule$ applies. \\
Let the last rule applied in $\deriv'$ be \rn{t-imm}, i.e.,
\[
\prooftree
\begin{array}{c}
\deriv:\AuxTypeCheck{\Gamma}{\LentLocked\ \xs}{\emptyset}{\domGeqMut(\Gamma) }{\Block{\dvs}{\_}}{\Type{\readable}{\C}}
\end{array}
\justifies
\AuxTypeCheck{\Gamma}{\LentLocked}{\MutGroup}{\StronglyLocked}{\Block{\dvs}{\_}}{\Type{\imm}{\C}}
\using
\rn{t-imm}
\endprooftree
\]
%where $\xs=\domMut{\Gamma}{\setminus}\LentLocked$.  \\
From \refToLemma{inversionBlock}, since we may assume that in the derivation $\deriv$ the sub-derivation
that ends with rule \rn{t-block} is not the antecedent of rule \rn{t-unrst} (otherwise the block would already have type $\Type{\imm}{\C}$) we have that
\begin{enumerate}[(a)]
\item  $\AuxTypeCheck{{\SubstFun{\Gamma}{\TypeEnv{\dvs}}}}{{\LentLocked'}}{\MutGroup'}{\StronglyLocked'}{\_}{\Type{\mu'}{\C}}$ 
\item $\AuxDecsOK{\SubstFun{\Gamma}{\TypeEnv{\dvs}}}{\LentLocked'}{\MutGroup'}{\StronglyLocked'}{\Dec{\Type{\mu_i}{\C_i}}{\z_i}{\stVal_i}}$  $\Range{i}{1}{n}$
\item $\LessEq{(\LentLocked\ \xs){\setminus}\dom{\TypeEnv{\dvs}}}{\LentLocked'\ \xs'}$ and
\item $\domGeqMut(\Gamma){\setminus}\dom{\dvs}=\StronglyLocked'$ 
\end{enumerate}
If $\y\in\FV{\Block{\dvs}{\_}}$, then either $\y\in\FV{{\_}}{\setminus}\dom{\dvs}$ or $\y\in\FV{\stVal_i}{\setminus}\dom{\dvs}$ for some $i\in 1..n$.\\
If $\y\in\FV{{\_}}{\setminus}\dom{\dvs}$ and $\y\in\domGeqMut(\Gamma)$, from (d) $\y\in\StronglyLocked'$. In order to apply rule \rn{t-var} to derive a type for $\y$ we
would have to apply rule \rn{t-unrst}. Therefore $\Gamma[\TypeEnv{\dvs}](\y)=\Type{\mu}{\D}$ with $\mu\leq\imm$.
(Rule \rn{t-unrst} cannot follow \rn{t-new}.) So $\y\not\in\domGeqMut(\Gamma)$, and $\Gamma(\y)=\Type{\mu}{\D}$ with $\mu\leq\imm$.\\
If for some $i\in 1..n$, $\y\in\FV{\stVal_i}{\setminus}\dom{\dvs}$ and $\y\in\domGeqMut(\Gamma)$, let 
$\deriv:\AuxTypeCheck{{\SubstFun{\Gamma}{\TypeEnv{\dvs}}}}{{\LentLocked_i}}{\MutGroup_i}{\StronglyLocked'}{\stVal_i}{\Type{\mu_i}{\C_i}}$. Since $\y\in\domGeqMut(\Gamma)$ there is a sub-derivation $\deriv'$ of $\deriv$, such that
$\deriv':\AuxTypeCheck{{\SubstFun{\Gamma}{\TypeEnv{\dvs}}}}{{\LentLocked'_i}}{\MutGroup'_i}{\StronglyLocked'}{\val'}{\Type{\mu'}{\C'}}$, the last  rule applied is \rn{t-unrst}, and $\y\in\FV{\val'}$, i.e.,
\[
\prooftree
\AuxTypeCheck{{\SubstFun{\Gamma}{\TypeEnv{\dvs}}}}{{\LentLocked'_i}}{\MutGroup'_i}{\emptyset}{\val'}{\Type{\imm}{\C'}}\ \ 
\justifies
\AuxTypeCheck{{\SubstFun{\Gamma}{\TypeEnv{\dvs}}}}{{\LentLocked'_i}}{\MutGroup'_i}{\StronglyLocked'}{\val'}{\Type{\imm}{\C'}}
\using
\rn{t-unrst}
\endprooftree
\] 
If $\val'=\y$ then $\SubstFun{\Gamma}{\TypeEnv{\dvs}}(\y)=\Type{\mu}{\D}$ with $\mu\leq\imm$, therefore also $\Gamma(\y)=\Type{\mu}{\D}$ with $\mu\leq\imm$, which is impossible since $\y\in\domGeqMut(\Gamma)$. \\
If $\val'=\stVal'$ for some $\stVal'$, by induction hypothesis
on $\stVal'$, we have that $\AuxTypeCheck{{\SubstFun{\Gamma}{\TypeEnv{\dvs}}}}{{\LentLocked'_i}}{\MutGroup'_i}{\emptyset}{\y}{\Type{\mu}{\D}}$ where $\mu\leq\imm$. Since $\y\not\in\dom{\dvs}$ and $\mu\leq\imm$, by \refToLemma{weakening} and \refToLemma{typeVars}, also $\Gamma(\y)=\Type{\mu}{\D}$ with $\mu\leq\imm$, which is again impossible since $\y\in\domGeqMut(\Gamma)$. \\
Therefore, for all $\y\in\FV{\Block{\dvs}{\_}}$, we have that 
$\TypeCheck{\Gamma}{\LentLocked}{\StronglyLocked}{\y}{\Type{\imm}{\D}}$ with $\mu\leq \imm$. 

\medskip\noindent
If \underline{$\mu=\lent$} or \underline{$\mu=\readable$}, then $\y$ must be such that
$\TypeCheck{\Gamma}{\LentLocked}{\StronglyLocked}{\y}{\Type{\mu'}{\D}}$, so also
$\TypeCheck{\Gamma}{\LentLocked}{\StronglyLocked}{\y}{\Type{\readable}{\D}}$.\\
Moreover, if \underline{$\mu=\lent$} and $\TypeCheck{\Gamma}{\LentLocked}{\StronglyLocked}{\Block{\dvs}{\cOrx}}{\Type{\lent}{\C}}$, in case $\cOrx=\ConstrCall{\C}{\xs'}$, we have that 
$\typeOf{\Block{\dvs}{\cOrx}}=\Type{\mutable}{\C}$, and if $\cOrx=\x$, so $\x=\y_1$,must
be that $\T_1=\Type{\mu_1}{C}$ with $\mu_1\leq\lent$.
\qed
\end{proof}
Note that there are no constraints for $\readable$ right-values (by subtyping any well-typed right-value is also $\readable$), and no constraints on free variables for $\lent$ right-values.  
Moreover, note that a right-value with only $\leq\imm$ free variables is $\imm$ regardless of its extracted type, since we can apply  {immutability recovery}.

\subsection{Soundness of the type system for the operational semantics}

As usual soundness is proved by proving that typability is preserved by reduction, ``subject reduction'',  and  that well-typed expressions are either values of reduce, ``progress''. The subject reduction result, in our system, is particularly
relevant, since invariants on the store are expressed at the syntactic level by the modifiers assigned to the expression
and by the lent and mutable groups used in the typing judgement. 
Preserving typability of expressions means  not only enforcing the properties expressed by modifiers, but also
preserving the relationship between variables expressed by the lent and mutable groups. 

To identify  subexpressions of expressions we define {\em general contexts} $\genCtx$ by:
\begin{quote}
\begin{grammatica}
\produzione{\genCtx}{\emptyctx\mid\FieldAccess{\genCtx}{\f}\mid\MethCall{\genCtx}{\m}{\vs}\mid\MethCall{\val}{\m}{\vs\ \genCtx\ \vs'}\mid\FieldAssign{\genCtx}{\f}{\val}
\mid\FieldAssign{\val}{\f}{\genCtx}}{}\\*
\seguitoproduzione{\mid \ConstrCall{\C}{\vs\ \genCtx\ \vs'}\mid\Block{\decs\ \Dec{\T}{\x}{\genCtx}\ \decs'}{\val}\mid\Block{\decs}{\genCtx}}{}
\end{grammatica}
\end{quote}

Given a general context $\genCtx$, $\TypeEnv{\genCtx}$ is defined by: 
\begin{itemize}
 \item $\TypeEnv{\emptyctx}=\emptyset$
 \item $\TypeEnv{\FieldAccess{\genCtx}{\f}}=\TypeEnv{\MethCall{\genCtx}{\m}{\vs}}=\TypeEnv{\MethCall{\val}{\m}{\vs\ \genCtx\ \vs'}}=\TypeEnv{\FieldAssign{\genCtx}{\f}{\val}}=\TypeEnv{\FieldAssign{\val}{\f}{\genCtx}}=\TypeEnv{\ConstrCall{\C}{\vs\ \genCtx\ \vs'}}=\TypeEnv{\genCtx}$
  \item $\TypeEnv{\Block{\decs\ \Dec{\T}{\x}{\genCtx}\ \decs'}{\val}}=\TypeEnv{\decs},\TypeEnv{\decs'},\TypeDec{\T}{\x} [\TypeEnv{\genCtx}]$
  \item $\TypeEnv{\Block{\decs}{\genCtx}}=\TypeEnv{\decs}[\TypeEnv{\genCtx}]$
 \end{itemize}
$\TypeEnv{\ctx}$ is defined similarly.  We also use $\LentVars{\genCtx}$ and $\LentVars{\ctx}$ with the obvious meaning.

The following lemma expresses the property that  subexpressions of well-typed expressions are themselves
well-typed in a type context that may contain more variables, introduced by inner blocks.
The equivalence relation on the variable of the expressions
induced by the partition determined by the lent and mutable groups is preserved.
However, more  variables may be added to existing groups, there
could be new groups, and due to rule \rn{t-swap}, 
the mutable group of the  subexpression may be different from the
one on the whole expression.
\begin{lemma}[Context]\label{lemma:subCtx}
If $\deriv:\AuxTypeCheck{\Gamma}{\LentLocked}{\MutGroup}{\StronglyLocked}{{\e}}{\T}$, $\e=\GenCtx{\e'}$ and
$\deriv':\AuxTypeCheck{\SubstFun{\Gamma}{\TypeEnv{\genCtx}}}{\LentLocked'}{\MutGroup'}{\StronglyLocked'}{{\e'}}{\T'}$ is a sub-derivation of $\deriv$, 
then
\begin{enumerate}  
\item 
\begin{itemize}
\item $\LessEq{(\ \LentLocked\ \xs\ ){\setminus}\dom{\TypeEnv{\genCtx}}}{\LentLocked'\ \xs'}$, 
\item $\dom{\TypeEnv{\genCtx}}=\emptyset$ implies 
$\LentLocked\ \xs=\LentLocked'\ \xs'$, and
\item
we may assume that the last rule applied in $\deriv'$ is a structural rule;
\end{itemize}
  \item if $\e''$ is such that $\deriv'':\TypeCheck{\SubstFun{\Gamma}{\TypeEnv{\genCtx}}}{\LentLocked'}{\StronglyLocked'}{{\e''}}{\T'}$, then substituting $\deriv''$ with $\deriv'$ in $\deriv$ we get a derivation for $\AuxTypeCheck{\Gamma}{\LentLocked}{\MutGroup}{\StronglyLocked}{\Ctx{\e'}}{\T}$.
\end{enumerate}
\end{lemma}
\begin{proof}
The proof is in~\ref{sect:proof-sound}.\qed
\end{proof}
Note that, if $\genCtx$ is $\emptyctx$ or $\FieldAccess{\emptyctx}{\f}$ or $\MethCall{\emptyctx}{\m}{\vs}$ or 
$\MethCall{\val}{\m}{\vs\ \emptyctx\ \vs'}$ or $\FieldAssign{\emptyctx}{\f}{\val}$ or 
$\FieldAssign{\val}{\f}{\emptyctx}$ or $\ConstrCall{\C}{\vs\ \emptyctx\ \vs'}$, then  $ \dom{\genCtx}=\emptyset$ and 
\refToLemma{subCtx} implies \refToLemma{nonStructural}. 

\begin{lemma}[Field access]\label{lemma:fieldAccess}
Let  $\fields{\C}=\Field{\Type{\mu_i}{\C_i}}{\f_1}\ldots\Field{\Type{\mu_n}{\C_n}}{\f_n}$.
\begin{enumerate} 
  \item If $\AuxTypeCheck{\Gamma}{\LentLocked}{\MutGroup}{\StronglyLocked}{\FieldAccess{\ConstrCall{\C}{\xs}}{\f_i}}{\Type{\mu}{\C_i}}$ where the last rule applied is \rn{t-field-access}, then $\TypeCheck{\Gamma}{\LentLocked}{\StronglyLocked}{\x_i}{\Type{\mu}{\C_i}}$.
  \item If $Aux\TypeCheck{\Gamma}{\LentLocked}{\MutGroup}{\StronglyLocked}{\FieldAccess{\stVal}{\f_i}}{\T}$ where the last rule applied in $\deriv$ is \rn{t-field-access},
%and $\deriv':\TypeCheck{\Gamma}{\LentLocked}{\StronglyLocked}{{\stVal}}{\Type{\mu}{\C}}$ is a sub-derivation of $\deriv$, then
$\AuxTypeCheck{\Gamma}{\LentLocked}{\MutGroup}{\StronglyLocked}{\fieldOf{\stVal}{i}}{\T}$.
\item If $\AuxTypeCheck{\Gamma}{\LentLocked}{\MutGroup}{\StronglyLocked}{\FieldAccess{\stVal}{\f_i}}{\Type{\imm}{\C_i}}$, then for all $\LentLocked'$, $\StronglyLocked'$ such that $\WellFormedTypeCtx{\Gamma;\LentLocked';\StronglyLocked'}$ we have that
$\AuxTypeCheck{\Gamma}{\LentLocked'}{\MutGroup'}{\StronglyLocked'}{\fieldOf{\stVal}{i}}{\Type{\imm}{\C_i}}$.
\end{enumerate}
\end{lemma}
\begin{proof}
The proof is in~\ref{sect:proof-sound}.\qed
\end{proof}
    \begin{lemma}[Field assign]\label{lemma:fieldAssign}
Let $\deriv:\AuxTypeCheck{\Gamma}{\LentLocked}{\MutGroup}{\StronglyLocked}{\FieldAssign{\x}{\f}{\valPrime}}{\T}$, and the
last rule applied in $\deriv$ be \rn{t-field-assign}. 
Then, $\x\in\xs$ and for all 
$\y\in\FV{\valPrime}$ such that $\TypeCheck{{\Gamma}}{\LentLocked}{\StronglyLocked}{\y}{\Type{\mu}{\_}}$ with
$\mu\geq\mutable$ we have that $\y\in\xs$.
\end{lemma}
\begin{proof}
The proof is in~\ref{sect:proof-sound}.\qed
\end{proof}

The following theorem asserts that reduction preserves typability of expressions. The theorem is proved
by cases on the reduction rule used.
Here we give the proof for the two difficult cases, which are the application of \rn{Field-Access}, and \rn{Field-Assign-Move}.
In the appendix we present the proof for other interesting cases, that use similar techinques.
The difficulty with the proof of subject reduction, for \rn{Field-Access}, is the non standard semantics of  the 
construct, that replaces the ``field access expression'' with the ``value of the
field''. The ``value of the field'' must be given the same type it had in the context of the declaration of the variable.
For the rule \rn{Field-Assign-Move}, the difficult cases are when declarations are moved outside a block to which
the recovery rules are applied. As highlighted in the proof, in these cases, the derivation contained application 
of the rule \rn{t-swap} (if it was a $\capsule$) and \rn{t-unrst} (if it was $\imm$), which are still applicable after
moving the declarations.

\begin{theorem}[Subject Reduction]\label{theo:subjectReductionShort}
Let $\TypeCheckGround{\e}{\T}$, and
$\reduce{{\e}}{{\e'}}$. Then $\TypeCheckGround{\e'}{\T}$.
\end{theorem}
\begin{proof} 
Let $\TypeCheckGround{\e}{\T}$. If $\reduce{{\e}}{{\e'}}$, then one of the
rules of \refToFigure{reduction} was applied. 
Here we prove the result for the most interesting rules: \rn{field-access}, and \rn{field-assign-move}, whose proofs embodies
the techniques used. In the \refToSection{proof-sound} we present two other interesting cases  \rn{field-assign}, and \rn{mut-move}, which are easier, but still interesting.

\medskip\noindent
Consider \underline{rule \rn{field-access}}.
In this case  
\begin{enumerate} [(1)]
\item $\e=\Ctx{\FieldAccess{{\val}}{\f}}$, and 
\item $\e'={\Ctx{\extractField{\ctx}{\val}{i}}}$, 
\end{enumerate}
where 
$\fields{\C}=\Field{\Type{\mu_i}{\C_i}}{\f_1}\ldots\Field{\Type{\mu_n}{\C_n}}{\f_n}\ \mbox{and}\  \f=\f_i$.\\
From (1) and  \refToLemma{subCtx}.1 for some $\T'$, $\LentLocked$ and $\StronglyLocked$ we have that
\begin{enumerate}[(1)]\addtocounter{enumi}{2}
  \item $\AuxTypeCheck{{\TypeEnv{\ctx}}}{\LentLocked}{\MutGroup}{\StronglyLocked}{\FieldAccess{{\val}}{\f}}{\T'}$ and
  \item the last rule appled in the derivation is \rn{Field-Access}.
\end{enumerate} 
From Proposition \ref{lemma:congruenceValue} either $\val\cong\z$ for some $\z$ or  $\val\cong\stVal$ for some 
\storableVal\  $\stVal$ such that $\WFrv{\stVal}$. \\
If \underline{$\val\cong\stVal$}, from (3), (4), \refToLemma{fieldAccess}.2 and  $\extractField{\ctx}{\val}{i}=\fieldOf{\stVal}{i}$ we have that 
$\TypeCheck{{\TypeEnv{\ctx}}}{\LentLocked}{\StronglyLocked}{\fieldOf{\stVal}{i}}{\T'}$. Therefore,
from \refToLemma{subCtx}.2 we derive $\TypeCheckGround{\e'}{\T}$.\\
If \underline{$\val\cong\x$}, since $\extractField{\ctx}{\x}{i}$ is defined $\e=\CtxP{\Block{\decs'}{\val'}}$ where 
\begin{enumerate}[(a)]
  \item $\decs'=\dvs\ \Dec{\Type{\mu_x}{\C}}{\x}{\stVal_x}\ \Dec{\T_z}{\z}{\ctx_z[\FieldAccess{{\x}}{\f}]}\ \decs$
  \item $\WFdv{\Dec{\Type{\mu_x}{\C}}{\x}{\stVal_x}}$
  \item $\noCapture{\x}{\HB{\ctx_z}}$ and we may assume that $\noCapture{\stVal_x}{\HB{\ctx_z}}$.
\end{enumerate}
Let $\Gamma'=\TypeEnv{\ctxP}[\TypeEnv{\decs'}]$, from Lemmas \ref{lemma:subCtx} and \ref{lemma:inversionBlock} for some $\LentLocked'$ and $\StronglyLocked'$ 
\begin{enumerate} [(a)]\addtocounter{enumi}{3}
\item $\AuxDecsOK{\Gamma'}{\LentLocked'}{\MutGroup'}{\StronglyLocked'}{\Dec{\Type{\mu_x}{\C}}{\x}{\stVal_x}}$, i.e., 
\begin{enumerate}[i]
\item $\AuxTypeCheck{\Gamma'}{\LentLocked'}{\MutGroup'}{\StronglyLocked'}{{\stVal_x}}{{\Gamma'}({x})}$ if $\mu_x\not=\lent$
\item $\AuxTypeCheck{\Gamma'}{\LentLocked_x}{\MutGroup_x}{\StronglyLocked'}{\stVal_x}{{\Gamma'}({x})}$ if $\mu_x=\lent$ and $\LentLocked'\ \xs'=\LentLocked_x\ \xs_x$ where $\x\in\xs_x$
\end{enumerate}
\item $\LessEq{(\LentLocked'\ \xs'){\setminus}\dom{\TypeEnv{\ctx_z}}}{\LentLocked\ \xs}$.
\end{enumerate}
From (3), (4) and rule \rn{Field-Access} we derive that 
\begin{enumerate} [(a)]\addtocounter{enumi}{5}
 \item $\AuxTypeCheck{{\TypeEnv{\ctx}}}{\LentLocked}{\MutGroup}{\StronglyLocked}{\x}{\Type{\mu'}{\C}}$ and
  \item $\T'=\Type{\mu}{\C_i}$ where: if either $\mu_i=\imm$ or $\mu'=\imm$, then $\mu=\imm$, else $\mu=\mu'$.
\end{enumerate}
Consider the case \underline{$\T'=\Type{\imm}{\C_i}$}: either $\mu'=\imm$ or $\mu_i=\imm$.\\
If \underline{$\mu'=\imm$}, $\TypeCheck{{\TypeEnv{\ctx}}}{\LentLocked}{\StronglyLocked}{\x}{\Type{\imm}{\C}}$. From Lemmas \ref{lemma:typeVars} and \ref{lemma:weakening} we have that also $\TypeCheck{{\Gamma'}}{\LentLocked'}{\StronglyLocked'}{\x}{\Type{\imm}{\C}}$, therefore from (d).i we derive that
$\AuxTypeCheck{\Gamma'}{\LentLocked'}{\MutGroup'}{\StronglyLocked'}{{\stVal_x}}{\Type{\imm}{\C}}$. 
Applying rule \rn{Field-Access},
$\TypeCheck{\Gamma'}{\LentLocked'}{\StronglyLocked'}{\FieldAccess{\stVal_x}{\f_i}}{\Type{\imm}{\C_i}}$.
From (c), (e), and \refToLemma{weakening} we have that also 
$\TypeCheck{\TypeEnv{\ctx}}{\LentLocked}{\StronglyLocked}{\FieldAccess{\stVal_x}{\f_i}}{\Type{\imm}{\C_i}}$.
From \refToLemma{fieldAccess}.2, we get that
$\TypeCheck{\TypeEnv{\ctx}}{\LentLocked}{\StronglyLocked}{\fieldOf{\stVal_x}{i}}{\Type{\imm}{\C_i}}$. Since 
$\extractField{\ctx}{\x}{i}=\fieldOf{\stVal_x}{i}$, from \refToLemma{subCtx}.2 we derive $\TypeCheckGround{\e'}{\T}$.\\
If \underline{$\mu_i=\imm$}, from (d).i (or (d).ii) and rule \rn{Field-Access} we derive
$\AuxTypeCheck{\Gamma'}{\LentLocked'}{\MutGroup'}{\StronglyLocked'}{\FieldAccess{\stVal_x}{\f_i}}{\Type{\imm}{\C_i}}$
(or $\AuxTypeCheck{\Gamma'}{\LentLocked_x}{\MutGroup_x}{\StronglyLocked'}{\FieldAccess{\stVal_x}{\f_i}}{\Type{\imm}{\C_i}}$). From 
(e), \refToLemma{fieldAccess}.3 and \refToLemma{weakening} we get $\TypeCheck{{\TypeEnv{\ctx}}}{\LentLocked}{\StronglyLocked}{\fieldOf{\stVal_x}{i}}{\Type{\imm}{\C_i}}$, which implies, as for the previous case, $\TypeCheckGround{\e'}{\T}$.\\
Consider the case \underline{$\T'=\Type{\mu'}{\C_i}$} where $\mu'\neq\imm$ and $\mu_i\neq\imm$, and, since we
do not allow forward references to unevaluated declarations, also $\mu'\neq\capsule$. Therefore
$\mu'\geq\mutable$ and $\mu_i=\mutable$. From $\mu'\geq\mutable$ and (b) the declaration of $\x$ is of the shape
\begin{enumerate} [(a)]\addtocounter{enumi}{7}
  \item $\Dec{\Type{\mu_x}{C}}{\x}{\ConstrCall{\C}{\xs}}$ where $\mu_x\geq\mutable$ and $\extractField{\ctx}{\x}{i}=\x_i$.
\end{enumerate}
If $\mu_x=\mutable$, then $\Gamma'(\x)=\Type{\mutable}{\C}$. From (d).i we have
$\AuxTypeCheck{\Gamma'}{\LentLocked'}{\MutGroup'}{\StronglyLocked'}{\ConstrCall{\C}{\xs}}{\Type{\mutable}{\C}}$
So $\x\in\MutGroup'$ and from \refToLemma{constrMut} also $\x_i\in\MutGroup'$.\\
If $\mu_x=\lent$, then $\Gamma'(\x)=\Type{\mutable}{\C}$. From (d).ii we have that $\AuxTypeCheck{\Gamma'}{\LentLocked_x}{\MutGroup_x}{\StronglyLocked'}{\ConstrCall{\C}{\xs}}{\Type{\mutable}{\C}}$
where $\x\in\MutGroup_x$. Again from \refToLemma{constrMut} also $\x_i\in\MutGroup_x$.\\
From (e) and (c) we have that $\x$ and $\x_i$ are in the same group in $\LentLocked\ \xs$. Therefore,
$\AuxTypeCheck{{\TypeEnv{\ctx}}}{\LentLocked}{\MutGroup}{\StronglyLocked}{{\x_i}}{\T'}$, and 
from \refToLemma{subCtx}.2 we derive $\TypeCheckGround{\e'}{\T}$.\\
Finally, if $\mu_x=\readable$, then $\Gamma'(\x)=\Type{\readable}{\C}$, and from (d).i and rule \rn{Field-Access}
we derive
$\AuxTypeCheck{\Gamma'}{\LentLocked'}{\MutGroup'}{\StronglyLocked'}{\FieldAccess{\stVal_x}{\f_i}}{\Type{\readable}{\C_i}}$. From $\AuxTypeCheck{\Gamma'}{\LentLocked'}{\MutGroup'}{\StronglyLocked'}{\x_i}{\Type{\readable}{\C_i}}$. and \refToLemma{weakening} we get $\TypeCheck{{\TypeEnv{\ctx}}}{\LentLocked}{\StronglyLocked}{\x_i}{\Type{\readable}{\C_i}}$, which implies, as for the previous case, $\TypeCheckGround{\e'}{\T}$.

\medskip\noindent
Consider \underline{rule \rn{field-assign-move}}.
In this case  
\begin{enumerate}[(1)]
\item $\e=\Ctx{\Block {\dvs'\ \Dec{\Type{\mu}{\C}}{\z'} {\e_1}\ \decs' }{\val'}}$, and 
\item $\e'=\Ctx{\Block{\dvs'\ \dvs\ \Dec{\Type{\mu}{\C}}{\z'} {\e_2}\  \decs'}{\val'}}$, 
\end{enumerate}
where
\begin{itemize}
\item $\e_1={\Block{\dvs\ \Dec{\Type{\mu_z}{\C_z}}{\z}{\CtxP{\FieldAssign{\x}{\f}{\valPrime}}}\ \decs}{\val}}$, 
\item $\e_2={\Block{\Dec{\Type{\mu_z}{\C_z}}{\z}{\CtxP{\FieldAssign{\x}{\f}{\valPrime}}}\ \decs}{\val}}$,
\end{itemize}
$\fields{\C}=\Field{\T_1}{\f_1}\ldots\Field{\T_n}{\f_n}$ with $\f=\f_i$ and
\begin{enumerate}[(1)]\addtocounter{enumi}{2}
\item $\FV{\valPrime}\cap\dom{\dvs}=\zs\neq\emptyset$,
\item $\noCapture{\x}{\HB{\ctxP}\cup\dom{\dvs}}$, $\noCapture{\valPrime}{\HB{\ctxP}}$,
$\noCapture{\Block{\dvs'\ \dvs}{\val'}}{\dom{\dvs}}$,   
\item $\z'\not\in\dom{\dvs}$ and $\Reduct{(\dvs\ \decs)}{{\zs}}=\dvs$.
\end{enumerate}
Moreover, since forward definitions are only allowed to evaluated declarations, we have that 
\begin{enumerate}[(1)]\addtocounter{enumi}{5}
\item $\z\not\in\FV{\dvs}$.
\end{enumerate}
From $\TypeCheckGround{\e}{\T}$ and \refToLemma{subCtx}.1 we get that, for some $\T_b$, $\LentLocked''$, and $\StronglyLocked''$
\begin{itemize}
  \item [$(\ast)$]$
\AuxTypeCheck{{\TypeEnv{\ctx}}}{\LentLocked''}{\MutGroup''}{\StronglyLocked''}{\Block {\dvs'\ \Dec{\Type{\mu_z}{\C}}{\z'} {\e_1}\ \decs' }{\val'}}{\T_b}.
$
\end{itemize}
Let $\Gamma=\TypeEnv{\dvs'},\TypeEnv{\decs'},\TypeDec{\T'}{z'}$ 
From $(\ast)$ and  \refToLemma{inversionBlock} for some $\T'$, $\mu'$, $\LentLocked'$ and $\StronglyLocked'$
\begin{enumerate} [(A)]
\item $\AuxDecsOK{{\SubstFun{\TypeEnv{\ctx}}{\Gamma}}}{\LentLocked'}{\MutGroup'}{\StronglyLocked'}{\dvs'\ \decs'}$,
\item $\AuxDecsOK{{\SubstFun{\TypeEnv{\ctx}}{\Gamma}}}{\LentLocked'}{\MutGroup'}{\StronglyLocked'}{\Dec{\Type{\mu}{\C}}{\z'} {e_1}}$, i.e.
\begin{enumerate}[i]
\item $\AuxTypeCheck{{\SubstFun{\TypeEnv{\ctx}}{\Gamma}}}{\LentLocked'}{\MutGroup'}{\StronglyLocked'}{e_1}{{\SubstFun{\TypeEnv{\ctx}}{\Gamma}}(\z')}$ if $\mu\neq\lent$
\item $\AuxTypeCheck{{\SubstFun{\TypeEnv{\ctx}}{\Gamma}}}{\LentLocked_{\z'}}{\MutGroup_{\z'}}{\StronglyLocked'}{\e_1}{{\SubstFun{\TypeEnv{\ctx}}{\Gamma}}(\z')}$ if $\mu=\lent$ and $\LentLocked'\ \xs'=\LentLocked_{\z'}\ \xs_{\z'}$ where $\z'\in\xs_{\z'}$
\end{enumerate}
\item $\TypeCheck{\SubstFun{\TypeEnv{\ctx}}{\Gamma}}{\LentLocked'}{\StronglyLocked'}{\val'}{\T'}$ and
\item $\LessEq{(\LentLocked''\ \xs''){\setminus}\dom{\Gamma}}{\LentLocked'\ \xs'}$.
\end{enumerate}
If \underline{$\mu\geq\mutable$} , we can give a proof similar to the 
case of {rule \rn{mut-move}}. Therefore we can assume that $\mu=\capsule$
or $\mu=\imm$. \\
In both cases $\mu\neq\lent$, so from (B).i we have that
\begin{itemize}
\item [(B1)] $\AuxTypeCheck{{\SubstFun{\TypeEnv{\ctx}}{\Gamma}}}{\LentLocked'}{\MutGroup'}{\StronglyLocked'}{e_1}{\Type{\mu}{\C}}$
\end{itemize}
Consider first \underline{$\mu=\capsule$}. \\
In this case, from \refToLemma{typeStruct}.2 we can assume that the last rule applied in the derivation of $\e_1$ is
\rn{t-capsule}:
\[
\prooftree
\AuxTypeCheck{{\SubstFun{\TypeEnv{\ctx}}{\Gamma}}}{\LentLocked'\ \xs'}{\emptyset}{\StronglyLocked'}{e_1}{\Type{\mutable}{\C}}
\justifies
\AuxTypeCheck{{\SubstFun{\TypeEnv{\ctx}}{\Gamma}}}{\LentLocked'}{\MutGroup'}{\StronglyLocked'}{e_1}{\Type{\capsule}{\C}}
\using
\rn{t-capsule}
\endprooftree
\]
Let $\Gamma'=\TypeEnv{\dvs},\TypeEnv{\decs},\TypeDec{\T}{\z}$. From \refToLemma{inversionBlock}, for some $\LentLocked$ and $\StronglyLocked$
we have that
\begin{enumerate} [(a)]
\item $\AuxDecsOK{{\SubstFun{\TypeEnv{\ctx}}{\SubstFun{\Gamma}{\Gamma'}}}}{\LentLocked}{\MutGroup}{\StronglyLocked}{\dvs}$,
\item $\AuxDecsOK{{\SubstFun{\TypeEnv{\ctx}}{\SubstFun{\Gamma}{\Gamma'}}}}{\LentLocked}{\MutGroup}{\StronglyLocked}{\decs}$,
\item $\AuxDecsOK{{\SubstFun{\TypeEnv{\ctx}}{\SubstFun{\Gamma}{\Gamma'}}}}{\LentLocked}{\MutGroup}{\StronglyLocked}{\Dec{\T_z}{\z}{\CtxP{\FieldAssign{\x}{\f}{\valPrime}}}}$, i.e.
\begin{enumerate}[i.]
\item $\AuxTypeCheck{{\SubstFun{\TypeEnv{\ctx}}{\SubstFun{\Gamma}{\Gamma'}}}}{\LentLocked}{\MutGroup}{\StronglyLocked}{\CtxP{\FieldAssign{\x}{\f}{\valPrime}}}{{\SubstFun{\TypeEnv{\ctx}}{\Gamma}}(\z)}$ if $\mu_z\neq\lent$
\item $\AuxTypeCheck{{\SubstFun{\TypeEnv{\ctx}}{\SubstFun{\Gamma}{\Gamma'}}}}{\LentLocked_{\z}}{\MutGroup_{\z}}{\StronglyLocked'}{\CtxP{\FieldAssign{\x}{\f}{\valPrime}}}{{\SubstFun{\TypeEnv{\ctx}}{\Gamma}}(\z)}$ if $\mu_z=\lent$ and $\LentLocked\ \xs=\LentLocked_{\z}\ \xs_{\z}$ where $\z'\in\xs_{\z'}$
\end{enumerate}
\item $\AuxTypeCheck{{\SubstFun{\TypeEnv{\ctx}}{\SubstFun{\Gamma}{\Gamma'}}}}{\LentLocked}{\MutGroup}{\StronglyLocked}{\val}{\Type{\mutable}{\C}}$
\item $\LessEq{(\LentLocked'\ \xs'){\setminus}\dom{\Gamma'}}{\LentLocked\ \xs}$. 
%\item $\mu'\leq\mu$.
\end{enumerate}
From (c) and \refToLemma{subCtx}.1 we have that, for some $\LentLocked_x$ and $\StronglyLocked_x$
\begin{itemize}
\item $\AuxTypeCheck{{\SubstFun{\TypeEnv{\ctx}}{\SubstFun{\Gamma}{\SubstFun{\Gamma'}{\TypeEnv{\ctxP}}}}}}{\LentLocked_x}{\MutGroup_x}{\StronglyLocked_x}{\FieldAssign{\x}{\f}{\valPrime}}{\T_i}$
\end{itemize}
From \refToLemma{fieldAssign} we derive that $\x\in\xs_x$, and, 
if $Y=\{\y\ |\ \dvs(\y)=\Dec{\Type{\mu_y}{\_}}{\y}{\_}\ \mu_y\geq\mutable\}$, then
for all $\y\in Y$ we have that $\y\in\xs_x$.\\
From (3), (4),  $\noCapture{\x}{\dom{\decs}}$ (forward references are only allowed to
evaluated declarations), (e) and (D) we have that $\{\x\}\cup Y$ is a subset of one of the groups 
in $\LentLocked\ \xs$. 
(Note that, in the derivation of the judgement (c).i or (c).ii, there must be an application of rule \rn{t-swap} to
make $\xs_x$ the mutable group, since $\x$ is in one of the groups of $\LentLocked\ \xs$.)
Define $\LentLocked^{\ast}$ and $\StronglyLocked^{\ast}$ as follows. If there is $\zs'\in\LentLocked'$ such that
$\x\in\zs'$, then $\LentLocked^{\ast}=(\LentLocked'-\zs')\ (\zs'\cup Y)$, otherwise  
$\LentLocked^{\ast}=\LentLocked'$. If $\x\in\StronglyLocked'$ then $\StronglyLocked^{\ast}=\StronglyLocked'\cup Y$,
otherwise $\StronglyLocked^{\ast}=\StronglyLocked'$.\\
From (4), (A), (B), (C) and \refToLemma{weakening}, we derive
\begin{itemize}
\item [(A1)] $\AuxDecsOK{{\SubstFun{\TypeEnv{\ctx}}{\Gamma,\TypeEnv{\dvs}}}}{\LentLocked^{\ast}}{\MutGroup^{\ast}}{\StronglyLocked^{\ast}}{\dvs'\ \decs'}$,
\item [(C1)] $\AuxTypeCheck{\SubstFun{\TypeEnv{\ctx}}{\Gamma,\TypeEnv{\dvs}}}{\LentLocked^{\ast}}{\MutGroup^{\ast}}{\StronglyLocked^{\ast}}{\val'}{\T'}$
\item [(D1)]  $\LessEq{(\LentLocked'\ \xs'){\setminus}{\dom{\dvs}}}{\LentLocked^{\ast}\ \xs^{\ast}}$ and  $\StronglyLocked'{\setminus}{\dom{\dvs}}=\StronglyLocked^{\ast}{\setminus}{\dom{\dvs}}$.
\end{itemize}
From (a), (4), and \refToLemma{weakening}, we have
\begin{itemize}
\item [(a1)] $\AuxDecsOK{{\SubstFun{\TypeEnv{\ctx}}{{\Gamma,\TypeEnv{\dvs}}}}}{\LentLocked^{\ast}}{\MutGroup^{\ast}}{\StronglyLocked^{\ast}}{\dvs}$
\end{itemize}
From (4) and the fact that, for well-formedness of
declarations $x\not\in\dom{\dvs}$, we derive
that $\SubstFun{\Gamma}{\TypeEnv{\dvs},\TypeEnv{\decs},\TypeDec{\T'}{\z'}}=\SubstFun{\Gamma,\TypeEnv{\dvs}}{\TypeEnv{\decs},\TypeDec{\T'}{\z'}}$.\\
From (b), (c), (d), (e) and  rule \rn{t-block},  we have that 
\begin{itemize}
\item$\AuxTypeCheck{{\SubstFun{\TypeEnv{\ctx}}{\Gamma,\TypeEnv{\dvs}}}}{\LentLocked'\ \xs'}{\emptyset}{\StronglyLocked'}{\Block{\Dec{\Type{\mu_z}{\C_z}}{\z}{\CtxP{\FieldAssign{\x}{\f}{\valPrime}}}\ \decs}{\val}}
{\Type{\mutable}{\C}}
$
\end{itemize}
and therefore applying rule \rn{t-capsule} we derive $\AuxDecsOK{{\SubstFun{\TypeEnv{\ctx}}{\Gamma,\TypeEnv{\dvs}}}}{\LentLocked'}{\MutGroup'}{\StronglyLocked'}{\e_2}
$. 
From (D1) and \refToLemma{weakening} we also have 
\begin{itemize}
\item [(B2)] $\AuxDecsOK{{\SubstFun{\TypeEnv{\ctx}}{\Gamma,\TypeEnv{\dvs}}}}{\LentLocked^{\ast}}{\MutGroup^{\ast}}{\StronglyLocked^{\ast}}{\Dec{\Type{\mu_z}{\C}}{\z'} {\e_2}}
$. 
\end{itemize} 
From (A1), (a1), (B2), (C1),  (D), and (D1), applying rule \rn{t-block}, we derive
\begin{center}
$\AuxTypeCheck{{\TypeEnv{\ctx}}}{\LentLocked''}{\MutGroup''}{\StronglyLocked''}{\Block{\dvs'\ \dvs\ \Dec{\Type{\mu}{\C}}{\z'}{\e_2}\ \decs'}{\val'}}{\T_b}$.
\end{center}
From \refToLemma{subCtx}.2, we obtain the result.\\
Consider now \underline{$\mu=\imm$}.\\
If the typing is obtained from \rn{t-capsule} followed by \rn{t-sub} the result follows from the
previous proof. If instead the last rule applied was \rn{t-imm}, from \refToLemma{typeStruct}.2 
\[
\prooftree
\AuxTypeCheck{{\SubstFun{\TypeEnv{\ctx}}{\Gamma}}}{\LentLocked'\ \xs'}{\emptyset}{\domGeqMut{\SubstFun{\TypeEnv{\ctx}}{\Gamma}}}{e_1}{\Type{\readable}{\C}}
\justifies
\AuxTypeCheck{{\SubstFun{\TypeEnv{\ctx}}{\Gamma}}}{\LentLocked'}{\MutGroup'}{\StronglyLocked'}{e_1}{\Type{\imm}{\C}}
\using
\rn{t-imm}
\endprooftree
\]
Let
$\Gamma'=\TypeEnv{\dvs},\TypeEnv{\decs},\TypeDec{\T'}{\z'}$.  From \refToLemma{inversionBlock}, and the fact that
there is no application of \rn{t-unrst} in the derivation of $\e_1$ (if there was, then  the type derived for $\e_1$ should be $\Type{\imm}{\C}$), for some $\LentLocked$, and $\ys$ 
\begin{enumerate} [(a)]
\item $\AuxDecsOK{{\SubstFun{\TypeEnv{\ctx}}{\SubstFun{\Gamma}{\Gamma'}}}}{\LentLocked}{\MutGroup}{\StronglyLocked}{\dvs}$,
\item $\AuxDecsOK{{\SubstFun{\TypeEnv{\ctx}}{\SubstFun{\Gamma}{\Gamma'}}}}{\LentLocked}{\MutGroup}{\StronglyLocked}{\decs}$,
\item $\AuxDecsOK{{\SubstFun{\TypeEnv{\ctx}}{\SubstFun{\Gamma}{\Gamma'}}}}{\LentLocked}{\MutGroup}{\StronglyLocked}{\Dec{\T_z}{\z}{\CtxP{\FieldAssign{\x}{\f}{\valPrime}}}}$, i.e.
\begin{enumerate}[i.]
\item $\AuxTypeCheck{{\SubstFun{\TypeEnv{\ctx}}{\SubstFun{\Gamma}{\Gamma'}}}}{\LentLocked}{\MutGroup}{\StronglyLocked}{\CtxP{\FieldAssign{\x}{\f}{\valPrime}}}{{\SubstFun{\TypeEnv{\ctx}}{\Gamma}}(\z)}$ if $\mu_z\neq\lent$
\item $\AuxTypeCheck{{\SubstFun{\TypeEnv{\ctx}}{\SubstFun{\Gamma}{\Gamma'}}}}{\LentLocked_{\z}}{\MutGroup_{\z}}{\StronglyLocked}{\CtxP{\FieldAssign{\x}{\f}{\valPrime}}}{{\SubstFun{\TypeEnv{\ctx}}{\Gamma}}(\z)}$ if $\mu_z=\lent$ and $\LentLocked\ \xs=\LentLocked_{\z}\ \xs_{\z}$ where $\z'\in\xs_{\z'}$
\end{enumerate}
\item $\AuxTypeCheck{{\SubstFun{\TypeEnv{\ctx}}{\SubstFun{\Gamma}{\Gamma'}}}}{\LentLocked}{\MutGroup}{\StronglyLocked}{\val}{\Type{\readable}{\C}}$
\item $\LessEq{(\LentLocked'\ \xs'){\setminus}\dom{\Gamma'}}{\LentLocked\ \xs}$
\item $\StronglyLocked=\domGeqMut{\SubstFun{\TypeEnv{\ctx}}{\Gamma}}{\setminus}\dom{\Gamma'}$
%\item $\mu'\leq\mu$.
\end{enumerate}
As before, we can assume that  $\noCapture{\x}{\dom{\decs}}$. From (c), (4), \refToLemma{subCtx}.1, and rule \rn{t-field-assign}, we derive that $\x\in\domMut{\SubstFun{\TypeEnv{\ctx}}{\Gamma}}$. So $\x\in\domGeqMut{\SubstFun{\TypeEnv{\ctx}}{\Gamma}}$.
Therefore, in the derivation of the judgement (c).i or (c).ii, there must be an application of rule \rn{t-unrst} to make possible the
application of rule \rn{t-var} to $\x$.\\
The proof of this case now proceeds as for the case of $\mu=\capsule$ defining $\LentLocked^{\ast}$ and $\StronglyLocked^{\ast}$ and proving that $\e_2$ and then the resulting expression are typeable.
\qed
\end{proof}

To prove progress, given an expression we need to find a rule that may be applied, and prove that 
its side conditions are verified. To match an expression with the left-side of a rule we define
 {\em pre-redexes}, $\preRedex$, by:
\begin{center}
$
\begin{array}{ll}
\preRedex ::=& \FieldAccess{\val}{\f}\mid\MethCall{{\val}}{\m}{{\vals}}\mid\FieldAssign{\val}{\f}{\val'}
  \\
 &
 \mid \Block{\dvs\ \Dec{\T}{\x}{\val}\ \decs}{\val}\quad\WFdv{\dvs}\ \wedge\ \not\WFdv{\Dec{\T}{\x}{\val}}\\
\end{array}
$
\end{center}
Any expression can be uniquely decomposed in a in a context filled a pre-redex.
\begin{lemma}[Unique Decomposition]\label{lemma:decomposition}
Let $\e$ be an expression. Either $\congruence{\e}{\val}$ where $\val$ is well-formed,
and if $\val=\Block{\dvs}{\cOrx}$, then $\WFdv{\dvs}$,
or there are $\ctx$ and $\preRedex$ such that $\congruence{\e}{\Ctx{\preRedex}}$.
\end{lemma}
\begin{proof}
The proof is in~\ref{sect:proof-sound}.\qed
\end{proof}
The progress result is proved by structural induction on expressions. The interesting case
is field assignment, in which we have to prove that one of the three rules may be applied. 
\begin{theorem}[Progress]\label{theo:progress}
Let $\TypeCheckGround{\e}{\T}$. Then either $\congruence{\e}{\val}$ where $\val$ is well-formed,
and if $\val=\Block{\dvs}{\cOrx}$, then  $\WFdv{\dvs}$, or $\reduce{\e}{\e'}$ for some $\e'$.
\end{theorem}
\begin{proof}
Let $\e$ be such that $\TypeCheckGround{\e}{\T}$
for some $\T$, and for no $\val$ we have that
$\congruence{\e}{\val}$ where $\val$ is well-formed,
and if $\val=\Block{\dvs}{\cOrx}$, then $\WFdv{\dvs}$.
By \refToLemma{decomposition}, for some $\ctx$ and $\preRedex$
we have that $\congruence{\e}{\Ctx{\preRedex}}$, and so
$\TypeCheckGround{\Ctx{\preRedex}}{\T}$.
From \refToLemma{subCtx}, 
$\TypeCheck{{\TypeEnv{\ctx}}}{\LentLocked}{\StronglyLocked}{{\preRedex}}{\T'}$,
for some $\LentLocked$, $\StronglyLocked$,
$\T'$. \\
Let \underline{$\preRedex=\FieldAccess{\val}{\f}$}. 
From $\TypeCheck{\TypeEnv{\ctx}}{\LentLocked}{\StronglyLocked}{\FieldAccess{\val}{\f}}{\T'}$, we have that 
$\TypeCheck{\TypeEnv{\ctx}}{\LentLocked}{\StronglyLocked}{{{\val}}}{\Type{\_}{\C}}$ with 
$\fields{\C}=\Field{\T_1}{\f_1}\ldots\Field{\T_n}{\f_n}$ and $\f=\f_i$.
Therefore $\typeOf{\ctx,\val}=\Type{\_}{\C}$ and $\extractField{\ctx}{\val}{i}$ 
is defined. So rule \rn{Field-Access}
is applicable.\\
Let \underline{$\preRedex=\MethCall{{\val}}{\m}{{\vals}}$}, then rule \rn{Invk}
is applicable.\\
Let \underline{$\preRedex=\FieldAssign{\val}{\f}{\val'}$}, and $\notRef{\val}$. Then rule \rn{field-assign-prop} is applicable.\\
Let \underline{$\preRedex=\FieldAssign{\x}{\f}{\val}$}. From from $\TypeCheck{{\TypeEnv{\ctx}}}{\LentLocked}{\StronglyLocked}{{\FieldAssign{\x}{\f}{\val}}}{\T'}$, we have that
$\TypeCheck{{\TypeEnv{\ctx}}}{\LentLocked}{\StronglyLocked}{\x}{\Type{\mutable}{\C}}$ with
$\fields{\C}=\Field{\T_1}{\f_1}\ldots\Field{\T_n}{\f_n}$ and $\f=\f_i$. Since $\x\in\HB{\ctx}$, 
then $\ctx=\DecEvCtx{\ctx}{\x}[\Block{\dvs\ \Dec{\T}{\y}{\ctxP}\ \decs}{\valPrime}]$ for some $\ctxP$ and $\ctx_x$ such that
$\dvs(\x)=\dv$ and $\x\not\in\HB{\ctxP}$. From $\WFdv{\dv}$, we have that $\dv=\Dec{\Type{\mutable}{\C}}{\x}{\ConstrCall{\C}{{\xs}}}$.\\
There are two cases: either $\noCapture{\val}{\HB{\ctxP}}$, or for some $\z\in\FV{\val}$, $\z\in\HB{\ctxP}$.\\
In the first case rule \rn{field-assign} is applicable. \\
In the second, let $\ctxP={\ctx_z}[\Block{\dvs'\ \Dec{\T'}{\y'}{\ctx''}\ \decs'}{\valPrime'} ]$
such that $\noCapture{\val}{\HB{\ctx''}}$, and $\dvs'(\z)=\dv'$ and $\z\not\in\HB{\ctx''}$. \\
If $\ctx_z$ is $\emptyctx$, then \rn{field-assign-move} is applicable to
\begin{center}
$\DecEvCtx{\ctx}{\x}[\Block{\dvs\ \Dec{\T}{\y}{\Block{\dvs'\ \Dec{\T'}{\y'}{\ctx''[{\FieldAssign{\x}{\f}{\val}}]}\ \decs'}{\valPrime'}}\ \decs}{\valPrime}]$.
\end{center}
Otherwise,  $\ctx_z=\ctx'_z[\Block{\dvs''\ \Dec{\T''}{\y''}{\Block{\dvs'\ \Dec{\T'}{\y'}{\ctx''}\ \decs'}{\valPrime'}}\ \decs''}{\valPrime''}]$ and  
then \rn{field-assign-move} is applicable to
\begin{center}
$\DecEvCtx{\ctx}{\x}[\ctx'_z[\Block{\dvs''\ \Dec{\T''}{\y''}{\Block{\dvs'\ \Dec{\T'}{\y'}{\ctx''[{\FieldAssign{\x}{\f}{\val}}]}\ \decs'}{\valPrime'}}\ \decs''}{\valPrime''}]]$.
\end{center}
Therefore, there is always a rule applicable to $\ctx[\FieldAssign{\x}{\f}{\val}]$. \\
Let \underline{$\preRedex=\Block{\dvs\ \Dec{\T'}{\x}{\val'}\ \decs}{\val}$}, $\WFdv{\dvs}$ and $\not\WFdv{\Dec{\T'}{\x}{\val'}}$.
From Proposition \ref{lemma:congruenceValue}, either $\congruence{\val}{\x}$ or 
$\congruence{\val}{\stVal}$ for some $\WFrv{\stVal}$. In the first case rule \rn{Alias-Elim} is applicable. In the
second, let $\T'=\Type{\mu}{\D}$ for some, $\mu$ and $\D$. By cases on $\mu$.\\
If \underline{$\mu=\capsule$}, then \rn{Capsule-Elim} is applicable. \\
If \underline{$\mu\geq\mutable$}, then since $\not\WFdv{\Dec{\T'}{\x}{\stVal}}$ we have that $\stVal=\Block{\dvs'}{\val'}$.
By renaming bound variables in $\Block{\dvs'}{\val'}$ we can have that 
$\noCapture{\Block{\dvs\ \decs}{\val}}{\dom{\dvs'}}$. Therefore, rule \rn{Mut-Move} is applicable
moving $\dvs'$ outside.\\
If \underline{$\mu=\imm$},  
let $\dvs'=\dvs_{im}\ \dvs_{mt}$, where $\dv\in\dvs_{im}$ if $\extractMod{\dv}\leq\imm$,
and $\dv\in\dvs_{mt}$ if $\extractMod{\dv}\geq\mutable$. 
The side condition, 
$\noCapture{\Block{\dvs\ \decs}{\val}}{\dom{\dvs_{im}}}$, can be satisfied by renaming
of declared variables in $\Block{\dvs'}{\val'}$. We have to show that 
$\noCapture{\dvs_{im}}{\dom{\dvs_{mt}}}$. 
Let $\y\in\FV{\stVal'}$ for some $\Dec{\Type{\mu'}{\C'}}{\x'}{\stVal'}\in\dvs_{im}$ with $\mu'\leq\imm$. 
So $\TypeCheck{\SubstFun{\Gamma}{\TypeEnv{\dvs_{im}\ \dvs_{mt}}}}{\LentLocked''}{\StronglyLocked''}{\stVal'}{\Type{\mu'}{\C'}}$,
and from \refToTheorem{canonicalForm}.1 and 3, we have that
$\TypeCheck{\SubstFun{\Gamma}{\TypeEnv{\dvs_{im}\ \dvs_{mt}}}}{\LentLocked''}{\StronglyLocked''}{\y}{\Type{\imm}{\_}}$.
Therefore, $\y\not\in\dom{\dvs_{mt}}$, and rule \rn{Imm-Move} can be applied
since $\noCapture{\dvs_{im}}{\dom{\dvs_{mt}}}$ holds. \qed
\end{proof}

\subsection{Properties of expressions having immutable and capsule modifiers}

In addition to the standard soundness property, we prove two theorems stating that the $\capsule$ and $\imm$ qualifier, respectively, have the expected behaviour.  A nice consequence of our non standard operational model 
is that this can be formally expressed {and proved} in a simple way. 

In the two theorems, we need to trace the reduction of the right-hand side of a reference declaration.  
 To lighten the notation, we assume in the following that expressions contain at most one declaration for a variable (no shadowing, as can be always obtained by alpha-conversion).

We need some notations and lemmas. We define \emph{contexts} $\decctx{\mux}$:

\begin{quote}
\begin{grammatica}
\produzione{\decctx{\mux}}{\Block{\decs\ \Dec{\Type{\mu}{\C}}{\x}{\emptyctx}\ \decs'}{\val}\mid\FieldAccess{\decctx{\mux}}{\f}\mid\MethCall{\decctx{\mux}}{\m}{\vals}\mid\MethCall{\val}{\m}{\vals\ \decctx{\mux}\ \vals'}}{}\\*
\seguitoproduzione{\mid\FieldAssign{\decctx{\mux}}{\f}{\val}
\mid\FieldAssign{\val}{\f}{\decctx{\mux}}\mid \ConstrCall{\C}{\vals\ \decctx{\mux}\ \vals'}}{}\\*
\seguitoproduzione{\mid\Block{\decs\ \Dec{\T}{\y}{\decctx{\mux}}\ \decs'}{\val}\mid\Block{\decs}{\decctx{\mux}}}{}
\end{grammatica}
\end{quote}

That is, in $\Decctx{\mux}{\e}$ the expression $\e$ occurs as right-hand side of the (unique) declaration for reference $\x$, which has qualifier $\mu$. We will simply write $\Xctx$ when the $\mux$ suffix is not relevant.

The type assignment extracted from a context $\Xctx$, denoted $\TypeEnv{\Xctx}$, is defined as for the general contexts.
% \begin{quote}
% $\TypeEnv{\Block{\decs\ \Dec{{\T}}{\x}{\emptyctx}\ \decs'}{\val}}=
%\SubstFun{\TypeEnv{\decs\ \decs'}}{\TypeDec{\T}{\x}}$\\
%$\TypeEnv{\FieldAccess{\Xctx}{\f}}=\ldots=\TypeEnv{\ConstrCall{\C}{\vals\ \Xctx\ \vals'}}=\TypeEnv{\Xctx}$\\
% $\TypeEnv{\Block{\decs\ \Dec{\T}{\y}{\Xctx}\ \decs'}{\val}}=
%\SubstFun{ \SubstFun{\TypeEnv{\decs\ \decs'}}{\TypeDec{\T}{\y}}}{\TypeEnv{\Xctx}}$\\
%$\TypeEnv{\Block{\decs}{\Xctx}}=\SubstFun{\TypeEnv{\decs}}{\TypeEnv{\Xctx}}$
% \end{quote}
 
 The declaration for a variable $\y$ in a context $\Xctx$, denoted $\extractDec{\Xctx}{\y}$, can be defined analogously.
 
The following lemma states that the type of a reference under the type assignment extracted from the surrounding context is a subtype of the declared type  {when the type modifier is not equal to $\lent$. If the type modifier of the declaration is $\lent$, the variable
could be accessed in a sub-context in which, due to the \rn{t-swap} rule, the variable belongs to the current mutable
group and so we derive a type with the $\mutable$ modifier.}
\begin{lemma}\label{lemma:dec}
If $\TypeCheck{\TypeEnv{\Xctx}}{\LentLocked}{\StronglyLocked}{{\y}}{\Type{\mu}{\_}}$, then $\extractDec{\Xctx}{\y}=\Dec{\Type{\mu'}{\_}}{\y}{\_}$  {is such that $\mu'\neq\lent$ implies $\mu'\leq\mu$}.
\end{lemma}
\begin{proof}
By induction on $\Xctx$. \qed
\end{proof}
%Note that this lemma would trivially hold in a type system with no  {recovery}  {and swapping rules}, that is, with no way to move the type of an expression against the subtype hierarchy. 

 The following lemma states that a subexpression which occurs as right-hand-side of a declaration in a well-typed expression is well-typed, under the type assignment extracted from the surrounding context, and has a subtype of the reference type.
\begin{lemma}\label{lemma:decctx}
If $\IsWellTyped{\Decctx{\mux}{\e}}$, then $\TypeCheck{\TypeEnv{\decctx{\mux}}}{\LentLocked}{\StronglyLocked}{{\e}}{\Type{\mu'}{\_}}$ for some
$\LentLocked$, $\StronglyLocked$, and $\mu'\leq\mu$.
\end{lemma}
\begin{proof}
By induction on $\decctx{\mux}$. \qed
\end{proof}

The expected behaviour of the $\capsule$ qualifier is, informally,  that the reachable object subgraph denoted by a $\capsule$ reference should not contain nodes reachable from the outside, unless they are immutable.
In our calculus, a reachable object subgraph is a right-value $\stVal$, nodes reachable from the outside are free variables,  hence the condition can be formally expressed by requiring that free variables in $\stVal$ are declared $\imm$ or $\capsule$ in the surrounding context:
\begin{quote}
$\ImmClosed{\Xctx}{\stVal}$ iff for all $\y\in\FV{\stVal}$, $\extractDec{\Xctx}{\y}=\Dec{\Type{\mu}{\_}}{\y}{\_}$ with ${\mu}\leq\imm$
\end{quote}
Moreover, the reachable object subgraph denoted by a $\capsule$ reference should be typable $\mutable$, since it can be assigned to a mutable reference.
Altogether, the fact that the $\capsule$ qualifier guarantees the expected behaviour can be formally stated as in the theorem below, where the qualifier $\capsule$ is abbreviated \lstinline{c}.
\begin{theorem}[Capsule]\label{theo:capsule}
If $\IsWellTyped{\Capsulectx{\e}}$ and $\Capsulectx{\e}\longrightarrow^\star\CapsulectxP{\stVal}$, then:
\begin{itemize}
\item {$\typeOf{\stVal}=\Type{\mutable}{\_}$}
\item $\ImmClosed{{\cal C}'_{\cx}}{\stVal}$
\end{itemize}
\end{theorem}
\begin{proof}
By subject reduction (\refToTheorem{subjectReductionShort}) we get $\IsWellTyped{\CapsulectxP{\stVal}}$. 
Then, from \refToLemma{decctx}, $\TypeCheck{\Gamma}{\LentLocked}{\StronglyLocked}{\stVal}{\Type{\capsule}{\_}}$ with $\Gamma=\TypeEnv{{\cal C}'_{\cx}}$. Hence, from \refToTheorem{canonicalForm}, for all $\y\in\FV{\stVal}$, $\TypeCheck{\Gamma}{\LentLocked}{\StronglyLocked}{\y}{\Type{\imm}{\_}}${, and $\typeOf{\stVal}=\Type{\mutable}{\_}$. Hence,} by \refToLemma{dec},  $\extractDec{{\cal C}'_{\cx}}{\y}=\Dec{\Type{\mu}{\_}}{\y}{\_}$ with $\mu\leq\imm$.\qed
\end{proof}

Note that the context can change since it is not an evaluation context and, moreover, reduction can modify the store. Consider for instance $\Capsulectx{\e}$ to be the following expression:
\begin{lstlisting}
mut C y= new C(0); mut C z= new C(y.f=1); capsule C x= $\e$; ...
\end{lstlisting}
Before reducing to a right value the initialization expression of \lstinline{x}, the initialization expression of \lstinline{z} should be reduced, and this has a side-effect on the right value of \lstinline{y}. Hence $\CapsulectxP{\stVal}$ is:
\begin{lstlisting}
mut C y= new C(1); mut C z= new C(1); capsule C x= $\stVal$; ...
\end{lstlisting}

The expected behaviour of the $\imm$ qualifier is, informally, that the reachable object subgraph denoted by an $\imm$ reference should not be modified through any alias. 
Hence, the right value of an immutable reference:
\begin{itemize}
\item should not be modified
\item should only refer to external references which are immutable
\end{itemize}
as formally stated in the theorem below, where the qualifier $\imm$ is abbreviated \lstinline{i}.
 \begin{theorem}[Immutable]\label{theo:imm}
If $\IsWellTyped{\mbox{$\Immctx{\e}$}}$ and $\Immctx{\e}\longrightarrow^\star\ImmctxP{\stVal}$, then:
\begin{itemize}
\item $\ImmctxP{\stVal}\longrightarrow^\star\ImmctxS{\stVal'}$
implies $\stVal=\stVal'$
\item $\ImmClosed{{\cal C}'_{\ix}}{\stVal}$
\end{itemize}
\end{theorem}
\begin{proof}
The first property is directly ensured by reduction rules, since rule \rn{field-assign} is only applicable on $\geq\mutable$ references (note that the progress theorem guarantees that reduction cannot be stuck for this reason). The second property can be proved analogously to \refToTheorem{capsule} above.\qed
\end{proof}

Again, the context can change during reduction. Consider for instance $\Immctx{\e}$ to be the following expression:
\begin{lstlisting}
mut C y= new C(0); imm C x= $\e$; mut C z= new C(y.f=1); ...
\end{lstlisting}
In this case, the initialization expression of $\x$ is firstly reduced to a right value, hence $\ImmctxP{\stVal}\equiv\Immctx{\stVal}$ is:
\begin{lstlisting}
mut C y= new C(0); imm C x= $\stVal$; mut C z= new C(y.f=1); ...
\end{lstlisting}
In the following reduction steps, the context can change, for instance $\ImmctxS{\stVal}$ can be:
\begin{lstlisting}
mut C y= new C(1); imm C x= $\stVal$; mut C z= new C(1); ...
\end{lstlisting}
However, the right-value $\stVal$ cannot be modified.

%% file: RelatedWork.tex
\section{Related work}\label{sect:related}

\paragraph{{Recovery}} {The works most closely most closely related to ours are those based on the \emph{recovery} notion, that is, the type system of Gordon et al. \cite{GordonEtAl12} and the Pony language  \cite{ClebschEtAl15}.} Indeed, the capsule property has many variants in the literature, such as \emph{isolated} \cite{GordonEtAl12}, \emph{uniqueness} \cite{Boyland10} and \emph{external uniqueness}~\cite{ClarkeWrigstad03}, \emph{balloon} \cite{Almeida97,ServettoEtAl13a}, \emph{island} \cite{DietlEtAl07}. 
%The fact that aliasing can be controlled by using \emph{lent} (\emph{borrowed}) references is well-known~\cite{Boyland01,NadenEtAl12}.
However, before the work of Gordon et al. \cite{GordonEtAl12}, the capsule property was only ensured in simple situations, such as using a primitive deep clone operator, or composing subexpressions with the same property.

The important novelty of the type system of Gordon et al. \cite{GordonEtAl12} has been \emph{recovery}, that is, the ability to ensure properties (e.g., capsule or immutability) by keeping into account not only the expression itself but the way the surrounding context is used. {Notably,} an expression which does not use external mutable references is recognized to be a capsule. 
{In the Pony language \cite{ClebschEtAl15}  the ideas of Gordon et al. \cite{GordonEtAl12} are extended to a richer set of reference immutability permissions. In their terminology \texttt{value} is immutable, \texttt{ref} is mutable, \texttt{box} is similar to \emph{readonly} as often found in literature, different from our $\readable$ since it can be aliased. An ephemeral isolated reference \lstinline{iso^} is similar to a $\capsule$ reference in our calculus, whereas non ephemeral \texttt{iso} references offer destructive reads and are more
similar to isolated fields \cite{GordonEtAl12}. Finally, \texttt{tag} only allows object identity checks and \texttt{trn} (transition) is a subtype of \texttt{box} that can be converted to \texttt{value}, providing a way to create values without using isolated references. The last two qualifiers have no equivalent in our
work or in  \cite{GordonEtAl12}.}

Our {type system greatly enhances the recovery mechanism used in such previous work \cite{GordonEtAl12,ClebschEtAl15} by using lent references, and rules \rn{t-swap} and \rn{t-unrst}.} For instance, the examples in \refToFigure{TypingOne} and \refToFigure{TypingTwo} would be ill-typed in \cite{GordonEtAl12}. 

{A minor difference with the type systems of Gordon et al. \cite{GordonEtAl12} and Pony \cite{GordonEtAl12,ClebschEtAl15} is that we only allow fields to be $\mutable$ or $\imm$.
Allowing \emph{readonly} fields means holding a reference that is useful for observing but non making remote modifications. However, our type system supports the $\readable$ modifier rather than the \emph{readonly}, and the $\readable$ qualifier includes the $\lent$ restriction. Since something which is $\lent$ cannot be saved as part of a $\mutable$ object, $\lent$ fields are not compatible with the current design where objects are born $\mutable$. The motivation for supporting $\readable$ rather than \emph{readonly} is discussed in a specific point later.
Allowing $\capsule$ fields means that programs can store an externally unique object graph into the heap and decide later whether to unpack
 permanently or freeze the reachable objects.  This can be useful, but, as for $\readable$ versus \emph{readonly}, our opinion is that this power is hard to use for good, since
it requires destructive reads, as discussed in a specific point later. 
In most cases, the same expressive power can be achieved by having the
field as $\mutable$ and recovering the $\capsule$ property for the outer object.}

\paragraph{Capabilities}
 {In other proposals \cite{HallerOdersky10,CastegrenWrigstad16} types are compositions of one or more \emph{capabilities}. The modes of the capabilities in a type control how resources of that
type can be aliased. The compositional aspect of capabilities is an important difference
from type qualifiers, as accessing different parts of an object through different capabilities in the same type gives different properties. 
By using capabilities it is possible to obtain an expressivity which looks similar to our type system, even though with different sharing notions and syntactic constructs. For instance, the \emph{full encapsulation} notion in \cite{HallerOdersky10}\footnote{{This paper includes a very good survey of work in this area, notably explaining the difference between \emph{external uniqueness}~\cite{ClarkeWrigstad03} and \emph{full encapsulation}.}}, apart from the fact that sharing of immutable objects is not allowed, is equivalent to the guarantee of our $\capsule$ qualifier, while
our $\lent$ and their \Q|@transient| achieve similar results in different ways.}
Their model has a higher syntactic/logic overhead to explicitly  track regions.
As for all work before~\cite{GordonEtAl12}, objects need to be born \Q|@unique| and the type system 
permits to manipulate data preserving their uniqueness. With recovery~\cite{GordonEtAl12},
instead, we can forget about uniqueness, use normal code designed to work on conventional shared data, and then
recover the aliasing encapsulation property.

\paragraph{Destructive reads} Uniqueness can be enforced by destructive reads, i.e., assigning a copy of 
the unique reference to a variable an destroying the original reference, see
\cite{GordonEtAl12,Boyland10}. Traditionally, borrowing/fractional permissions~\cite{NadenEtAl12} are related to uniqueness  in the opposite way: a unique reference can be borrowed,
it is possible to track when all borrowed aliases are buried~\cite{Boyland01}, and then uniqueness can be recovered.
These techniques offers a sophisticate alternative to destructive reads. 
We also wish to avoid destructive reads. In our work, we ensure uniqueness by linearity, that is, by allowing at most
one use of a $\capsule$ reference.

In our opinion, programming with destructive reads is involved and hurts the correctness of the program, since it leads to the style of programming outlined below, where \Q@a.f@ is a unique/isolated field with destructive read.
\begin{lstlisting}
a.f=c.doStuff(a.f)//style suggested by destructive reads
\end{lstlisting}
The object referenced by \lstinline{a}{} has an \emph{unique/isolated} field \lstinline{f} containing an object \lstinline{b}.
This object \lstinline{b}{} is passed to a client \lstinline{c}{}, which can use (potentially modifying) it. A typical pattern is that the result of such computation is a reference to \lstinline{b}{}, which \lstinline{a}{} can then recover. This approach allows \emph{isolated} fields, as shown above, but has  a serious drawback:
an \emph{isolated} field can become unexpectedly not available (in the example, during execution of \lstinline{doStuff}{}), hence any object contract
involving such field can be broken.
{This can cause {subtle} bugs if \Q@a@ is in the reachable object graph of \Q@c@.}

In our approach, the  $\capsule$ qualifier cannot be applied to fields. Indeed, the ``only once'' use of capsule variables 
makes no sense on fields.
{However, we support the same level of control of the reachable object graph by passing mutable objects to clients as $\lent$, in order to control aliasing behaviour.
That is, the previous code can be rewritten} as follows:
\begin{lstlisting}
c.doStuff(a.f())//our suggested style
\end{lstlisting}
{where \Q@a.f()@ is a getter returning the field as $\lent$.
Note how, during the execution of \Q@doStuff@, \Q@a.f()@ is still available, and,} after the execution of \Q@doStuff@, the aliasing relation {for this field is the same as it was
before \Q@doStuff@ was called.}

\paragraph{Ownership} A {closely related} stream of research is that on \emph{ownership} (see an overview in~\cite{ClarkeEtAl13}) which, however, offers an {opposite} approach. In the ownership approach, it is provided a way to express and prove the ownership invariant\footnote{{Ownership invariant (owner-as-dominator):
Object $o_1$ is owned by object  $o_2$ if in the object graph $o_2$
is a dominator node for $o_1$;
that is, all paths from the roots of the graph (the stack variables)
to $o_1$ pass throw $o_2$.
Ownership invariant (owner-as-modifier):
Object $o_1$ is owned by object  $o_2$ if any field update over $o_1$
is initiated by $o_2$, that is, a call of a method of $o_2$ is present
in the stack trace.}}, which, however, is expected to be guaranteed by defensive cloning, as explained below. In our approach, instead, the capsule concept models an efficient \emph{ownership transfer}. In other words, when an object $\x$ is ``owned'' by another object $\y$, it remains always true that $\y$ can be only accessed only through $\x$, whereas the capsule notion is more dynamic: a capsule can be ``opened'', that is, assigned to a standard reference and modified, and then we can recover the original capsule guarantee. 

For example, assuming a graph with a list of nodes, and a constructor taking in input such list,
the following code establishes the ownership invariant using $\capsule$, and ensures that it cannot be violated using $\lent$.
\begin{lstlisting}
class Graph{
  private final NodeList nodes;
  private Graph(NodeList nodes){this.nodes=nodes; }

  static Graph factory(capsule NodeList nodes){
    return new Graph(nodes);
    }
  
  lent NodeList borrowNodes(mut){return nodes;}
}
\end{lstlisting}
Requiring the parameter of the \lstinline{factory}{} method to be a $\capsule$ guarantees that the list of nodes provided as argument is not referred from the external environment. 
The factory \emph{moves} an isolated portion of store as local store of the newly created object. 
Cloning, if needed, becomes responsibility of the client which provides the list of nodes to the graph. The getter tailors the exposure level of the private store. 

Without aliasing control ($\capsule$ qualifier),  in order to ensure ownership of its list of nodes, the {factory method} should clone the argument, since it comes from an external client environment.
This solution, called  \label{cloning} \emph{defensive cloning}~\cite{Bloch08}, is very popular in the Java community, but inefficient,
since it requires to duplicate the reachable object
graph of the parameter, until immutable nodes are
reached.\footnote{{In most languages, for owner-as-modifier defensive cloning is needed
only when new data is saved inside of an object, while for owner-as-dominator it is needed also when internal data are exposed.}}
Indeed, many programmers prefer to write {unsafe}
 code instead of using defensive cloning for efficiency reasons.
However, this unsafe approach is only possible when programmers have control of the client code, that is, they are not 
working in a library setting.
Indeed many important Java libraries (including the standard Java libraries) today
use defensive cloning to ensure ownership of their internal state.

As mentioned above, our approach is the opposite of the one offered by many ownership approaches, which provide a formal way to express  and prove the ownership invariant that, however, are expected to be guaranteed by defensive cloning. 
We, instead, model an efficient \emph{ownership transfer} through the capsule concept, then, 
duplication of memory, if needed, is performed on the client side\footnote{
Other work in literature supports ownership transfer, see for example~\cite{MullerRudich07, ClarkeWrigstad03}.
In literature it is however applied to uniquess/external uniqueness, thus not {the whole} reachable object graph is transfered.
}.

Moreover, depending on how we expose the owned data, we can closely model
both \emph{owners-as-dominators} (by providing no getter)
and \emph{owners-as-qualifiers} (by providing a \Q@read@ getter). In the example, the method \lstinline{borrowNodes}{} is an example of a $\lent$ getter, a third variant besides the two described on page \pageref{exposer}.  This variant is particularly useful in the case of a field which is owned, indeed, \Q@Graph@ instances can release the mutation control of their nodes without permanently {losing} the aliasing control.

In our approach all properties are deep. On the opposite side, most ownership approaches allows one to distinguish
subparts of the reachable object graph that are referred but not logically owned. This viewpoint has many advantages, for example the Rust language\footnote{\texttt{rust-lang.org}} leverages on ownership to control object deallocation without a garbage collector.
Rust employs a form of uniqueness that can be seen as a restricted ``owners-as-dominators" discipline.  
Rust lifetime parameters behave like additional ownership parameters~\cite{OstlundEtAl08}.

However, in most ownership based approaches 
it is not trivial to encode the concept of full encapsulation, while supporting (open) sub-typing and avoiding defensive cloning.
This depends on how any specific ownership approach entangles subtyping with 
 gaining extra ownership parameters
and extra references to global ownership domains.

\paragraph{Readable notion} Our $\readable$ qualifier is different from \emph{readonly} as used, e.g., in \cite{BirkaErnst04}.
 An object cannot be modified through a readable/readonly reference. However, 
$\readable$ also prevents aliasing.
As discussed in \cite{Boyland06}, readonly semantics can be easily misunderstood by
programmers. Indeed, some wrongly believe it means immutable, whereas the object denoted by a readonly reference can be modified through other references, while others do not realize that readonly data can still be saved in fields, and thus used as a secondary window to observe the change in the object state.
Our proposal addresses both problems, since we explicitly support the $\imm$ qualifier, hence it is more difficult for programmers to confuse the two concepts, and our $\readable$ (readonly + lent) data  cannot be saved in client's fields.

Javari~\cite{TschantzErnst05} also supports the \emph{readonly} type qualifier, and makes a huge design effort to support \emph{assignable} and \emph{mutable} fields, to have fine-grained readonly constraints.  The need of such flexibility is motivated by performance reasons.
In our design philosophy, we do not offer any way of breaking the invariants enforced by the type system. Since our invariants are very strong, we expect compilers to be able to perform optimization, thus recovering most of the efficiency lost to properly use immutable and readable objects.

%% file: Conclu.tex
\section{Conclusion}\label{sect:conclu}
The key contributions of the paper are:
\begin{itemize}
\item A powerful type system for control of mutation and aliasing in class-based imperative languages, providing: type qualifiers for restricting the use of references;  rules for {recovering} a less restrictive type at the price of {constraining} the use of other references; rules for temporarily {unconstrain such} references for typing subexpressions. 
\item A non standard operational model of the language, relying on the language's ability to represent {store}.
\end{itemize}

{We have extensively illustrated the former feature in \refToSection{examples}, here we briefly discuss the latter.}

{In our operational model, aliasing properties are directly expressed at the syntactic level. Notably, in a subterm $\e$ of a program, objects reachable from other parts of the program are simply those denoted by free variables in $\e$, whereas local objects are those denoted by local variables declared in $\e$. In other terms, \emph{the portion of memory only reachable from $\e$ is encoded in $\e$ itself}.  For instance, in}
\begin{lstlisting}
mut D y=new D(0); 
capsule C z = {mut D x = new D(1); new C(x,x)}
\end{lstlisting}
{it is immediately clear that the reference $\z$ denotes an isolated portion of memory, since its right-value is a closed expression.
In the conventional model, instead, memory is flat. For instance, the previous example would be modeled having three locations, say $\iota_y$, $\iota_z$, and $\iota_x$, all at the top-level, hence the fact that $\iota_x$ is not reachable from $\iota_y$ should be reconstructed inspecting the store.}

{In our opinion, this offers a very intuitive and simple understanding of when a subterm of a program can be safely typed $\capsule$, as we have exploited in the examples in \refToSection{typesystem}. More generally, we argue that our model is more adequate for didactic purposes since it does not rely on run-time structures that do not exist in the source program, see \cite{ServettoLindsay13} for an extended discussion on this. Another advantage is that, being the store encoded in terms, proofs can be done by structural induction, whereas in the traditional model invariants on the memory should be proved showing that \emph{all} locations satisfy a given property. This is especially important when using a proof assistant, such as Coq, which we plan as further work. }

{On the opposite side, a disadvantage of our model is the fact that it is possibly more esoteric for people used to the other one. Moreover, since isolation is encoded by scoping, some care is needed to avoid scope extrusion during reduction. Notably, a field access on a reference which points to an isolated portion of store (a block) cannot produce just a reference.}

{As mentioned in \refToSection{typesystem}, the type system presented in this paper, as those in previous work on recovery \cite{GordonEtAl12,ClebschEtAl15}, uses rules which are not syntax-directed, hence \emph{transparent} to the programmer, in the sense that they are applied when needed.  The programmer can simply rely on the fact that expressions are $\capsule$ or $\imm$, respectively, in situations where this is intuitively expected. Of course, the negative counterpart is that the type system is not algorithmic. This is due, as said above, to the presence of non structural rules, which are, besides standard subsumption, recovery rules \rn{t-capsule} and \rn{t-imm}, rules \rn{t-swap} and \rn{t-unrst} for swapping and unrestricting, and of the rule \rn{t-block} for blocks, where variables declared $\lent$ are assigned to groups in an arbitrary way.}

{We considered two different ways to provide an algorithmic version of the type system.}

{The first way is to modify the type system as it is usually done for the subsumption rule, so that non structural rules are applied, roughly, only when needed. We did not formally develop this solution, but the technique is applied in the prototype of L42, a novel programming language designed to support massive use of libraries\footnote{Description and prototype for the full language (in progress)  can be found at \texttt{L42.is}. For testing purposes a small step reduction closely resembling the one presented in this paper is also implemented.}. The current L42 prototype  is important as proof-of-evidence that the type system presented in this paper not only can be implemented, but also smoothly integrated with other features of a realistic language. 
For testing purposes a small step reduction closely resembling the one presented in this paper is also implemented.}

{The other way is a different design of the type system based on type inference. The basic idea is to compute the sharing relations possibly introduced by the evaluation of an expression. A preliminary version of this approach is in \cite{GianniniEtAl17}.}

A  first, informal, presentation of the type modifiers in this paper has been given in \cite{ServettoEtAl13a}.
In \cite{ServettoZucca15} we have provided a formal type system including the $\capsule$ and $\lent$ modifiers, and a preliminary version of 
capsule {recovery}. In \cite{GianniniEtAl16} we have added  the immutable and readable modifiers, and the {immutability recovery}. This paper largely extends \cite{GianniniEtAl16}. The novel contributions include reduction rules, more examples, theorems about the behaviour of modifiers, and proofs of results. 
Imperative calculi where the block construct models store have been introduced  in \cite{ServettoLindsay13,CapriccioliEtAl15}, and used in \cite{ServettoZucca15,GianniniEtAl16}.

{Concerning further work, it would be interesting to investigate how to extend mainstream languages, such as, e.g., Java, with our type qualifiers, and to develop implementations of the calculus presented in this paper, possibly using, e.g., Coq, to also be able to formalize and prove properties.}

{On the more theoretical side, the first direction we plan to explore is the relation between the non-algorithmic type system presented in this paper and the previously mentioned type system based on inference of sharing relations \cite{GianniniEtAl17}}. {We also plan to formally state and prove behavioural equivalence of the calculus with the conventional imperative model.} 

{The fact that our type system tracks requirements on the type context makes it a form of \emph{coeffect type system} in the sense of \cite{PetricekEtAl14}. We plan to investigate better the relation.}
We believe that the novel operational model presented in this paper has the potential of  achieving a better understanding of aliasing.   

As a long term goal, we also plan to investigate (a form of) Hoare logic on top of our model. We believe that the hierarchical structure of our memory representation should help local reasoning, allowing specifications and proofs to mention only the relevant portion of the memory, analogously to what is achieved by separation logic \cite{Reynolds02}.
Finally, it should be possible to use our approach to enforce safe parallelism, on the lines of~\cite{GordonEtAl12,ServettoEtAl13a}.

{\paragraph{Acknowledgement} 
{We are indebted to the anonymous TCS referees for the thorough job they did reviewing our paper and for their valuable suggestions which improved it greatly.}  We also thank the ICTCS'16 referees for their helpful comments on the preliminary version of the paper.

%% file: Appendix.tex
 \appendix
\section{Auxiliary definitions}
\begin{small}
\noindent $\HB{\ctx}$:
\begin{quote}
$
\begin{array}{l}
\HB{\emptyctx}=\emptyset\\
{\HB{\Block{\dvs\,\Dec{\T}{\x}{\ctx}\,\decs}{\e}}=\HB{\ctx}\cup\dom{\dvs}\cup\dom{\decs}}\cup\{x\}\\
\end{array}
$
\end{quote}
\noindent
$\FV{\e}$:
\begin{quote}
$
\begin{array}{l}
\FV{\x}=\{\x\}\\
\FV{\FieldAccess{\e}{\f}}=\FV{\e}\\
\FV{\MethCall{\e_0}{\m}{\e_1,\ldots,\e_n}}=\FV{\e_0}\cup\ldots\cup\FV{\e_n}\\
\FV{\FieldAssign{\e}{\f}{\e'}}=\FV{\e}\cup\FV{\e'}\\
\FV{\ConstrCall{\C}{\e_1,\ldots,\e_n}}=\FV{\e_1}\cup\ldots\cup\FV{\e_n}\\
\FV{\Block\decs{\e}}=\FV{\decs}\cup(\FV{\e}\setminus\dom\decs)\\
\FV{\Dec{\T_1}{\x_1}{\e_1}\ldots\Dec{\T_n}{\x_n}{\e_n}}=
(\FV{\e_1}\cup\ldots\cup\FV{\e_n})\setminus\{\x_1,\ldots,\x_n\}
\end{array}
$
\end{quote}

\noindent
$\Subst{\e}{\val}{\x}$:
\begin{quote}
$
\begin{array}{l}
\Subst{\x}{\val}{\x}=\e\\
\Subst{\z}{\val}{\x}=\z\ \mbox{if}\ \z\neq\x\\
\Subst{\FieldAccess{\e}{\f}}{\val}{\x}=
\FieldAccess{\Subst{\e}{\val}{\x}}{\f}
\\
\Subst{\MethCall{\e_0}{\m}{\e_1,\ldots,\e_n}}{\val}{\x}=
\MethCall{\Subst{\e_0}{\val}{\x}}{\m}{\Subst{\e_1}{\val}{\x},\ldots,\Subst{\e_n}{\val}{\x}}
\\
\Subst{(\FieldAssign{\e}{\f}{\e'})}{\val}{\x}=
\FieldAssign{\Subst{\e}{\val}{\x}}{\f}{\Subst{\e'}{\val}{\x}}
\\
\Subst{\ConstrCall{\C}{\e_1,\ldots,\e_n}}{\val}{\x}=
\ConstrCall{\C}{\Subst{\e_1}{\val}{\x},\ldots,\Subst{\e_n}{\val}{\x}}\\
\Subst{\Block{\decs}{\e}}{\val}{\x}=
\Block{\Subst{\decs}{\val}{\x}}{\Subst\e{\val}{\x}}
\mbox{ if }\x\notin\dom\decs
\\
\Subst{\Block\decs{\e}}{\val}{\x}=
\Block\decs{\e}
\mbox{ if }\x\in\dom\decs\\
{\Subst{(\Dec{{\T}_1}{\x_1}{\e_1}\ldots\Dec{{\T}_n}{\x_n}{\e_n})}{\val}{\x}=
\Dec{{\T}_1}{\x_1}{\Subst{\e_1}{\val}{\x}}\ldots\Dec{{\T}_n}{\x_n}{\Subst{\e_n}{\val}{\x}}}
\end{array}
$
\end{quote}
\end{small}

\section{Proofs omitted from \refToSection{results}}\label{app:proofs}
%In this section we give {the} proofs of the lemmas of \refToSection{results}. 
%
%\subsection{Proofs for the Canonical Form Theorem}
\label{sect:proof-cf}

\noindent
{\bf Proof of \refToLemma{nonStructural}}.
%\begin{proof}
By induction on type derivations.\\
If the last rule in $\deriv$ is a structural rule, then   $\deriv'=\deriv$.\\
If the last rule in $\deriv$ is \rn{T-sub}, by induction hypothesis on 
the sub-derivation of the premise of \rn{T-sub} and transitivity of $\leq$ we derive the result.\\
Let $\deriv$ end with an application of \rn{T-Capsule} of \rn{T-Imm}, then
\begin{center}
$
\deriv_p:\AuxTypeCheck{{\Gamma}}{\LentLocked\ (\domMut{\Gamma}{\setminus}\LentLocked)}{\emptyset}{\StronglyLocked''}{\e}{\Type{\mu'}{\C}}
$
\end{center}
where $\StronglyLocked''=\StronglyLocked$ or $\StronglyLocked''=\domGeqMut(\Gamma)$, is the sub-derivation of the premise of the last rule. Note that the mutable group is always empty. By induction hypothesis on 
$\deriv_p$, there is a sub-derivation of $\deriv'$ of $\deriv_p$
such that $\deriv':\AuxTypeCheck{{\Gamma}}{\LentLocked'}{\MutGroup'}{\StronglyLocked'}{\e}{\Type{\mu''}{\C}}$ ends with the application of a structural rule, and 
$\LentLocked\ (\domMut{\Gamma}{\setminus}\LentLocked)=\LentLocked'\ \MutGroup'$. This proves the result. \\
If the last rule applied is \rn{T-Unrst} 
$\AuxTypeCheck{{\Gamma}}{\LentLocked}{\MutGroup}{\emptyset}{\e}{\Type{\mu}{\C}}$, by induction hypothesis on 
the sub-derivation of the premise of the rule, we derive the result. \\
If the last rule applied is \rn{T-Swap},  let $\LentLocked=\LentLocked_1\ \xs_1$ 
\begin{center}
$\deriv_p:\AuxTypeCheck{\Gamma}{\LentLocked_1\ \xs}{\MutGroup_1}{\StronglyLocked}{\e}{\Type{\mu'}{\C}}$
\end{center}
By induction hypothesis on $\deriv_p$ there is a sub-derivation $\deriv'$
such that $\deriv':\AuxTypeCheck{{\Gamma}}{\LentLocked'}{\MutGroup'}{\StronglyLocked'}{\e}{\Type{\mu''}{\C}}$ ends with the application of a structural rule,  
 $\LentLocked_1\ \xs\ \xs_1=\LentLocked'\ \xs'$, and either $\mu'=\capsule$ or $\mu'=\imm$ or $\mu''\leq\mu'$. Since $\LentLocked=\LentLocked_1\ \xs_1$
 we have that $\LentLocked\ \xs=\LentLocked'\ \xs'$, moreover, since $\mu'\leq\mu$ by transitivity of 
 $\leq$ we have the result. \qed
%\end{proof}

\medskip
\noindent
{\bf Proof of \refToLemma{typeVars}}. 
By induction on the derivation $\TypeCheck{\Gamma}{\LentLocked}{\StronglyLocked}{{\x}}{\T}$.
Consider the last rule applied in the derivation.\\
If the last rule is \rn{T-var}, then either $\T=\Gamma(x)$ or $\T=\Type{\lent}{\C}$ if $\Gamma(x)=\Type{\mutable}{\C}$ and $\x\in\LentLocked$. 
Therefore, in both cases $\Gamma(x)\leq\T$.
Otherwise the last rule applied must be a non-structural rule.\\
If the last rule applied is \rn{T-sub}, the result follows by induction hypothesis, and
transitivity of $\leq$.\\
If the last rule applied is \rn{T-swap}, then $\TypeCheck{\Gamma}{\LentLocked'}{\StronglyLocked}{{\x}}{\T'}$, where 
$\T'\leq\T$. Therefore, again by induction induction hypothesis, and
transitivity of $\leq$ we get the result.\\
If the last rule applied is \rn{T-unrst} the result follows by induction hypothesis.\\
If the last rule applied is \rn{T-Capsule}, then
$\TypeCheck{\Gamma}{\LentLocked\ \xs}{\StronglyLocked}{{\x}}{\Type{\mutable}{\C}}$
where $\xs=\domMut{\Gamma}{\setminus}\LentLocked$, and $\T=\Type{\capsule}{\C}$.
By induction hypothesis, $\Gamma(x)\leq\Type{\mutable}{\C}$. So, either 
$\Gamma(x)=\Type{\capsule}{\C}$ or $\Gamma(x)=\Type{\mutable}{\C}$. In the first case
$\Gamma(x)=\T$. The second case is not possible, since $\x\in\LentLocked\ \xs$, and therefore 
$\TypeCheck{\Gamma}{\LentLocked\ \xs}{\StronglyLocked}{{\x}}{\Type{\lent}{\C}}$,
so rule \rn{T-Capsule} cannot be applied. \\
Finally, if the last rule applied is \rn{T-Imm}, it must be that
$\TypeCheck{\Gamma}{\LentLocked\ \xs}{\domGeqMut(\Gamma) }{\x}{\Type{{\readable}}{\C}}$.
So, $\x\not\in\domGeqMut(\Gamma)$, and therefore $\Gamma(x)=\Type{\capsule}{\C}$ or 
$\Gamma(x)=\Type{\imm}{\C}$. In both cases $\Gamma(x)\leq\T=\Type{\imm}{\C}$.\\
Let $\TypeCheck{\Gamma}{\LentLocked}{\StronglyLocked}{{\x}}{\Type{\mu}{\C}}$ with
$\mu\leq\imm$. We have that $\Gamma(x)\leq\Type{\mu}{\C}$. Therefore, if
$\LentLocked'$ and $\StronglyLocked'$ are such that  $\WellFormedTypeCtx{\Gamma};{\LentLocked'};\StronglyLocked'$,
then $\x\not\in\LentLocked'$ and $\x\not\in\StronglyLocked'$. So 
$\TypeCheck{\Gamma}{\LentLocked'}{\StronglyLocked'}{{\x}}{\Type{\mu}{\C}}$.\qed

\medskip
\noindent
{\bf Proof of \refToLemma{inversionBlock}}. 
Let $\deriv:\AuxTypeCheck{{\Gamma}}{\LentLocked}{\MutGroup}{\StronglyLocked}{\Block{\decs}{\val}}{\Type{\mu}{\C}}$.
From \refToLemma{nonStructural} there is a sub-derivation $\deriv':\AuxTypeCheck{{\Gamma}}{\LentLocked''}{\MutGroup''}{\StronglyLocked''}{\Block{\decs}{\val}}{\Type{\mu'}{\C}}$ of $\deriv$ which ends with an application of rule \rn{T-Block}
such that  
\begin{enumerate}[(a)]
\item $\LentLocked\ \xs=\LentLocked''\ \xs''$ and
\item $\mu\not=\imm$ and $\mu\not=\capsule$ implies $\mu'\leq\mu$.
\end{enumerate}
 From rule \rn{T-block} and \refToLemma{typeStruct}.1 we have that
\begin{enumerate}[(a)]\addtocounter{enumi}{2}
\item $\DecsOK{\SubstFun{\Gamma}{\TypeEnv{\decs}}}{\LentLocked'}{\StronglyLocked'}{\decs}$ and
 $\TypeCheck{{\SubstFun{\Gamma}{\TypeEnv{\decs}}}}{{\LentLocked'}}{\StronglyLocked'}{\val}{\Type{\mu'}{\C}}$
where
  \item $\LessEq{(\ \LentLocked''\ \xs''\ ){\setminus}\dom{\TypeEnv{\decs}}}{\LentLocked'\ \xs'}$.
\end{enumerate}
From (a) and (d) we derive that $\LessEq{(\ \LentLocked\ \xs\ ){\setminus}\dom{\TypeEnv{\decs}}}{\LentLocked'\ \xs'}$ which, with (b) and (c), proves the result.\qed

\medskip
\noindent
{\bf Proof of \refToLemma{constrMut}}.
\begin{enumerate}
\item
From \refToLemma{typeStruct} we have that $\AuxTypeCheck{\Gamma}{\LentLocked}{\MutGroup}{\StronglyLocked}{\ConstrCall{\C}{\z_1,\ldots,\z_n}}{\Type{\mutable}{\C}}$ is derived applying rule \rn{T-New}. Therefore  $\AuxTypeCheck{\Gamma}{\LentLocked}{\MutGroup}{\StronglyLocked}{\z_i}{\Type{\mu_i}{\C_i}}$ $\Range{i}{1}{n}$.
If $\T_i=\mutable$, then $\z_i\not\in\LentLocked$ and therefore $\z_i\in\xs$. If $\T_i=\imm$, from 
\refToLemma{typeVars}, we have $\AuxTypeCheck{\Gamma}{\LentLocked}{\MutGroup}{\StronglyLocked}{\z_i}{\Type{\mu}{\C}}$ with $\mu\leq\imm$.
\item 
Let  $\deriv:\AuxTypeCheck{\Gamma}{\LentLocked}{\emptyset}{\domGeqMut(\Gamma) }{\ConstrCall{\C}{\z_1,\ldots,\z_n}}{\Type{\readable}{\C}}$. There is a sub-derivation $\deriv'$ of $\deriv$ ending with
an application of \rn{T-New}, which can be either the premise of rule \rn{T-Sub} or of rule \rn{T-Swap}.\\
In the first case, $\AuxTypeCheck{\Gamma}{\LentLocked}{\emptyset}{\domGeqMut(\Gamma) }{{\ConstrCall{\C}{\xs}}}{\Type{\mutable}{\C}}$. 
Therefore, from clause 1. of this lemma, since the current mutable group is empty, we have that
for all $i\in1..n$, $\TypeCheck{\Gamma}{\LentLocked}{\StronglyLocked}{\z_i}{\Type{\imm}{\C_i}}$.\\
In the second case,
$\AuxTypeCheck{\Gamma}{\LentLocked'}{\MutGroup'}{\domGeqMut(\Gamma) }{{\ConstrCall{\C}{\z_1,\ldots,\z_n}}}{\Type{\mutable}{\C}}$ for some 
$\LentLocked'\ \xs'=\LentLocked$.  Assume that, for some $i$, $1\leq i\leq n$,  $\T_i=\Type{\mu_i}{\C_i}$ with
$\mu_i\geq\mutable$, and therefore, $\z_i\in\domGeqMut(\Gamma)$. 
Rule \rn{T-Var} is not be applicable to derive a type for $\z_i$. (Rule \rn{T-UnRst} is not applicable to unlock
variables, since the variable has type $\Type{\mu_i}{\C_i}$ with $\mu_i\geq\mutable$.)
Therefore, for all $i\in1..n$, we have that
$\Gamma(\z_i)=\Type{\mu}{\C_i}$ with $\mu\leq\imm$. 
\end{enumerate}
\qed

\medskip
\noindent
{\bf Proof of \refToLemma{blockMut}}. Let  $\AuxTypeCheck{\Gamma}{\LentLocked}{\MutGroup}{\StronglyLocked}{\Block{\dvs}{\cOrx}}{\Type{\mutable}{\C}}$ where $\cOrx=\x$ or $\cOrx=\ConstrCall{\C}{\_}$. From \refToLemma{inversionBlock}
\begin{enumerate}[(a)]
\item  $\AuxTypeCheck{{\SubstFun{\Gamma}{\TypeEnv{\dvs}}}}{{\LentLocked'}}{\MutGroup'}{\StronglyLocked'}{\cOrx}{\Type{\mu'}{\C}}$ 
\item $\AuxDecsOK{\SubstFun{\Gamma}{\TypeEnv{\dvs}}}{\LentLocked'}{\MutGroup'}{\StronglyLocked'}{\Dec{\Type{\mu_i}{\C_i}}{\z_i}{\stVal_i}}$  $\Range{i}{1}{n}$
\item $\LessEq{(\LentLocked\ \xs){\setminus}\dom{\TypeEnv{\dvs}}}{\LentLocked'\ \xs'}$ and
\item $\mu'\leq\mutable$.
\end{enumerate}
We may assume that $\mu'=\mutable$, since $\mu'=\capsule$ would imply that (a) was obtained by a \rn{T-Capsule} rule and so $\TypeCheck{{\SubstFun{\Gamma}{\TypeEnv{\dvs}}}}{{\LentLocked'\ xs'}}{\StronglyLocked'}{\cOrx}{\Type{\mutable}{\C}}$.\\
From the definition of well-formed right-values of 
\refToFigure{wellformed} we have that $\dvs=\Reduct{\dvs}{\FV{\cOrx}}$, that is $\Range{j}{1}{n}$ we have that $\connected{\dvs}{\FV{\cOrx}}{\z_j}$.\\
We want to prove that $\Range{j}{1}{n}$, if $\connected{\dvs}{\FV{\cOrx}}{\z_j}$ and 
$\SubstFun{\Gamma}{\TypeEnv{\dvs}}(z_j)=\Type{\mutable}{\C_j}$, then 
\begin{itemize}
\item $\z_j\in\xs'$, and
\item if $\z\in\FV{\stVal_j}$ and $\Gamma[\TypeEnv{\dvs}](z)=\Type{\mutable}{\D}$, then $\z\in\xs'$, otherwise 
 $\Gamma[\TypeEnv{\dvs}](z)=\Type{\mu}{\D}$ with $\mu\leq\imm$.
\end{itemize}
By definition of connected, either $\z_j\in\FV{\cOrx}$ or 
 $\z_j\in\FV{\stVal_i}$ and $\connected{\dvs}{\FV{\cOrx}}{\z_i}$.\\
If \underline{$\z_j\in\FV{\cOrx}$ and $\cOrx=\ConstrCall{\C}{\_}$} by (a), (d) and \refToLemma{constrMut}.1 we have that either $\z_j\in\xs'$
or $\Gamma[\TypeEnv{\dvs}](z_j)=\Type{\mu}{\D}$ with $\mu\leq\imm$. \\
If \underline{$\z_j\in\FV{\cOrx}$ and $\cOrx=\z_j$} then, from (a) and (d) we get $\z_j\in\xs'$ ($\z_j\in\LentLocked'$ would
imply $\AuxTypeCheck{{\SubstFun{\Gamma}{\TypeEnv{\dvs}}}}{{\LentLocked'}}{\MutGroup'}{\StronglyLocked'}{\cOrx}{\Type{\lent}{\C}}$ ). \\
If \underline{$\z_j\in\FV{\stVal_i}$ and $\connected{\dvs}{\FV{\cOrx}}{\z_i}$}, by inductive hypothesis
on the connection relation, we have that 
\begin{enumerate}[(1)]
\item $\z_i\in\xs'$, and
\item if $\z\in\FV{\stVal_i}$ and $\Gamma[\TypeEnv{\dvs}](z)=\Type{\mutable}{\D}$, then $\z\in\xs'$, otherwise 
 $\Gamma[\TypeEnv{\dvs}](z)=\Type{\mu}{\D}$ with $\mu\leq\imm$.
\end{enumerate}
From (b) and Definition \ref{def:wellTypedDefsNew} we derive that
$\AuxTypeCheck{\Gamma[\TypeEnv{\dvs}]}{\LentLocked}{\xs'}{\StronglyLocked}{\stVal_i}{\Type{\mutable}{\C_i}}$.
If  $\z_j$ is such that $\SubstFun{\Gamma}{\TypeEnv{\dvs}}(z_j)=\Type{\mutable}{\C_j}$, from (2), $\z_j\in\xs'$.
Again, from (b) and Definition \ref{def:wellTypedDefsNew} we derive that
$\AuxTypeCheck{\Gamma[\TypeEnv{\dvs}]}{\LentLocked}{\xs'}{\StronglyLocked}{\stVal_j}{\Type{\mutable}{\C_j}}$.
Let $\z\in\FV{\stVal_j}$. If $\stVal_j=\ConstrCall{\D}{\_}$
then from \refToLemma{constrMut}.1 we derive the result. 
If $\stVal_j=\Block{\dvs'}{\cOrx'}$, then by inductive
hypothesis on $\Block{\dvs'}{\cOrx'}$ we have that either $\z\in\xs'$ or $\Gamma[\TypeEnv{\dvs}](z)=\Type{\mu}{\D}$ with $\mu\leq\imm$.
\qed

%\subsection{Proofs for the soundness theorems}
\label{sect:proof-sound}

\noindent
{\bf Proof of \refToLemma{subCtx}}. By induction on $\genCtx$. 
%\begin{enumerate}
For $\genCtx=\emptyctx$ clause 1 derives from \refToLemma{nonStructural}, and clause 2 is obvious. \\
%\item
If $\genCtx=\Block{\decs'\ \Dec{\T}{\x}{\genCtxP}\ \decs''}{\val}$, then $\TypeCheck{\Gamma}{\LentLocked}{\StronglyLocked}{\Block{\decs'\ \Dec{\T_x}{\x}{\GenCtxP{\e}}\ \decs''}{\val}}{\T}$. Let $\decs$ be $\decs'\ \Dec{\T_x}{\x}{\GenCtxP{\e}}\ \decs''$.  From \refToLemma{inversionBlock}, we have that $\DecsOK{\SubstFun{\Gamma}{\TypeEnv{\decs}}}{\LentLocked'} {\StronglyLocked'}{\decs}$ where
\begin{enumerate}[(a)]
  \item $\LessEq{(\ \LentLocked\ \xs\ ){\setminus}\dom{\TypeEnv{\decs}}}{\LentLocked'\ \xs'}$.
\end{enumerate}
From $\DecsOK{\SubstFun{\Gamma}{\TypeEnv{\decs}}}{\LentLocked'}{\StronglyLocked'}{\decs}$ we get 
$\DecsOK{\SubstFun{\Gamma}{\TypeEnv{\decs}}}{\LentLocked'}{\StronglyLocked'}{\Dec{\Dec{\Type{\mu_x}{\C_x}}{\x}{\GenCtxP{\e}}}{\x}{\GenCtxP{\e}}}$.\\
We have two cases: either $\mu_x\not=\lent$ and
\begin{itemize}
\item $\deriv_x:\AuxTypeCheck{\SubstFun{\Gamma}{\TypeEnv{\decs}}}{\LentLocked'}{\MutGroup'}{\StronglyLocked'}{{\GenCtxP{\e}}}{\SubstFun{\Gamma}{\TypeEnv{\decs}}(x)}$
\end{itemize}
 or $\mu_x\not=\lent$ and
\begin{itemize}
\item $\deriv'_x:\AuxTypeCheck{\SubstFun{\Gamma}{\TypeEnv{\decs}}}{\LentLocked''}{\MutGroup''}{\StronglyLocked'}{{\GenCtxP{\e}}}{\SubstFun{\Gamma}{\TypeEnv{\decs}}(x)}$ with $\x\in\xs''$ and 
\item $\LentLocked'\ \xs'=\LentLocked''\ \xs''$.
\end{itemize}
In both cases, applying the inductive hypothesis to $\genCtxP$ we derive that 
$\deriv':\TypeCheck{\SubstFun{\SubstFun{\Gamma}{\TypeEnv{\decs}}}{\TypeEnv{\ctx}}}{\LentLocked'''}{\StronglyLocked'''}{{\e}}{\T''}$ for some
$\LentLocked'''$, $\StronglyLocked'''$ and $\T''$ such that 
\begin{enumerate} [(A)]
  \item $\LessEq{(\LentLocked'\ \xs'=\LentLocked''\ \xs''){\setminus}\dom{\TypeEnv{\genCtxP}}}{\LentLocked'''\ \xs'''}$, and
  \item the last rule applied in $\deriv'$ is a structural rule.
\end{enumerate}
Since $\TypeEnv{\genCtx}=\SubstFun{\TypeEnv{\decs}}{\TypeEnv{\genCtxP}}$, from (a), (A) and transitivity of $\Extends{}{}{}$ we derive 
\begin{center}
$\LessEq{(\LentLocked\ \xs){\setminus}\dom{\TypeEnv{\genCtx}}}{\LentLocked'''\ \xs'''}$.
\end{center}
This proves clause 1.\\
Let $\e'$ is such that 
$\deriv'':\TypeCheck{\SubstFun{\SubstFun{\Gamma}{\TypeEnv{\decs}}}{\TypeEnv{\ctx}}}{\LentLocked''}{\StronglyLocked''}{{\e'}}{\T''}$. By induction hypothesis on $\genCtxP$, substituting $\deriv'$ with $\deriv''$ in $\deriv_x$ (or $\deriv'_x$) we obtain $\DecsOK{\SubstFun{\Gamma}{\TypeEnv{\decs}}}{\LentLocked'}{\StronglyLocked'}{\Dec{\Dec{\Type{\mu_x}{\C_x}}{\x}{\GenCtxP{\e}}}{\x}{\GenCtxP{\e}}}$. Applying rule \rn{T-block} we get  a derivation for
$\TypeCheck{\Gamma}{\LentLocked}{\StronglyLocked}{\Block{\decs'\ \Dec{\T_x}{\x}{\GenCtxP{\e'}}\ \decs''}{\val}}{\T}$,
which proves clause 2.\\
%\end{enumerate}
The proofs for the other general contexts similar and easier.
%\PGComm{Maybe a hint to the other cases, in particular when $\dom{\TypeEnv{\genCtx}}=\emptyset$!}
 \qed

\medskip
\noindent
{\bf Proof of \refToLemma{fieldAccess}}.  
\begin{enumerate}
  \item From \refToLemma{subCtx}.1, taking $\genCtx=\FieldAccess{\emptyctx}{\f_i}$, we have that 
  \begin{enumerate} [(a)]
\item  $\AuxTypeCheck{\Gamma}{\LentLocked'}{\MutGroup'}{\StronglyLocked'}{{\ConstrCall{\C}{\xs}}}{\Type{\mutable}{\C}}$ (ending with with an application of \rn{T-New}) and
\item 
 $\LentLocked\ \xs=\LentLocked'\ \xs'$.
\end{enumerate}
From (a) we get $\TypeCheck{{\TypeEnv{\ctx}}}{\LentLocked'}{\StronglyLocked'}{\x_i}{\Type{\mu_i}{\C_i}}$.\\
If $\mu_i=\imm$, then from \refToLemma{typeVars} we have that $\TypeCheck{\Gamma}{\LentLocked}{\StronglyLocked}{\x_i}{\Type{\mu_i}{\C_i}}$. From rule \rn{T-Field-Access} derive that $\mu=\mu_i$. Therefore, from (b) we get $\TypeCheck{\Gamma}{\LentLocked}{\StronglyLocked}{\x_i}{\Type{\mu}{\C_i}}$.\\
Let $\mu_i=\mutable$. If $\x_i\in\xs$, from (a) we derive  $\x_i\in\xs'$. Therefore, from (b)
$\LentLocked=\LentLocked'$. Since $\mu\geq\mutable$ (with an application of
\rn{T-Sub} if needed) we derive  $\TypeCheck{{\TypeEnv{\ctx}}}{\LentLocked}{\StronglyLocked}{\x_i}{\Type{\mu}{\C_i}}$.
If $\x_i\in\LentLocked$, then before rule \rn{T-New} we have to apply \rn{T-Swap}  (to exchange the group containing $\x_i$ with $\xs$). Since rule \rn{T-New} produces a type with $\mutable$ modifier, applying rule \rn{T-Swap} 
we get a $\lent$ modifier. Therefore from (b) also $\mu=\lent$ and
$\TypeCheck{\Gamma}{\LentLocked}{\StronglyLocked}{\x_i}{\Type{\mu}{\C_i}}$. 
 
 \item If $\stVal=\ConstrCall{\C}{\xs}$, the result follows from clause 1 of this lemma. \\
 Let $\stVal=\Block{\dvs}{\cOrx}$ where 
$\dvs=\Dec{\T_1}{\z_1}{\stVal_1}\cdots\Dec{\T_n}{\z_n}{\stVal_n}$, and  
$\cOrx=\ConstrCall{\C}{\xs}$ or $\cOrx=x$ and $\T_i=\Type{\mu_i}{\C_i}$.
If $\AuxTypeCheck{\Gamma}{\LentLocked}{\MutGroup}{\StronglyLocked}{\FieldAccess{\stVal}{\f_i}}{\Type{\mu'}{C_i}}$,
from rule \rn{T-Field-Access}, we have that 
  \begin{enumerate} [(a)]
\item  $\AuxTypeCheck{\Gamma}{\LentLocked'}{\MutGroup'}{\StronglyLocked'}{\Block{\dvs}{\cOrx}}{\Type{\mu}{\C}}$, and
\item if $\mu_i\neq\mutable$ then $\mu'=\mu_i$, otherwise $\mu'=\mu$
\end{enumerate}
From \refToLemma{inversionBlock} we have
  \begin{enumerate} [(a)]\addtocounter{enumii}{2}
\item  $\AuxDecsOK{\SubstFun{\Gamma}{\TypeEnv{\decs}}}{\LentLocked'}{\MutGroup'}{\StronglyLocked'}{\dvs}$
\item  $\AuxTypeCheck{{\SubstFun{\Gamma}{\TypeEnv{\dvs}}}}{{\LentLocked'}}{\MutGroup'}{\StronglyLocked'}{\cOrx}
{\Type{\mu''}{\C}}$ 
\item $\LessEq{(\LentLocked\ \xs){\setminus}\dom{\TypeEnv{\decs}}}{\LentLocked'\ \xs'}$ and
\item $\mu\not=\imm$ and $\mu\not=\capsule$ implies $\mu''\leq\mu$. 
\end{enumerate}
If $\cOrx=\ConstrCall{\C}{\xs}$, then  $\fieldOf{\stVal}{i}=\Block{\dvs}{\x_i}$.
From (b), (c) and rule \rn{T-Block}
 \begin{enumerate} [(a)]\addtocounter{enumii}{6}
\item  $\AuxTypeCheck{\Gamma}{\LentLocked'}{\MutGroup'}{\StronglyLocked'}{\Block{\dvs}{\x_i}}{\Type{\mu_i}{\C_i}}$ 
\end{enumerate}
%By cases on $\mu_i$.\\
If \underline{$\mu_i\neq\mutable$}, from (b) the result holds.\\
If \underline{$\mu_i=\mutable$}, and $\mu\leq\imm$, then (a) is obtained by either rule \rn{T-Capsule} or \rn{T-imm}. 
The same recovery rule can be applied to (g). If $\mu\geq\mutable$, result is obtained by (g) applying rule \rn{T-sub},
if needed.\\
If $\cOrx=\x$, then  $\fieldOf{\stVal}{i}=\Block{\dvs}{\val}$ if $\dvs(\x)=\Dec{\_}{\x}{\stVal'}$ and $\fieldOf{\stVal'}{i}=\val$.
In this case the result follow by induction on $\stVal'$.
\item By case analysis on the result of $\fieldOf{\stVal}{i}$, using the proof of the previous case we can see that, typing depends only on $\Gamma$ and therefore the result holds.
\end{enumerate}
\qed

\medskip
\noindent
{\bf Proof of \refToLemma{fieldAssign}}.
From rule  \rn{T-Field-Assign}, we have that
$\AuxTypeCheck{{\Gamma}}{\LentLocked}{\MutGroup}{\StronglyLocked}{\x}{\mutable\,\C}$ for some $\C$
and $\AuxTypeCheck{{\Gamma}}{\LentLocked}{\MutGroup}{\StronglyLocked}{\valPrime}{\Type{\mu}{\_}}$ where $\mu=\mutable$ or $\mu=\imm$. Therefore, from \refToLemma{typeVars} we have that $\Gamma(\x)\leq\Type{\mutable}{\C}$, and 
$\x\in\xs$.
Since $\mu=\mutable$ or $\mu=\imm$, from Canonical Form Theorem.2 and 3 we have that for all  $\y\in\FV{\valPrime}$, $\y\in\xs$ or $\TypeCheck{{\Gamma}}{\LentLocked}{\StronglyLocked}{\y}{\Type{\mu}{\_}}$ with $\mu\leq\imm$.
\qed

\medskip
\noindent
{\bf Proof of  \refToTheorem{subjectReductionShort}} (rules \rn{field-assign}, and \rn{mut-move})\\
Consider \underline{rule \rn{field-assign}}.
Let   $\decs'=\dvs\ \Dec{\Type{\mu_x}{\C}}{\x}{\ConstrCall{\C}{\xs}}\ \Dec{\T_z}{\z}{\ctx_z[{\FieldAssign{\x}{\f}{{\valPrime}}}]}\ \decs$, 
\begin{enumerate} [(1)]
\item $\e=\Ctx{\Block{\dvs\ \Dec{\Type{\mu_x}{\C}}{\x}{\ConstrCall{\C}{\xs}}\ \Dec{\T_z}{\z}{\ctx_z[{\FieldAssign{\x}{\f}{{\valPrime}}}]}\ \decs}{\val}}$, and 
\item $\e'=\Ctx{\Block{\dvs\ \Dec{\Type{\mu_x}{\C}}{\x}{\ConstrCall{\C}{\xs'}}\ \Dec{\T_z}{\z}{\ctx_z[\valPrime]}\ \decs}{\val}}$, 
\end{enumerate}
where $\fields{\C}=\Field{\T_1}{\f_1}\ldots\Field{\T_n}{\f_n}$ with $\f=\f_i$ and
\begin{enumerate}[(1)]\addtocounter{enumi}{2}
\item ${\mu\geq\mutable}$,  
\item ${\noCapture{\x}{\HB{\ctx_z}},\noCapture{\valPrime}{\HB{\ctx_z}}}$ and
\item $\xs'$ is obtained by $\xs$ replacing $\x_i$ with $\valPrime$.
\end{enumerate}
Let $\Gamma'=\TypeEnv{\ctx}[\TypeEnv{\decs'}]$, from (1) and Lemma \ref{lemma:subCtx}.1 for some
$\T'$, $\LentLocked$ and $\StronglyLocked$
\begin{enumerate} [(a)]
\item $\AuxTypeCheck{\Gamma'[\TypeEnv{\ctx_z}]}{\LentLocked}{\MutGroup}{\StronglyLocked}{\FieldAssign{\x}{\f}{{\valPrime}}}{\T'}$ and
  \item the last rule applied in the derivation is \rn{T-Field-assign}.  
 \end{enumerate}
From (1), \refToLemma{subCtx}.1 and \refToLemma{inversionBlock}, for some $\LentLocked'$ and $\StronglyLocked'$
\begin{enumerate} [(a)]\addtocounter{enumi}{2}
\item $\AuxDecsOK{\Gamma'}{\LentLocked'}{\MutGroup'}{\StronglyLocked'}{\Dec{\Type{\mu_x}{\C}}{\x}{\ConstrCall{\C}{\xs}}}$, i.e., 
\begin{enumerate}[i.]
\item $\AuxTypeCheck{\Gamma'}{\LentLocked'}{\MutGroup'}{\StronglyLocked'}{\ConstrCall{\C}{\xs}}{\Gamma'(\x)}$ if $\mu_x\not=\lent$
\item $\AuxTypeCheck{\Gamma'}{\LentLocked_x}{\MutGroup_x}{\StronglyLocked'}{\ConstrCall{\C}{\xs}}{\Gamma'(\x)}$ if $\mu_x=\lent$ and $\LentLocked'\ \xs'=\LentLocked_x\ \xs_x$ where $\x\in\xs_x$
\end{enumerate}
\item 
 $\LessEq{(\LentLocked'\ \xs'){\setminus}\dom{\TypeEnv{\ctx_z}}}{\LentLocked\ \xs}$.
\end{enumerate}
From (a) and rule \rn{T-Field-filed -Assign} we get $\T'=\T_i$
\begin{enumerate} [(a)]\addtocounter{enumi}{4}
\item $\AuxTypeCheck{\Gamma'[\TypeEnv{\ctx_z}]}{\LentLocked}{\MutGroup}{\StronglyLocked}{{{\valPrime}}}{\T'}$
\item $\x\in\xs$, and
\item  from \refToLemma{fieldAssign},  and (4): for all
$\y\in\FV{\valPrime}$ such that $\TypeCheck{\Gamma'[\TypeEnv{\ctx_z}]}{\LentLocked}{\StronglyLocked}{\y}{\Type{\mu}{\_}}$ with $\mu\geq\mutable$ we have that $\y\in\xs$
\end{enumerate}
From (c).i and (c).ii and rule \rn{T-New}, we have that $\AuxTypeCheck{\Gamma'}{\LentLocked'}{\MutGroup'}{\StronglyLocked'}{\x_i}{\T'}$ (if $\mu_x\not=\lent$) or  $\AuxTypeCheck{\Gamma'}{\LentLocked_x}{\MutGroup_x}{\StronglyLocked'}{\x_i}{\T'}$ (otherwise). 
Note that, in both cases $\x$ is in the current mutable group. Therefore, from (d), (g),
(e), (4), and \refToLemma{weakening} we have that
$\AuxTypeCheck{\Gamma'}{\LentLocked'}{\MutGroup'}{\StronglyLocked'}{\valPrime}{\T'}$ (if $\mu_x\not=\lent$) or  $\AuxTypeCheck{\Gamma'}{\LentLocked_x}{\MutGroup_x}{\StronglyLocked'}{\valPrime}{\T'}$ (otherwise).
From rule \rn{T-New}, (5) and \refToLemma{subCtx}.2 
\begin{enumerate} [(a)]\addtocounter{enumi}{7}
\item $\AuxDecsOK{\Gamma'}{\LentLocked'}{\MutGroup'}{\StronglyLocked'}{\Dec{\Type{\mu_x}{\C}}{\x}{\ConstrCall{\C}{\xs'}}}$, 
\end{enumerate}
Let $\decs''=\dvs\ \Dec{\Type{\mu_x}{\C}}{\x}{\ConstrCall{\C}{\xs'}}\ \Dec{\T_z}{\z}{\ctx_z[{\FieldAssign{\x}{\f}{{\valPrime}}}]}\ \decs$, $\ctxP=\ctx[\TypeEnv{\decs'}[\ctx_z]]$, and $\ctx''=\ctx[\TypeEnv{\decs''}[\ctx_z]]$. By definition of type context 
$\TypeEnv{\ctx''}=\TypeEnv{\ctx'}=\Gamma'[\TypeEnv{\ctx_z}]$. Therefore, from (h), (a), and (e) 
\begin{itemize}
\item $\AuxTypeCheck{\TypeEnv{\ctx''}}{\LentLocked}{\MutGroup}{\StronglyLocked}{\FieldAssign{\x}{\f}{{\valPrime}}}{\T'}$, and 
\item $\AuxTypeCheck{\TypeEnv{\ctx''}}{\LentLocked}{\MutGroup}{\StronglyLocked}{{{\valPrime}}}{\T'}$.
\end{itemize}
From \refToLemma{subCtx}.2 we get $\TypeCheckGround{\e'}{\T}$.

\medskip\noindent
Consider \underline{rule \rn{mut-move}}.
In this case  
\begin{enumerate}[(1)]
\item $\e=\Ctx{\Block{\dvs'\ \Dec{\Type{\mu_x}{\C}}{\x}{\Block{\dvs\ \dvs''}{\val}}\ \decs'}{\val'}}$, and 
\item $\e'=\Ctx{\Block{\dvs'\ \dvs\ \Dec{\Type{\mu_x}{\C}}{\x}{\Block{\dvs''}{\val}}\ \decs'}{\val'}}$, 
\end{enumerate}
where 
\begin{enumerate}[(1)]\addtocounter{enumi}{2}
\item $\mu_x \geq \mutable$,  
\item $\noCapture{\Block{\dvs'\ \decs'}{\val'}}{\dom{\dvs}}$ and 
\item $\noCapture{\dvs}{\dom{\dvs''}}$. 
\end{enumerate}
Let $\decs=\dvs'\ \Dec{\Type{\mu}{\C}}{\x}{\Block{\dvs\ \dvs''}{\val}}\ \decs'$. From $\IsWellTyped{{\e}:\T}$ and \refToLemma{subCtx} we get that, for some $\T_b$, $\LentLocked''$, $\xs''$, and $\StronglyLocked''$,
\begin{itemize}
  \item [$(\ast)$]$\AuxTypeCheck{\TypeEnv{\ctx}}{\LentLocked''}{\MutGroup''}{\StronglyLocked''}{\Block{\decs}{\val'}}{\T_b}$
\end{itemize}
Let $\Gamma=\TypeEnv{\dvs'},\TypeEnv{\decs'}\TypeDec{\Type{\mu}{\C}}{x}$. 
From $(\ast)$ and  \refToLemma{inversionBlock} for some $\T'$, $\mu'$, $\LentLocked'$, $\xs'$, and $\StronglyLocked'$
\begin{enumerate} [(A)]
\item $\AuxDecsOK{{\SubstFun{\TypeEnv{\ctx}}{\Gamma}}}{\LentLocked'}{\MutGroup'}{\StronglyLocked'}{\dvs'\ \decs'}$,
\item $\AuxDecsOK{{\SubstFun{\TypeEnv{\ctx}}{\Gamma}}}{\LentLocked'}{\MutGroup'}{\StronglyLocked'}{{\Block{\dvs\ \dvs''}{\val}}}$, i.e.
\begin{enumerate}[i.]
\item $\AuxTypeCheck{{\SubstFun{\TypeEnv{\ctx}}{\Gamma}}}{\LentLocked'}{\MutGroup'}{\StronglyLocked'}{{\Block{\dvs\ \dvs''}{\val}}}{{\SubstFun{\TypeEnv{\ctx}}{\Gamma}}(\x)}$ if $\mu_x\neq\lent$
\item $\AuxTypeCheck{{\SubstFun{\TypeEnv{\ctx}}{\Gamma}}}{\LentLocked_x}{\MutGroup_x}{\StronglyLocked'}{{\Block{\dvs\ \dvs''}{\val}}}{{\SubstFun{\TypeEnv{\ctx}}{\Gamma}}(\x)}$ if $\mu_x=\lent$ and $\LentLocked'\ \xs'=\LentLocked_x\ \xs_x$ where $\x\in\xs_x$
\end{enumerate}
\item $\AuxTypeCheck{\SubstFun{\TypeEnv{\ctx}}{\Gamma}}{\LentLocked'}{\MutGroup'}{\StronglyLocked'}{\val'}{\T'}$ 
\item $\LessEq{(\LentLocked''\ \xs''){\setminus}\dom{\Gamma}}{\LentLocked'\ \xs'}$
%\item  $\StronglyLocked''{\setminus}\dom{\Gamma}\subseteq\StronglyLocked'$.
\end{enumerate}
Let $\Gamma'=\TypeEnv{\dvs\,\dvs'}$. From (B), and \refToLemma{inversionBlock}, for some $\T''$, $\LentLocked$, and $\StronglyLocked$
we have that
\begin{enumerate} [(a)]
\item $\DecsOK{{\SubstFun{\TypeEnv{\ctx}}{\SubstFun{\Gamma}{\Gamma'}}}}{\LentLocked}{\StronglyLocked}{\dvs}$,
\item $\DecsOK{{\SubstFun{\TypeEnv{\ctx}}{\SubstFun{\Gamma}{\Gamma'}}}}{\LentLocked}{\StronglyLocked}{\dvs''}$, and
\item $\TypeCheck{{\SubstFun{\TypeEnv{\ctx}}{\SubstFun{\Gamma}{\Gamma'}}}}{\LentLocked}{\StronglyLocked}{\val}{\Type{\mu}{\C}}
$ with $\mu\leq\mu_x$
\end{enumerate}
From (4) and the fact that we can assume that $x\not\in\dom{\dvs}$, we derive
$\SubstFun{\Gamma}{\TypeEnv{\dvs\,\dvs''}}=\SubstFun{\Gamma,\TypeEnv{\dvs}}{\TypeEnv{\dvs''}}$.
Therefore, from (b), (c), and rule \rn{T-block}, followed by rule \rn{T-sub} (in case $\mu<\mu_x$) we derive
\begin{itemize}
\item [(C1)] $\TypeCheck{{\SubstFun{\TypeEnv{\ctx}}{\Gamma,\TypeEnv{\dvs}}}}{\LentLocked'}{\StronglyLocked'}{{\Block{\dvs''}{\val}}}{\Type{\mu_x}{\C}}
$
\end{itemize}
From (a), (5), and \refToLemma{weakening}, we have
\begin{itemize}
\item [(a1)] $\DecsOK{{\SubstFun{\TypeEnv{\ctx}}{{\Gamma,\TypeEnv{\dvs}}}}}{\LentLocked'}{\StronglyLocked'}{\dvs}$
\end{itemize}
From (4), {the fact that, for well-formedness of
declarations $x\not\in\dom{\dvs}$}, and \refToLemma{weakening}, we derive
\begin{itemize}
\item [(A1)] $\DecsOK{{\SubstFun{\TypeEnv{\ctx}}{\Gamma,\TypeEnv{\dvs}}}}{\LentLocked'}{\StronglyLocked'}{\dvs'}$,
\item [(B1)] $\DecsOK{{\SubstFun{\TypeEnv{\ctx}}{\Gamma,\TypeEnv{\dvs}}}}{\LentLocked'}{\StronglyLocked'}{\decs'}$, and
\item [(D1)] $\TypeCheck{\SubstFun{\TypeEnv{\ctx}}{\Gamma,\TypeEnv{\dvs}}}{\LentLocked'}{\StronglyLocked'}{\val'}{\T'}$
\end{itemize}
Therefore, from (A1), (B1), (C1), (D1), (a1), (E) and (F), applying rule \rn{T-lock}, we derive
\begin{center}
$\TypeCheck{{\TypeEnv{\ctx}}}{\LentLocked''}{\StronglyLocked''}{\Block{\dvs'\ \dvs\ \Dec{\Type{\mu}{\C}}{\x}{\Block{\decs}{\val}}\ \decs'}{\val'}}{\T_b}$.
\end{center}
From \refToLemma{subCtx}.2, we derive $\TypeCheckGround{\e'}{\T}$.
\qed

\medskip
\noindent
{\bf Proof of \refToLemma{decomposition}}.
By induction on $\e$. \\
If $\e=\x$ or $\e=\ConstrCall{\C}{{\xs}}$, then $\e$ is a well-formed value.\\
If $\e$ is  $\FieldAccess{{\val}}{\f}$, or $\MethCall{{\val}}{\m}{{\vals}}$,
or $\FieldAssign{{\val}}{\f}{\val}$, then with $\ctx=\emptyctx$ and $\preRedex=\e$
we have that $\e=\Ctx{\preRedex}$.\\
If $\e=\ConstrCall{\C}{{\vals}}$ and
$\exists\val\in\{\vals\}\ \notRef{\val}$, then applying congruence \rn{New} and \rn{Body} we get
$\congruence{\e}{\Block{\Dec{\T_1}{y_1}{\val_1}\cdots\Dec{\T_n}{y_n}{\val_n}}{\ConstrCall{\C}{{\xs}}}}$.
Therefore, either for all $i$, $1\leq i\leq n$, $\WFdv{\Dec{\T}{\y_i}{\val_i}}$, or there is
an $i$, $1\leq i\leq n$,  such that $\not\WFdv{\Dec{\T}{\y_i}{\val_i}}$. In the first case the lemma holds, and in the
second with $\ctx=\emptyctx$ and $\preRedex=\Block{\Dec{\T_1}{y_1}{\val_1}\cdots\Dec{\T_n}{y_n}{\val_n}}{\ConstrCall{\C}{{\xs}}}$
we have that $\e=\Ctx{\preRedex}$. \\
Let us now consider blocks.\\
If $\e=\Block{\dvs}{\val}$, since $\Block{\dvs}{\val}$ is a value,
from Proposition \ref{lemma:congruenceValue}, either $\congruence{\Block{\dvs}{\val}}{\x}$ or 
$\congruence{\Block{\dvs}{\val}}{\stVal}$ for some $\WFrv{\stVal}$. In the first case, 
and if $\stVal=\ConstrCall{\C}{{\xs}}$, then the lemma holds. 
In the second, let $\stVal=\Block{\dvs'}{\cOrx}$; either for all $\dv\in\dvs'$
we have that $\WFdv{\dv}$, or there is $\dv'\in\dvs'$ such that $\not\WFdv{\dv'}$, in which
case with $\ctx=\emptyctx$ and $\preRedex=\Block{\dvs'}{\cOrx}$ we have $\congruence{\e}{\Ctx{\preRedex}}$. \\
If $\e=\Block{\dvs\ \Dec{\T}{\x}{\e'}\ \decs}{\val}$ with $\WFdv{\dvs}$, we have that 
either $\e'=\val'$ and $\not\WFdv{\Dec{\T}{\x}{\val'}}$ or $\e'$ is not a value.\\
In the first case, define $\ctx=\emptyctx$ and $\preRedex=\e$, we derive that $\e=\Ctx{\preRedex}$.
In the second,
by induction hypothesis on $\e'$, we have that
$\e'=\CtxP{\preRedex}$ for some $\ctxP$ and $\rho$. Define $\ctx=\Block{\dvs\ \Dec{\T}{\x}{\ctxP}\ \decs}{\val}$.
We get that $\e=\Ctx{\rho}$.\qed